\documentclass[12pt]{article}
\usepackage{graphicx,amssymb,amsmath,empheq}
\usepackage{authblk}
\usepackage{amsthm}
\usepackage{amsfonts}
\usepackage{amssymb}
\usepackage{dsfont}
\usepackage{bm}
\usepackage{empheq}
\usepackage{subfigure}
\usepackage{braket}
\usepackage{siunitx}
\usepackage{amsmath,amssymb,amsfonts,stmaryrd}
\usepackage{hyperref}
\hypersetup{colorlinks = true, linkcolor=purple, citecolor=blue, urlcolor=blue}
\usepackage{url}
\usepackage{dsfont}
\usepackage{comment}
\usepackage{svg}
\usepackage{graphicx}
\usepackage{physics}
\usepackage{makecell}
\usepackage{rotating}
\usepackage{multirow} 
\usepackage[normalem]{ulem}

\theoremstyle{plain}

\usepackage{tikz}
\usepackage{tikz-cd}
\usetikzlibrary{arrows}
\usetikzlibrary{intersections}
\usetikzlibrary{shapes.geometric}
\usetikzlibrary{decorations.pathmorphing, patterns,shapes}
\usetikzlibrary{decorations.markings}

\tikzset{
  mid arrow/.style={postaction={decorate,decoration={
        markings,
        mark=at position .575 with {\arrow[#1]{stealth}}
      }}},
  near arrow/.style={postaction={decorate,decoration={
        markings,
        mark=at position .275 with {\arrow[#1]{stealth}}
      }}},
   far arrow/.style={postaction={decorate,decoration={
        markings,
        mark=at position .800 with {\arrow[#1]{stealth}}
      }}},
}

\usepackage[nosort]{cite}

\newtheorem{theorem}{Theorem}
\newtheorem*{theorem*}{Theorem}

\newtheorem{lemma}{Lemma}
\newtheorem{definition}{Definition}
\newtheorem{corollary}{Corollary}
\newtheorem*{corollary*}{Corollary}

\newcommand{\nn}{\nonumber}

\newcommand{\dkap}{\delta\kern-1.25pt\varkappa}

\newcommand{\CSS}{\operatorname{CSS}}

\newcommand{\dnorm}[1]{\left\lVert #1 \right\rVert_{\Diamond}}
\newcommand{\infnorm}[1]{\left\lVert #1 \right\rVert_{\infty}}
\newcommand{\infnormsm}[1]{\lVert#1 \rVert_{\infty}}
\newcommand{\onenorm}[1]{\left\lVert #1 \right\rVert_{1}}
\newcommand{\onenormsm}[1]{\lVert #1 \rVert_{1}}

\newcommand{\calE}{\mathcal{E}}

\newcommand{\calH}{\mathcal{H}}

\newcommand{\calM}{\mathcal{M}}

\newcommand{\calP}{\mathcal{P}}

\newcommand{\calS}{\mathcal{S}}

\newcommand{\calZ}{\mathcal{Z}}

\newcommand{\bbZ}{\mathbb{Z}}

\newcommand{\rmI}{\mathrm{I}}
\newcommand{\rmII}{\mathrm{II}}
\newcommand{\rmIII}{\mathrm{III}}
\newcommand{\rmIV}{\mathrm{IV}}

\newcommand{\rmN}{\mathrm{N}}

\newcommand{\bbrakket}[2]{\mbox{$ \langle\!\langle #1 | #2 \rangle\!\rangle $}}

\newcommand{\kket}[1]{\mbox{$| #1 \rangle\!\rangle$}}
\newcommand{\bbra}[1]{\mbox{$\langle\!\langle #1 |$}}

\newcommand{\Wg}{\operatorname{Wg}}

\newcommand{\sym}{\operatorname{sym}}
\newcommand{\EPR}{\operatorname{EPR}}
\newcommand{\C}{\operatorname{C}}
\renewcommand{\H}{\operatorname{H}}
\renewcommand{\epsilon}{\varepsilon}
\newcommand{\id}{\operatorname{id}}
\newcommand{\poly}{\operatorname{poly}}

\newcommand{\boldu}{\boldsymbol{u}}
\newcommand{\boldv}{\boldsymbol{v}}

\newcommand{\boldx}{\boldsymbol{x}}
\newcommand{\boldy}{\boldsymbol{y}}

\newcommand{\boldone}{\boldsymbol{1}}
\newcommand{\boldzero}{\boldsymbol{0}}

\newcommand{\Ens}{\mathbb{E}}

\newcommand{\sgn}{\operatorname{sgn}}

\makeatletter

\newcommand*{\wideboxed}[1]{\setlength{\fboxsep}{1ex}%
  \fbox{\m@th$\displaystyle#1$}}
\makeatother

\title{
Designs from magic-augmented Clifford circuits
}

\author[1]{Yuzhen Zhang\thanks{yuzhenzhang@ucsb.edu}}
\author[1]{Sagar Vijay\thanks{sagarvijay@ucsb.edu}}
\author[2]{Yingfei Gu\thanks{guyingfei@tsinghua.edu.cn}}
\author[3]{Yimu Bao\thanks{baoyimu@gmail.com}}

\affil[1]{\normalsize\it Department of Physics, University of California, Santa Barbara, CA 93106, USA}
\affil[2]{\normalsize\it Institute for Advanced Study, Tsinghua University, Beijing, 100084, China}
\affil[3]{\normalsize\it Kavli Institute for Theoretical Physics, Santa Barbara, CA 93106, USA}

\date{\today}

\begin{document}

\maketitle

\begin{abstract}
We introduce magic-augmented Clifford circuits -- architectures in which Clifford circuits are preceded and/or followed by constant-depth circuits of non-Clifford (``magic") gates -- as a resource-efficient way to realize approximate $k$-designs, with reduced circuit depth and usage of magic. We prove that shallow Clifford circuits, when augmented with constant-depth circuits of magic gates, can generate approximate unitary and state $k$-designs with $\epsilon$ relative error. The total circuit depth for these constructions on $N$ qubits is $O(\log (N/\epsilon)) +2^{O(k\log k)}$ in one dimension and $O(\log\log(N/\epsilon))+2^{O(k\log k)}$ in all-to-all circuits using ancillas, which improves upon previous results for small $k \geq 4$. Furthermore, our construction of relative-error state $k$-designs only involves states with strictly local magic. The required number of magic gates is parametrically reduced when considering $k$-designs with bounded additive error. As an example, we show that shallow Clifford circuits followed by $O(k^2)$ single-qubit magic gates, independent of system size, can generate an additive-error state $k$-design. We develop a classical statistical mechanics description of our random circuit architectures, which provides a quantitative understanding of the required depth and number of magic gates for additive-error state $k$-designs. We also prove no-go theorems for various architectures to generate designs with bounded relative error.
\end{abstract}

\maketitle

\newpage
\tableofcontents

\section{Introduction}
Random unitaries have broad applications in quantum information science, condensed matter physics, and high-energy physics. 
On the theoretical side, random unitary circuits provide a valuable toolset for investigating complex quantum many-body phenomena, such as thermalization and information scrambling in quantum dynamics~\cite{Hosur:2015ylk,Roberts:2016hpo,fisher2023random} and the physics of black holes~\cite{Page1993Average,Hayden:2007cs,Sekino2008Fast,Yoshida:2017non}.
On the practical side, random unitary operations provide a resource for benchmarking quantum devices~\cite{emerson2005scalable,Magesan:2012mfg}, performing efficient measurements~\cite{Elben2023Randomized} and demonstrating quantum advantage~\cite{arute2019quantum}. 
However, generating random unitaries drawn from the Haar measure is a challenging task, typically requiring resources that scale exponentially with system size~\cite{Knill:1995kz}.
This challenge has motivated significant interest in designing quantum circuits that can efficiently generate approximate random unitaries.

A useful concept for characterizing the degree to which a circuit functions as a random unitary generator is the {\it approximate unitary $k$-design}~\cite{Emerson2003Pseudo,Gross2007Designs}. 
An ensemble of unitaries, denoted by $\mathcal{E}_u$, is said to form an approximate $k$-design if it is indistinguishable from the Haar ensemble $\H$ by quantum experiments making at most $k$ queries to the unitary.
More precisely, for the unitary ensemble, one can consider two types of error constraints controlling the distinguishability in two classes of quantum experiments.
\begin{itemize}
    \item Additive error. A unitary ensemble $\calE_u$ forms a unitary $k$-design with $\epsilon$ additive error if
    \begin{equation}
        \lVert\Phi^{(k)}_{\calE_u} - \Phi^{(k)}_{\H}\rVert_\Diamond \leq \epsilon, \label{eq:def_unitary_additive}
    \end{equation}
    where $\lVert\cdot\rVert_\Diamond$ is the diamond norm~\cite{Aharonov:1998zf}, and $\Phi_{\calE_u}^{(k)}$ is the $k$-fold twirling channel defined as follows 
    \begin{equation}
        \Phi_{\calE_u}^{(k)}(\cdot) := \Ens_{U\sim \calE_u}  U^{\otimes k} (\cdot) U^{\dagger \otimes k} 
    \end{equation}
    and similarly for the Haar ensemble. A bounded additive error ensures that the ensemble is indistinguishable from the Haar ensemble in experiments that make $k$ {\it parallel} queries.
    \item Relative error (or multiplicative error)~\cite{Brandao:2012zoj}. 
    Relative error is a stronger notion than additive error.
    A unitary ensemble $\calE_u$ forms a unitary $k$-design with $\epsilon$ relative error if
    \begin{equation}
         (1-\epsilon)\Phi^{(k)}_{\H}\preceq \Phi^{(k)}_{\calE_u}\preceq(1+\epsilon)\Phi^{(k)}_{\H},\label{eq:def_unitary_relative}
    \end{equation}
    where $\Phi \preceq \Phi'$ means that $\Phi'-\Phi$ is a completely-positive map. A bounded relative error ensures that the expectation values of any positive operator in the output states of the two channels differ by at most a multiplicative factor~\cite{Brandao:2012zoj}. Furthermore, it ensures indistinguishability in any quantum experiment involving up to $k$ queries, including {\it adaptive} strategies~\cite{Schuster:2024ajb}.
\end{itemize}

A closely related notion is that of an {\it approximate state $k$-design}, which characterizes how well an ensemble of pure quantum states approximates the statistical properties of Haar random states. Specifically, an ensemble $\mathcal{E}_s$ of pure states is said to form an approximate state $k$-design if its $k$-th moment,
\begin{equation}
    \rho_{\mathcal{E}_s}^{(k)} := \mathbb{E}_{\psi \sim \mathcal{E}_s} |\psi\rangle \langle \psi|^{\otimes k},
\end{equation}
approximates the corresponding moment of the Haar ensemble. Again, the approximation can be quantified either in terms of additive error or relative (multiplicative) error: 
    \begin{align}
        \text{additive error} & \qquad \lVert \rho^{(k)}_{\calE_s} - \rho^{(k)}_{\H} \rVert_1 \leq \epsilon, \label{eq:def_state_additive}\\
        \text{relative error} & \qquad (1-\epsilon)\rho^{(k)}_{\H}\preceq \rho^{(k)}_{\calE_s}\preceq(1+\epsilon)\rho^{(k)}_{\H}, \label{eq:def_state_relative}
    \end{align}
    where, for Hermitian operators, $\rho_1 \preceq \rho_2$ denotes that $\rho_2 - \rho_1$ is positive semidefinite. 
Unlike the case of unitary designs—where only the relative error guarantees operational indistinguishability in adaptive experiments—the additive error for state designs already suffices to ensure indistinguishability even under adaptive protocols (see Appendix~\ref{app:operation_state_additive_error} for further discussion).
Similar to the unitary case, a bounded relative error further ensures that any positive operator in the two states acquires expectation values differ by at most a multiplicative factor.
We further remark that approximate state $k$-designs constitute a weaker notion than unitary $k$-designs and typically require fewer resources to construct. In particular, a state $k$-design can be constructed by applying random unitaries drawn from a unitary $k$-design to a fixed reference state, such as a product state.

Substantial progress has been made in constructing quantum circuits that generate approximate unitary $k$-designs. 
In a seminal work, Ref.~\cite{Brandao:2012zoj} considered unitary circuits composed of local two-qubit Haar random unitary gates and showed that the circuit over $N$ qubits with depth $O(k^{11}N)$ forms an approximate unitary $k$-design up to bounded relative error.
The required circuit depth was later improved in a series of works~\cite{nakata2016efficient,Haferkamp:2022hpm,Hunter-Jones:2019lps,Haferkamp2021PRA,haah2024efficient,chen2024efficient,chen2024efficient2,metger2024simple,Chen:2024wzq} to $O(kN)$ (linear in $k$)~\cite{Chen:2024wzq}.
The design construction is further simplified in a major recent breakthrough~\cite{Schuster:2024ajb,LaRacuente:2024rnh}, which showed that one-dimensional shallow quantum circuits with depth $O(k\log N)$ can achieve bounded relative error.

While the newer result significantly simplifies the theoretical requirements for constructing $k$-designs—particularly in terms of circuit depth—it still relies on the use of generic two-qubit gates, which remain challenging to implement on realistic quantum hardware. 
To this end, one may consider circuit architectures composed primarily of Clifford gates, with only a limited number of non-Clifford (``magic") gate insertions, which can reproduce various aspects of Haar random unitaries~\cite{Zhou_2020,True:2022ypf,Szombathy:2024tow,Magni:2025xbm,Turkeshi:2024pnj}. 
Clifford gates are often regarded as ``free'' operations, as they can be efficiently simulated on classical computers and are naturally compatible with fault-tolerant implementation in stabilizer codes~\cite{eastin2009restrictions,Bravyi:2012rnv}.
Along this line, Ref.~\cite{Haferkamp:2020qel} investigates circuits composed of global random Clifford unitaries interleaved with single-qubit magic gates.
They rigorously show that $O(k^4)$ magic gates are sufficient to generate a unitary $k$-design with bounded additive error, and that $O(kN)$ magic gates suffice to achieve bounded relative error. 
However, global random Clifford unitaries require local quantum circuits of linear depth for their implementation~\cite{Bravyi:2020xdv}, which results in a large overall circuit depth for the full protocol. This naturally raises the question of whether unitary $k$-designs can be realized by shallow circuits that consist predominantly of Clifford gates. 

In this work, we consider \emph{magic-augmented Clifford circuits}, which consist of random Clifford circuits preceded and/or followed by generic, constant-depth quantum circuits containing magic gates. 
We explore the possibility for the output state ensemble to form an approximate state $k$-design or the circuit ensemble to form an approximate unitary $k$-design with both bounded relative and additive errors\footnote{In this work, a ``constant" circuit depth will refer to one which is independent of the system size, but may depend on $k$.}.
A key feature of these circuits is that the magic gates and Clifford gates are not interspersed.  
Since the random Clifford gates are relatively easy to generate~\cite{Bravyi:2020xdv}, we are thus able to show that these circuit architectures can realize $k$-designs at an improved circuit depth.

More explicitly, we establish that the following circuit architectures -- many of which involve low-depth Clifford circuits augmented by constant-depth magic -- can form approximate state and unitary $k$-designs. Our results, along with a comparison to previous works, is also summarized in Tables~\ref{tab:main_results_rel_err} and~\ref{tab:main_results_add_err}). 

\begin{itemize}
\item[]{\bf Relative-error state and unitary designs:}  
The state ensemble generated by two-layer Clifford circuits operating on qubit blocks of size $O(\log(N/\epsilon)+k^2)$ followed by a product of unitary gates drawn from exact $k$-designs acting on clusters of qubits of size $O(k\log k)$ forms a state $k$-design with $\epsilon$ {relative} error for $k < O(\sqrt{N})$ (Sec.~\ref{sec:rel_error_state}).  

Furthermore, the circuit ensemble of two-layer Clifford circuits operating on qubit blocks of size $O(\log(N/\epsilon)+k^2)$, sandwiched by local gates drawn from exact $k$-designs acting on clusters of qubits of size $O(k\log k)$ forms a unitary $k$-design with $\epsilon$ relative error for $k < O(\sqrt{N})$ (Sec.~\ref{sec:rel_error_unitary}).

\item[]{\bf Additive-error state and unitary designs:} The state ensemble generated by Clifford circuits of depth $O(\log(N/\epsilon)+k^2)$ followed by $O(k^{2}+\log(1/\epsilon))$ single-qubit magic gates can form a state $k$-design with $\epsilon$ {additive} error for $k < O(\sqrt{N})$ (Sec.~\ref{sec:state_additive})\footnote{During the completion of our work, we became aware of Ref.~\cite{Leone:2025aes} on arXiv. This paper considered doped Clifford circuits consisting of global Clifford gates interspersed with single-qubit magic gates. Although the required number of magic gates has the same scaling, their circuit depth is linear in $N$, much larger than that in our construction.}.

Furthermore, the circuit ensemble of a unitary gate drawn from $k$-design acting on $O(k^{3}+\log(1/\epsilon))$ qubits, followed by a global random Clifford unitary, forms a unitary $k$-design with $\epsilon$ additive error for $k < N+1$ (Sec.~\ref{sec:unitary_additive}). 

\end{itemize}

A key intermediate result that allows us to prove these results is the ``uniformity condition", which we introduce in Sec.~\ref{sec:uniformity}. 
Roughly speaking, it measures the deviation of an ensemble of Clifford quantum circuits from a global, random Clifford unitary gate. 
This condition can be approximately satisfied by low-depth Clifford circuits and guarantees that such circuits, when augmented with magic gates, can form approximate $k$-designs under various measures.

\begin{table*}
\centering
\begin{tabular}{|c|c|c|}
\hline\hline
& \makecell{State $k$-design \\ w/ $\epsilon$ relative error} & \makecell{Unitary $k$-design \\ w/ $\epsilon$ relative error} \\
\hline
\makecell{Haar \\ random  \\
Ref.~\cite{Schuster:2024ajb}} & \multicolumn{2}{c|}{$d_{\text{1D}}= O(\log\frac{N}{\epsilon}\cdot k\poly\log k)$} \\
\hline
\makecell{Doped \\ Clifford \\
Ref.~\cite{Haferkamp:2020qel}} 
& \multicolumn{2}{c|}{\makecell{$d_{\text{1D}}= O(N^2k+N\log\frac{1}{\epsilon})$ \\\# of magic gates$= O(Nk+\log\frac{1}{\epsilon})$}} \\
\hline
\multirow{2}{*}{\makecell{Magic \\ augmented \\ Clifford \\(this work)}} 
& \multicolumn{2}{c|}{\makecell{$d_{\text{1D}} =  O(\log \frac{N}{\epsilon}) + 2^{O(k\log k)}$ 
\\ $d_{\text{all-to-all}}=O(\log\log\frac{N}{\epsilon}) + 2^{O(k\log k)}$
\\ magic depth$ = 2^{O(k\log k)}$}} \\
\cline{2-3}
& \makecell{strictly local magic\\ \includegraphics[width=6cm]{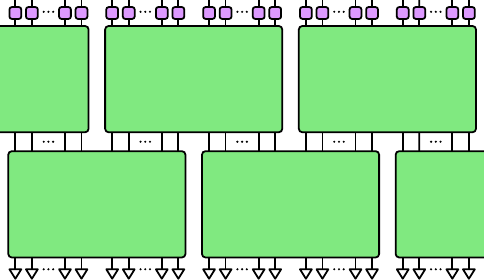}} 
& \makecell{\\ \includegraphics[width=6cm]{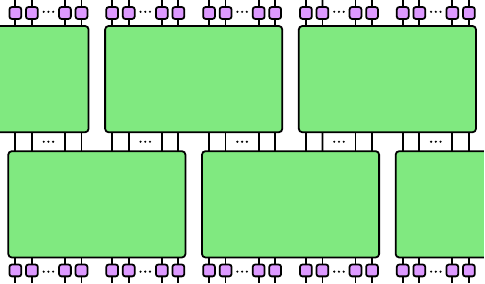}} \\
\hline\hline
\end{tabular}
\caption{{\bf Relative-error designs:} Circuit parameters for the constructions of approximate $k$-designs over $N$ qubits with $\epsilon$ relative error for $k\geq 4$. $d_{\text{1D}}$ is the depth of 1D circuits with geometrically local gates, $d_{\text{all-to-all}}$ is the depth of circuits with all-to-all connectivity and ancillas. Here, magic depth refers to the number of instances in time where magic gates are applied in the circuit. }
\label{tab:main_results_rel_err}
\end{table*}

\begin{table*}
\centering
\begin{tabular}{|c|c|c|}
\hline\hline
& \makecell{State $k$-design \\ w/ $\epsilon$ additive error} & \makecell{Unitary $k$-design \\ w/ $\epsilon$ additive error} \\
\hline
\makecell{Haar \\ random \\
Ref.~\cite{Schuster:2024ajb}} & \multicolumn{2}{c|}{$d_{\text{1D}} = O(\log\frac{N}{\epsilon}\cdot k\poly\log k)$} \\
\hline
\makecell{Doped \\ Clifford \\
Ref.~\cite{Haferkamp:2020qel,Leone:2025aes}} &\makecell{$d_{\text{1D}} =O(N)$~\cite{fux2024disentangling}\\\# of magic gates$= O(k^2+\log\frac{1}{\epsilon})$}  & \makecell{$d_{\text{1D}} = O(N(k^4+k\log\frac{1}{\epsilon}))$\\\# of magic gates$= O(k^4+k\log\frac{1}{\epsilon})$} \\
\hline
\makecell{Magic \\ augmented \\ Clifford \\
(this work)} & \makecell{$d_{\text{1D}}= O(\log \frac{N}{\epsilon}+k^2)$
\\$d_{\text{all-to-all}}=O(\log \log\frac{N}{\epsilon}+\log k)$
\\\# of magic gates$= O(k^2+\log\frac{1}{\epsilon})$\\magic depth $ = 1$ \\ \includegraphics[width=6cm]{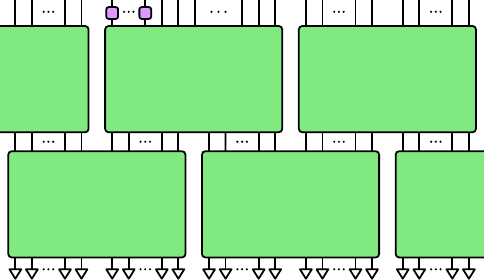}} 
& \makecell{$d_{\text{1D}}=O(N+k\poly\log\frac{k}{\epsilon})$ \\
support of magic gates$=O(k^3+\log\frac{1}{\epsilon})$ 
\\\# of magic gates$=O(k^4\poly\log\frac{k}{\epsilon})$  \vspace{0.2cm}\\ \includegraphics[width=4cm]{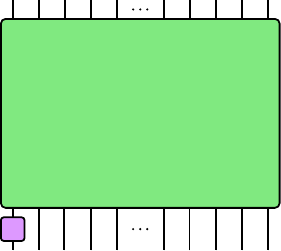}} \\
\hline\hline
\end{tabular}
\caption{{\bf Additive-error designs:} Circuit parameters for the constructions of approximate $k$-designs over $N$ qubits with $\epsilon$ additive error for $k\geq 4$.}
\label{tab:main_results_add_err}
\end{table*}

Our results have a few implications. 
\begin{enumerate}
\item[] {\bf Improved circuit depth for unitary designs:} Our construction of the unitary $k$-design with bounded relative error only requires a 1D local circuit depth of $O(\log(N/\epsilon))+2^{O(k\log k)}$. The construction has a decoupled scaling of depth with $N$ and $k$, which was an open problem\footnote{We acknowledge Thomas Schuster for bringing this point to our attention.
}. The scaling improves the known result of $O(\log(N/\epsilon)\cdot k\poly\log k)$ in Ref.~\cite{Schuster:2024ajb} for small $k\geq 4$. Using circuits with all-to-all connectivity and $O(N(\log (N/\epsilon)+k^2))$ ancillas, our construction requires only depth $O(\log\log(N/\epsilon))+2^{O(k\log k)}$.

\item[] {\bf State designs from local magic:} Our result suggests that an ensemble of states with strictly local magic~\cite{Amy_2013,Zhang:2024jlz,Korbany:2025noe,Wei:2025irp,Andreadakis:2025mfw,Parham:2025sxj} can form an approximate state $k$-design (with bounded relative error). 
Here, a state is said to contain only local magic if one can convert it to a stabilizer state using a constant-depth local unitary.
Furthermore, we provide a shallow-circuit construction for an additive-error state $k$-design with a system-size independent number, $O(k^2+\log(1/\epsilon))$, of single-qubit magic gates. 
This construction with a constant number of magic gates (constant ``magic") can be efficiently generated on a classical computer~\cite{Bravyi_2016} and is easy to learn~\cite{Leone:2023avk,Grewal:2023hzn}.

\item[] {\bf Additive-error stabilizer state designs at low depth:} We demonstrate that Clifford circuits can generate stabilizer state $k$-designs with $\epsilon$ additive error in depth $O(\log(N/\epsilon) + k^{2})$ using 1D circuits and in depth $O(\log\log(N/\epsilon)+\log k)$ using all-to-all circuits with ancillas.

\item[] {\bf Representational power of Clifford-augmented states:} Consider an ensemble of Clifford-augmented states generated by Clifford gates acting on matrix product states with bond dimension $D$ or the output states of local unitary circuits of a finite depth $t$ in any finite dimension. We show that these ensembles cannot form a relative-error state $k$-design with $D = \tilde{o}((N/\epsilon)^{1/4})$ or for any finite $t$.  %
\end{enumerate}

\subsection{Main results}
In this subsection, we introduce the specific circuit architectures used to construct the approximate designs presented in this work and provide a comprehensive summary of our main results. Also see Tables~\ref{tab:main_results_rel_err} and~\ref{tab:main_results_add_err} for comparison with existing literature.

\subsubsection{State designs with bounded relative error}

In Sec.~\ref{sec:rel_error_state}, we provide a construction for a state designs with bounded relative error (defined in Eq.~\eqref{eq:def_state_relative}) using low-depth circuits.
\begin{theorem*}[Relative-error state designs from low-depth Clifford unitaries followed by constant depth magic gates]
Consider a state ensemble $\calE = \{UV\ket{0}^{\otimes N}\}$ generated by the circuit architecture below
\begin{equation*}
\includegraphics[width=10.8cm]{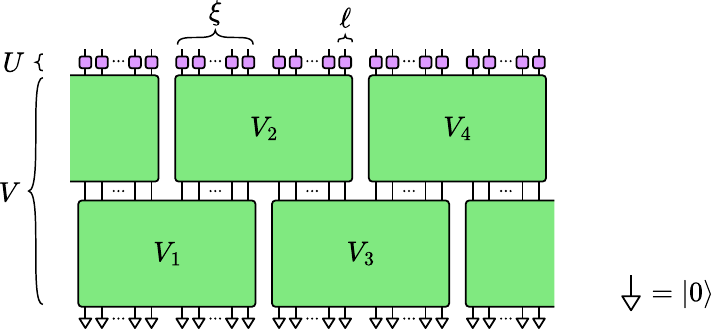}
\end{equation*}
Here, $U$ is a product of random unitaries drawn from exact unitary $k$-designs acting on disjoint qubit clusters of size $\ell$. 
The two-layer Clifford unitary $V$ consists of Clifford gates acting on $2\xi$ qubits. 
Each Clifford gate is drawn independently from an exact Clifford $k$-design.  
In an $N$-qubit system with $N=\omega(k^2+\log(1/\epsilon))$, the ensemble $\calE$ can form a state $k$-design with $\epsilon$ relative error for $\xi=O(\log(N/\epsilon)+k^2)$ and $\ell=O(k\log k)$.
\end{theorem*}

In Sec.~\ref{sec:rel_error_nogo}, we prove no-go theorems for state $k$-design with bounded relative error. 
In particular, we show that a state ensemble generated by Clifford gates acting on matrix product states with a low bond dimension or general short-range entangled states prepared by finite-depth local unitaries in any dimension cannot form a state $k$-design with bounded relative error for $k \geq 4$.
The key observation is that the stabilizer R\'enyi entropy of these states is not close enough to that of Haar random states as required by a small relative error.
\begin{theorem*}
Consider an ensemble of states $\calE = \{V\ket{\psi}\}$, where $V$ is a Clifford unitary.
\begin{itemize}
    \item For $\ket{\psi}$ being matrix product states with bond dimension $D$, the state ensemble cannot form a state $k$-design for $k\geq 4$ with $\epsilon$ relative error for $D=\tilde{o}((N/\epsilon)^{1/4})$.
    \item For $\ket{\psi}$ being states prepared by $d$-dimensional local unitary circuits of depth $t$, the ensemble cannot form a state $k$-design with $\epsilon$ relative error for $t = \tilde{o}((\log(N/\epsilon))^{1/d})$.
\end{itemize}
\end{theorem*}

\subsubsection{Unitary designs with bounded relative error}

In Sec.~\ref{sec:rel_error_unitary}, we construct unitary designs with bounded relative error (defined in Eq.~\eqref{eq:def_unitary_relative}) using low-depth circuits.
\begin{theorem*}[Relative-error unitary designs from low-depth Clifford unitaries sandwiched by constant depth magic gates]
Consider an ensemble $\calE = \{U V U'\}$ of unitary circuits given by the architecture below
\begin{equation*}
\includegraphics[width=8.2cm]{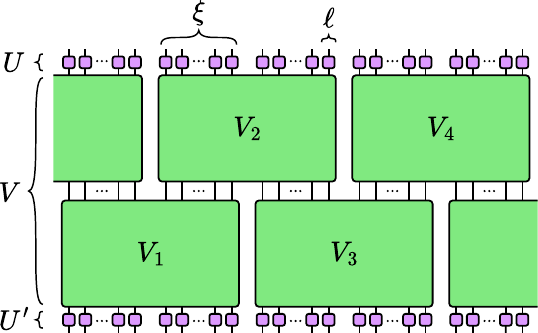}
\end{equation*}
where $U$ and $U'$ are a product of random unitaries drawn from exact $k$-design acting on disjoint qubit clusters of size $\ell$, $V$ is a two-layer Clifford unitary.
Each gate in the two-layer Clifford unitary acts on $2\xi$ qubits and is drawn independently from an exact Clifford $k$-design. 
In an $N$-qubit system with $N=\omega(k^2+\log(1/\epsilon))$, the ensemble $\calE$ can form a unitary $k$-design with $\epsilon$ relative error for $\xi = O(\log(N/\epsilon) + k^2)$ and $\ell = O(k\log k)$.
\end{theorem*}

The circuit depth of our construction follows from the results for constructing random Clifford and Haar random unitary gates~\cite{Bravyi:2020xdv,Knill:1995kz}:
\begin{corollary*}
In an $N$-qubit system with $N=\omega(k^2+\log(1/\epsilon))$, one can realize approximate unitary and state $k$-designs with $\epsilon$ relative error in depth 
\begin{itemize}
\item $O(\log (N/\epsilon)) + 2^{O(k\log k)}$ using 1D circuits.
\item $O(\log\log (N/\epsilon)) +2^{O(k\log k)}$ using all-to-all circuits with $O(N(\log (N/\epsilon)+k^2))$ ancillas.
\end{itemize}
\end{corollary*}
The depth of our construction exhibits a decoupled scaling of $k$ and $N$.
For 1D circuits with geometrically local gates, our result improves the previous result  $O(\log(N/\epsilon)\cdot k\poly\log k)$ for small $k\geq 4$ obtained in 1D shallow Haar random circuits~\cite{Schuster:2024ajb}. 
For circuits with all-to-all connectivity and ancillas, the $O(\log\log (N/\epsilon)) +2^{O(k\log k)}$ scaling was not known previously for $k\geq 4$.

Having presented our construction for unitary designs with bounded relative error, we prove a no-go theorem in Sec.~\ref{sec:rel_error_nogo}. 
We show that finite-depth local unitaries sandwiched by Clifford gates cannot form relative-error unitary $k$-designs with $k \geq 4$.
\begin{theorem*}
Let $V_1$ and $V_2$ be Clifford unitaries, and $U$ be d-dimensional local circuits with depth bounded by $t$. Then, the unitary ensemble $\{V_1 U V_2\}$ does not form a unitary $k$-design for $k\geq 4$ with $\epsilon$ relative error for $t=\tilde{o}((\log(N/\epsilon)^{1/d})$.
\end{theorem*}

\subsubsection{State designs with bounded additive error}

In Sec.~\ref{sec:state_additive}, we prove the following theorem for the state designs with bounded additive error (defined in Eq.~\eqref{eq:def_state_additive}).
\begin{theorem*}[Additive-error state designs from low-depth Clifford unitaries followed by constant magic]
Consider a state ensemble $\calE = \{U V\ket{0}^{\otimes N}\}$ generated by the circuit below
\begin{align*}
    \includegraphics[width=10.8cm]{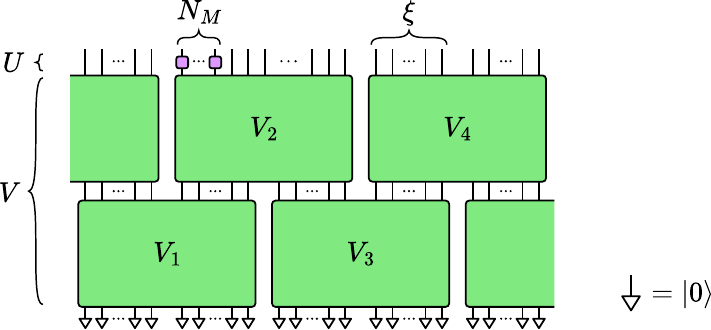}
\end{align*}
where $U$ is a product of $N_M$ single-qubit gates drawn independently from either an exact $k$-design or a gate set $\{\mathds{1},u,u^\dagger\}$ with $u$ being a magic gate, and $V$ is a two-layer Clifford unitary.
Each Clifford gate in $V$ acts on $2\xi$ qubits and is drawn from an exact Clifford $k$-design.
In an $N$-qubit system, the ensemble $\calE$ can form a state $k$-design for $k < O(\sqrt{N})$ with $\epsilon$ additive error for $\xi=O(\log(N/\epsilon)+k^2)$ and 
\begin{itemize}
    \item $N_M=O(k^2+\log(1/\epsilon))$ if the single-qubit gates are drawn from a $k$-design;
    \item $N_M=O((k^2+\log(1/\epsilon))\log^2 k)$ if the single-qubit gates are drawn from $\{\mathds{1},u,u^\dagger\}$.
\end{itemize}

\end{theorem*}

We note that $u$ in the single-qubit gate set $\{\mathds{1}, u, u^\dagger\}$ can be any magic gate, e.g. T-gate.
We also note that we can place magic gates before the low-depth circuit as shown in Appendix~\ref{sec:state_additive_initial_magic}. 
This result does not contradict our no-go theorem, which applies to the designs with bounded relative, but not additive error. 

While establishing these results, we also obtain a construction for an approximate stabilizer state design in shallow circuits. 
Specifically, an ensemble $\calE$ forms a stabilizer state $k$-design with $\epsilon$ additive error if $\onenormsm{\rho_\calE^{(k)} - \rho_{\C}^{(k)}} \leq \epsilon$, where $\rho_{\C}^{(k)} := \Ens_{\psi\sim\C} \ketbra{\psi} ^{\otimes k}$. 
\begin{theorem*}[Additive-error stabilizer state designs from low-depth Clifford circuits]
Consider a state ensemble $\calE = \{V\ket{0}^{\otimes N}\}$ generated by a two-layer Clifford unitary $V$. Each Clifford gate acts on $2\xi$ qubits and is drawn from an exact Clifford $k$-design.
In an $N$-qubit system, the ensemble can form a stabilizer state $k$-design with $\epsilon$ additive error for $\xi=O(\log(N/\epsilon)+k^2)$.
\end{theorem*}
\begin{corollary*}
In an $N$-qubit system, one can realize approximate stabilizer state $k$-designs with $\epsilon$ additive error in depth 
\begin{itemize}
\item $O(\log (N/\epsilon)+k^2)$ using 1D circuits.
\item $O(\log\log (N/\epsilon)+\log k)$ using all-to-all circuits with $O(N(\log(N/\epsilon)+k^2))$ ancillas. 
\end{itemize}
With $O(k^2+\log(1/\epsilon))$ single-qubit magic gates, one can realize a state $k$-design with $\epsilon$ additive error in the same depths.
\end{corollary*}

\subsubsection{Unitary designs with bounded additive error}

In Sec.~\ref{sec:unitary_additive}, we prove the following theorem for unitary designs with bounded additive error (defined in Eq.~\eqref{eq:def_unitary_additive}).
\begin{theorem*}[Additive-error unitary designs from circuits with magic gates over constant-number of qubits]
Consider a unitary ensemble $\calE = \{V U\}$ generated by the circuit below
\begin{equation*}
    \includegraphics[width=5.5cm]{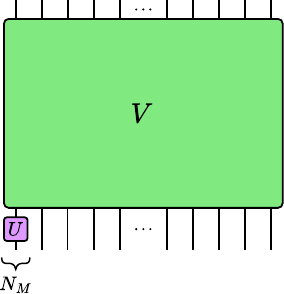}
\end{equation*}
Here, $U$ is a random unitary drawn from a unitary $k$-design with $\epsilon/2$ additive error, it acts on a subsystem $M$ of $N_M$ qubits, and $V$ is a global Clifford unitary drawn from an exact Clifford $k$-design. 
The ensemble $\calE$ can realize a unitary $k$-design for $k < N+1$ with $\epsilon$ additive error for $N_M = O(k^3+\log(1/\epsilon))$
\end{theorem*}
\begin{corollary*}
With $O(k^4\poly\log(k/\epsilon))$ two-qubit magic gates, one can realize approximate unitary $k$-designs with $\epsilon$ additive error in depth $O(N+k\poly\log (k/\epsilon))$ using 1D circuits and in depth $O(\log N+k\poly\log (k/\epsilon))$ using all-to-all circuits with $N^2$ ancillas. 
\end{corollary*}

We would like to highlight the distinct architectural choice made for constructing the additive-error unitary design. Specifically, our approach employs global random Clifford unitaries, which inherently require circuits of linear depth, in contrast to the shallow circuit architecture used in other constructions. 
For applications where minimizing the overall circuit depth—particularly with respect to the system size $N$—is the primary objective, one may instead use the aforementioned relative-error unitary design, constructed by sandwiching local $k$-design unitary gates acting on disjoint clusters of qubits. 
However, this alternative architecture incurs a substantial magic cost, requiring that the number of non-Clifford gates scales linearly in $N$. 
By contrast, our construction for an additive-error design achieves favorable circuit depth under a fixed magic budget. 
For a detailed comparison with the existing results, see Table~\ref{tab:main_results_add_err}.

\subsection{Statistical mechanics picture}
A portion of our results can be interpreted in the language of an effective classical statistical mechanics model that describes the ensemble of circuits that we consider. Similar mappings have been successfully used \cite{fisher2023random} to investigate a diverse set of quantum many-body phenomena, ranging from entanglement growth and the onset of quantum chaos \cite{Nahum:2016muy,Nahum:2017yvy,vonKeyserlingk:2017dyr,Zhou:2019pob} to the physics of random tensor networks \cite{Hayden:2016cfa,Vasseur:2018gfy}, and measurement-induced entanglement phases and transitions \cite{Skinner:2018tjl, jian2020measurement, bao2020theory, li2023entanglement, li2018quantum,Li:2021dbh}.

For a quantum circuit composed of random gates drawn from a $k$-design, the average of the $k$-fold replicated circuit can be mapped onto the statistical mechanics of ``spins" taking values in the permutation group $S_{k}$, which labels the linearly-independent set of operators that commute with the $k$-th tensor power of an arbitrary unitary operator. 
For a quantum circuit composed of random Clifford gates, the structure of the effective statistical mechanics description is richer.  
The average $k$-fold channel of a random Clifford gate maps onto a summation over stochastic Lagrangian subspaces $\Sigma_{k,k}$, which label operators in the Clifford commutant, (i.e. the set of operators that commute with the $k$-th tensor powers of Clifford unitaries, reviewed in Sec.~\ref{sec:prelim_clifford}).
This results in a statistical mechanics model consisting of ``spin" degrees of freedom, each of which labels a stochastic Lagrangian subspace. 
The $S_k$ permutation operators belong to the Clifford commutant and correspond to a subset of the stochastic Lagrangian subspaces $\Sigma_{k,k}$.

Section~\ref{sec:stat_mech_uniformity} shows that \emph{approximate uniformity} --- a key property satisfied by two-layer shallow Clifford circuits that underlies our shallow-circuit constructions of $k$-designs --- is linked to a notion of \emph{strong ordering} in the statistical mechanics description.
The strong ordering requires that any domain-wall excitation (neighboring spins representing non-identical Lagrangian subspaces) is parametrically suppressed, which is a stronger condition than simply being within the ordered phase of the statistical mechanics model. 
We show that strong ordering can happen when each gate acts on $O(\log N)$ number of qubits, corresponding to a sufficiently low temperature $O(1/\log N)$ in the stat-mech model, which prevents small domain walls from proliferating.

We note that in one dimension, strong ordering emerges when the stat-mech model exhibits usual ordering, which necessarily requires that the domain walls flipping an extensive number of spins are suppressed. 
In higher dimensions, although the stat-mech model can exhibit usual ordering when the gates in the two-layer circuit act on an $O(1)$ number of qubits, strong ordering still requires each gate to act on $O(\log N)$ qubits.

The analysis in Sec.~\ref{sec:stat_mech_uniformity} can be easily generalized to the case of two-layer Haar random circuits studied in Ref.~\cite{Schuster:2024ajb}.
There, the strong ordering that emerges when each Haar random gate acts on $O(\log N)$ qubits explains why such circuits form unitary designs with bounded relative error. 

Section~\ref{sec:stat_mech_state_additive} provides a statistical mechanics understanding for why $\tilde{O}(k^2)$\footnote{The notation $\tilde{O}$ accounts for the potential $\log^2 k$ factor that depends on the choice of magic gate set.} magic-gate insertions are sufficient to convert shallow Clifford circuits to state $k$-designs with bounded additive error.
The key observation is that bounded additive error requires a strong ordering in the stat-mech model, in which the spins are aligned spatially to the subset of the stochastic Lagrangian subspaces that corresponds to the permutation group $S_k$.  In the stat-mech model for the random Clifford circuit, inserting single-qubit magic gates can be regarded as introducing a local ``symmetry-breaking" field, which favors the stochastic Lagrangian subspaces that correspond to  $S_k$.
Each magic gate thus creates an additional energy cost for spins taking a non-permutation value; the cost is $O(1)$ for magic gates drawn from a $k$-design and is $O(1/\log^2k)$ if the gates are drawn from $\{\mathds{1},u,u^\dagger\}$ with $u$ being a generic magic gate. %
In the strongly ordered stat-mech model for Clifford circuits, there are a total of $\abs{\Sigma_{k,k}} - \abs{S_k} = 2^{O(k^2)}$ spin configurations spatially aligned to a non-permutation value, corresponding to an entropy $O(k^{2})$. 
It is therefore enough to insert $\tilde{O}(k^2)$ single-qubit magic gates -- which create an energy cost of $O(k^2)$ for these non-permutation configurations within the stat-mech description -- to suppress these states and achieve strong ordering within the permutation subspace.

Our stat-mech mapping strongly suggests that the $\tilde{O}(k^2)$ number of generic, magic gates inserted at arbitrary locations in the two-layer random Clifford circuit operating on logarithmic-sized qubit blocks is sufficient to generate an additive-error state $k$-design.  
This heuristic picture further motivates the search for other unitary circuit architectures (e.g. with local, two-qubit Clifford gates and a constant number of magic gates) which can realize additive-error designs.

\subsection{Reader's guide}
The rest of this paper is organized into several sections that can be read either sequentially or selectively, depending on reader's interest.

\begin{itemize}
    \item Section~\ref{sec:prelim} includes the preliminary background for our main results. Section~\ref{sec:prelim_design_unitary} and~\ref{sec:prelim_design_state} review the definitions of approximate unitary and state designs, respectively.
    Section~\ref{sec:prelim_clifford} reviews the algebraic theory of the Clifford commutant developed in Ref.~\cite{Gross_2021}.

    \item Section~\ref{sec:uniformity} introduces the uniformity conditions.

    \item Section~\ref{sec:rel_error} presents our main results for designs with bounded relative error. 
    We prove that our circuit architectures can generate state and unitary designs with bounded relative error in Sec.~\ref{sec:rel_error_state} and~\ref{sec:rel_error_unitary}, respectively.
    Section~\ref{sec:rel_error_nogo} proves no-go theorems for various other circuit architectures to form designs with bounded relative error.
    
    \item Section~\ref{sec:add_error} presents the results for designs with bounded additive error. 
    We prove that the circuit architectures we consider can generate state and unitary designs with bounded additive error in Sec.~\ref{sec:state_additive} and~\ref{sec:unitary_additive}, respectively.

    \item Section~\ref{sec:stat_mech} provides the physical understanding of our results based on statistical mechanics mappings.

    \item We close with discussions and outlook in Sec.~\ref{sec:discussion}.

    \item We present the results that are not our primary focus and a list of notation in the Appendices.
\end{itemize}

Suggested reading paths:
\begin{itemize}
    \item For readers interested in technical details: read Sec.~\ref{sec:uniformity},~\ref{sec:rel_error}, and~\ref{sec:add_error};
    \item For readers interested in an understanding of the uniformity condition in terms of the statistical mechanics: read Sec.~\ref{sec:uniformity} and~\ref{sec:stat_mech_uniformity};
    \item For readers interested in the statistical mechanics picture of our construction for state design with bounded additive error: read Sec.~\ref{sec:stat_mech}.
\end{itemize}

\section{Preliminaries}\label{sec:prelim}
In this section, we first review the definition of state and unitary $k$-design and their approximation up to additive and relative errors.
We then review the algebraic structure of the Clifford commutant.

\subsection{Approximate unitary designs}\label{sec:prelim_design_unitary}
We 
define the $k$-fold channel $\Phi^{(k)}\left( \cdot \right)$ for an ensemble $\calE$ of unitary gates as follows, 
\begin{equation}
\Phi_\calE^{(k)}\left( \cdot \right):=\Ens_{U\sim\calE}\ U^{\otimes k}\left( \cdot \right) U^{\dagger\otimes k}.
\end{equation}
\begin{definition}[Approximate unitary design with bounded additive error]
The ensemble $\calE$ forms an approximate unitary $k$-design with $\epsilon$ {\bf additive} error if
\begin{equation}
\lVert\Phi^{(k)}_\calE-\Phi^{(k)}_{\H}\rVert_\Diamond\leq \epsilon,
\end{equation}
where $\Phi^{(k)}_{\H}$ is the $k$-fold channel associated with the Haar random unitary ensemble.
\end{definition}
The diamond norm $\lVert\cdot\rVert_\Diamond$ for superoperators is defined as follows~\cite{Aharonov:1998zf}:
\begin{equation}
\lVert \Phi \rVert_\Diamond:=\lVert \Phi \otimes \id\rVert_{1\rightarrow 1} :=\max_{\lVert O\rVert_1\leq 1}\lVert [\Phi\otimes\id](O)\rVert_1,
\end{equation}
Here, the superoperator $\Phi$ and the identity channel $\id$ act on operators in the Hilbert space $\calH$ and a second Hilbert space $\calH'$ with the same dimension as $\calH$, respectively~\footnote{We note that $\lVert \Phi \otimes \id\rVert_{1\rightarrow 1}$ is independent of $\calH'$ as long as $\dim{\calH'}\geq \dim{\calH}$.}.

Operationally, the diamond norm bounds the additive distinguishability of two channels. For any input density matrix $\rho$ on the Hilbert space $\calH \otimes \calH'$,
\begin{equation}
 \lVert [\Phi_1\otimes \id](\rho)-[\Phi_2\otimes \id](\rho)\rVert_1 \leq \lVert\Phi_1-\Phi_2\rVert_\Diamond,
\label{def add}
\end{equation}
where $\Phi_{1,2}$ and $\id$ act on operators in $\calH$ and $\calH'$, respectively. %
The trace distance is, in turn, related to the success probability of distinguishing two density matrices, $p_{\text{succ}} = 1/2 + \lVert \rho_1 - \rho_2\rVert_1/2$, by the Holevo-Helstrom theorem~\cite{wilde2013quantum}.

\begin{definition}[Approximate unitary design with bounded relative error~\cite{Brandao:2012zoj}]
The unitary ensemble $\calE$ is an approximate unitary $k$-design with $\epsilon$ {\bf relative} error if
\begin{equation}
(1-\epsilon)\Phi^{(k)}_{\H}\preceq \Phi^{(k)}_\calE\preceq(1+\epsilon)\Phi^{(k)}_{\H},
\label{def rel}
\end{equation}
where ``$\preceq$" is the semidefinite ordering between channels: $\Phi_1\preceq\Phi_2$ iff $\Phi_2-\Phi_1$ is a completely positive map.
\end{definition}
In our detailed proofs, we always work with an equivalent formulation of the semidefinite ordering $\Phi_1\preceq\Phi_2$ in terms of the operator inequality below,
\begin{equation}
[\Phi_1\otimes \id](P_{\EPR})\preceq [\Phi_2\otimes \id](P_{\EPR}),
\end{equation}
where $P_{\EPR}$ is the density matrix of EPR states between $\calH$ and $\calH'$~\footnote{Here is a brief explanation. A completely positive channel $\Phi_2-\Phi_1$ maps positive operators to positive operators. When acting on the positive operator $P_{\EPR}$, the resulting operator $[(\Phi_2-\Phi_1)\otimes \id](P_{\EPR})\succeq 0$ is positive. To prove the other direction, the positivity of $[(\Phi_2-\Phi_1)\otimes \id](P_{\EPR})$ implies $(\Phi_2-\Phi_1)(\ketbra{\psi})$ being a positive operator $\forall \ket{\psi}$. Hence, $\Phi_2-\Phi_1$ is a semi-positive map.}.

A $k$-design with bounded relative error has the following operational significance.  First, given a state $\rho$ on $\calH\otimes\calH'$, the probabilities for measurement outcomes on the states $[\Phi_{\calE}^{(k)}\otimes \id](\rho)$ and $[\Phi_{\H}^{(k)}\otimes \id](\rho)$ only differ by a small \emph{multiplicative} factor. 
This setup describes quantum experiments where a unitary $U$ drawn from an appropriate ensemble can be queried $k$ times in parallel.  
However, a bounded relative error also implies that in a more general adaptive quantum experiment, where $k$ separate queries are made to a unitary $U$, the output states when $U$ is sampled from $\calE$ or from the Haar ensemble are only separated by $2\epsilon$ in trace-distance~\cite{Schuster:2024ajb}.

We remark that relative error is a stronger notion than additive error; an approximate $k$-design with bounded relative error is necessarily an approximate $k$-design with bounded additive error\footnote{Here, we provide an explicit derivation. Assume $\Phi_\calE^{(k)}$ is within $\epsilon$ relative error of $\Phi^{(k)}_{\H}$. We have
\begin{equation}
\begin{aligned}
\lVert \Phi_\calE^{(k)} - \Phi_{\H}^{(k)}\rVert_\Diamond&=\max_{\lVert O\rVert_1\leq 1}\lVert (\Phi_\calE - \Phi_{\H})\otimes \id (O)\rVert_1=2\max_{\lVert O\rVert_1\leq 1}\max_{0\leq M\leq 1}\Tr M[(\Phi_\calE - \Phi_{\H})\otimes \id](O)\\
&\le2\epsilon\max_{\lVert O\rVert_1\leq 1}\max_{0\leq M\leq 1}\abs{\Tr M[\Phi_{\H}\otimes \id](O)}\leq 2\epsilon \max_{\lVert O\rVert_1\leq 1}\lVert [\Phi_{\H}\otimes \id](O)\rVert_1 \\
&\le2\epsilon\max_{\lVert O\rVert_1\leq 1}\Ens_{U\sim\H}\lVert (U^{\otimes k}\otimes I)O(U^{\dagger \otimes k}\otimes I)\rVert_1=2\epsilon\max_{\lVert O\rVert_1\leq 1}\lVert O\rVert_1=2\epsilon.
\end{aligned}
\end{equation}
The last inequality follows from the convexity of the trace norm.}, though the converse is not necessarily true~\cite{Brandao:2012zoj}.

In this work, we also consider approximate Clifford $k$-designs, which are defined with respect to the distance from the $k$-fold channel of random Clifford unitaries, $\Phi^{(k)}_{\C}(\cdot) := \Ens_{V\sim\C} V^{\otimes k}(\cdot)V^{\dagger\otimes k}$.

\subsection{Approximate state designs}\label{sec:prelim_design_state}
We define the $k$-th moment of an ensemble $\calE = \{\ket{\psi}\}$ of quantum states  as follows
\begin{equation}
\rho_{\calE}^{(k)} = \Ens_{\psi \sim \calE} \ketbra{\psi}^{\otimes k}.
\end{equation}
\begin{definition}[Approximate state design with bounded additive error]
A state ensemble $\calE = \{\ket{\psi}\}$ forms an approximate state $k$-design with $\epsilon$ additive error if
\begin{equation}
\lVert\rho_\calE^{(k)}-\rho_{\H}^{(k)}\rVert_1\leq \epsilon,
\end{equation}
where $\rho_{\H}^{(k)}$ is the $k$-th moment of the Haar ensemble. 
\end{definition}

We note that the $k$-th moment of the Haar random ensemble takes the explicit form $\rho_{\H}^{(k)}=\frac{1}{d_{\sym}}P_{\sym}$, where $P_{\sym}$ is the projection onto the permutation symmetric subspace of the replicated Hilbert space $\calH^{\otimes k}$, and $d_{\sym} = \binom{d+k-1}{k}$ is the dimension of this subspace, with $d = \dim \calH$.

Operationally, an ensemble $\calE$ being a state $k$-design with $\epsilon$ additive error implies that, for any algorithm that makes $k$ queries of the state $\ket{\psi}$, the output state when $\ket{\psi}$ is sampled from $\calE$ or from an exact $k$-design is at most $\epsilon$ in trace distance (see Appendix~\ref{app:operation_state_additive_error}). 
While the trace norm is often challenging to compute, it is upper-bounded by the two-norm, which is directly related to the \emph{frame potential}
\begin{equation}
\lVert \rho_\calE^{(k)}-\rho_{\H}^{(k)}\rVert_1\leq \sqrt{d^k}\lVert \rho_\calE^{(k)}-\rho_{\H}^{(k)}\rVert_2 = \sqrt{d^k\big(F_\calE^{(k)}-F_H^{(k)}\big)}.
\end{equation}
where the frame potential for an ensemble of states $\calE$ is defined as
\begin{align}
    F_{\calE}^{(k)} = \Ens_{\psi,\phi \sim \calE} \abs{\braket{\psi}{\phi}}^{2k}.
\end{align}
For the Haar ensemble, the frame potential is given by $F_{\H}^{(k)}=1/d_{\sym}\approx k!/d^k$ for large $d$.

We note that an ensemble $\{U\}$ forming a unitary $k$-design with $\epsilon$ additive error implies that $\calE = \{U \ket{0}^{\otimes N}\}$ forms a state
$k$-design with $\epsilon$ additive error because
\begin{align}
\lVert \rho_\calE^{(k)}-\rho_{\H}^{(k)}\rVert_1 &=\lVert(\Phi_{\calE}^{(k)}-\Phi_{\H}^{(k)})(\ketbra{0}^{\otimes N})\rVert_1\le\lVert\Phi_{\calE}^{(k)}-\Phi_{\H}^{(k)}\rVert_\Diamond\le\epsilon.
\end{align}

In this work, we also consider approximate stabilizer state $k$-designs, which are defined according to the trace distance from the $k$-th moment of random stabilizer states
\begin{equation}
    \rho_{\C}^{(k)} := \Ens_{\psi\sim\C}\ketbra{\psi}^{\otimes k}.
\end{equation}
The notation $\Ens_{\psi\sim\C}$ means averaging over random stabilizer states.

\begin{definition}[Approximate state design with bounded relative error]
A state ensemble $\calE = \{\ket{\psi}\}$ forms an approximate state $k$-design with $\epsilon$ relative error if
\begin{equation}
(1-\varepsilon)\rho_{\H}^{(k)}\preceq \rho_\calE^{(k)}\preceq(1+\varepsilon)\rho_{\H}^{(k)}.\label{eq:approx_state_design_multp_error}
\end{equation}
where $\rho_1 \preceq \rho_2$ denotes that $\rho_2 - \rho_1$ is a positive semidefinite operator.
\end{definition}

The operational meaning of this definition is that there is only a small, multiplicative error between the probability of any given measurement outcome in $\rho_{\H}^{(k)}$ or in $\rho_{\calE}^{(k)}$.

A sufficient condition that guarantees  \eqref{eq:approx_state_design_multp_error}, is $\lVert \rho_\calE^{(k)}-\rho_{\H}^{(k)}\rVert_\infty\leq \binom{d+k-1}{k}^{-1}\varepsilon$.  This is because $\rho_\calE^{(k)}$ is supported only on the permutation symmetric subspace for any ensemble $\calE$, as this operator commutes with any permutation. Furthermore, $\rho_{\H}^{(k)}$ is the maximally mixed state in this subspace, with eigenvalues $\binom{d+k-1}{k}^{-1}$.

We make a few remarks. 
First, if the ensemble $\{U\}$ forms a unitary $k$-design within $\epsilon$ relative error, then $\calE = \{U \ket{0}\}$ forms a state $k$-design within $\epsilon$ relative error because $(1+\epsilon)\Phi_{\H}^{(k)}-\Phi_{\calE}^{(k)}\succeq 0$ implies that
\begin{equation}
(1+\epsilon)\rho_{\H}^{(k)}-\rho_{\calE}^{(k)}=[(1+\epsilon)\Phi_{\H}^{(k)}-\Phi_{\calE}^{(k)}](\ketbra{0}^{\otimes N})\succeq 0.
\end{equation}
Similarly, $\Phi_{\calE}^{(k)} - (1-\epsilon)\Phi_{\H}^{(k)}\succeq 0$ implies the operator inequality $\rho_{\calE}^{(k)}-(1-\epsilon)\rho_{\H}^{(k)}\succeq 0$.
Second, a state $k$-design with $\epsilon$ relative error is also a $k$-design with bounded additive error
\begin{align}
\onenormsm{\rho_{\calE}^{(k)} - \rho_{\H}^{(k)}}
    &= \max_{0\preceq M_+,M_-\preceq 1} \Tr\left((M_+ - M_-) (\rho_\calE^{(k)} - \rho_{\H}^{(k)})\right) \nn\\
    &\leq \epsilon \max_{0\preceq M_+\preceq 1}\Tr\left(M_+ \rho_{\H}^{(k)}\right) + \epsilon \max_{0\preceq M_-\preceq 1}\Tr\left(M_-\rho_{\H}^{(k)}\right) = 2\epsilon.
\end{align}

\subsection{Commutant of Clifford tensor powers}\label{sec:prelim_clifford}

In this section, we review the theory of the commutant of Clifford tensor powers -- the space of operators which commute with the $k$-fold tensor power of any Clifford unitary --  as developed in Ref.~\cite{Gross_2021}. This commutant naturally appears in averages over random Clifford unitaries.
We first explain the algebraic structure of the commutant of Clifford tensor powers and then provide formulae for averages over random Clifford unitaries and random stabilizer states, which will be used in this work.
We restrict our attention to the Clifford group over qubits.

Operators in the commutant of the $k$-th tensor power of the Clifford group are labeled by {stochastic Lagrangian subspaces}, a linear space defined below.
\begin{definition}[Stochastic Lagrangian subspace]
A linear subspace $T\subseteq \mathbb{Z}_2^{2k}$ is a stochastic Lagrangian subspace iff
\begin{enumerate}
\item $\boldx\cdot\boldx-\boldy\cdot\boldy=0$ mod $4$, $\forall (\boldx,\boldy)\in T$.
\item $\dim T=k$.
\item $\boldone_{2k}:=(\underbrace{1,\cdots,1}_{2k})\in T$.
\end{enumerate}
\end{definition}
We use $\Sigma_{k,k}$ to denote the set of all stochastic Lagrangian subspaces $T\subseteq \mathbb{Z}_2^{2k}$.
For each $T$, we have an associated operator in the commutant
\begin{equation}
\mathrm{r}(T)=\sum_{(\boldx,\boldy)\in T}\ket{\boldx}\bra{\boldy},\qquad \ket{\boldx}:=\bigotimes_{i=1}^k\ket{x_i}.
\end{equation}
For a system $S$ with $n$ qubits,
\begin{equation}
r(T)_S:=\mathrm{r}(T)^{\otimes n}.
\end{equation}
We omit the subscript when $S$ represents the entire system. 
In this work, we only consider the Clifford tensor powers on the subsystem with $n \geq k - 1$ such that the operators $r(T)$ are linearly independent.
The number of such independent operators is $|\Sigma_{k,k}|=\prod_{i=0}^{k-2}(2^i+1)=2^{O(k^2)}$.

The operator $\mathrm{r}(T)$ in the commutant follows the decomposition\footnote{The decomposition of $\mathrm{r}(T)$ stems from the structure of the Clifford commutant. The commutant can be written in terms of the double coset with respect to stochastic rotations, i.e. $\Sigma_{k,k} = O_k T_\mathrm{N_1} O_k \cup \cdots \cup O_k T_{\mathrm{N}_K} O_k$, where $T_{\mathrm{N_1}},\cdots T_{\mathrm{N_K}}$ are stochastic Lagrangian subspace associated with a defect subspace, and $O_k$ is the stochastic rotation group.}:
\begin{equation}
\mathrm{r}(T)=\mathrm{r}(T_{O,1})\mathrm{r}(T_\rmN)\mathrm{r}(T_{O,2}).
\end{equation}
The operator $\mathrm{r}(T_O)$ is a unitary operator, and the subspace $T_O = \{(\boldx,O\boldx), \boldx \in \bbZ_2^k\}$, where $O$ belongs to the stochastic rotation group defined below:
\begin{definition}[Stochastic rotation]
A $k\times k$ matrix $O$ with $\mathbb{Z}_2$-valued entries is a stochastic rotation iff
\begin{itemize}
\item $(O\boldx)\cdot(O\boldx)=\boldx\cdot\boldx$ mod $4$, $\forall\boldx\in\mathbb{Z}_2^k$.
\item $O\boldone_k=\boldone_k$.
\end{itemize}
\end{definition}
The stochastic rotations form a group, denoted as $O_k$, which includes the permutation group $S_k$ as a subgroup.
In the Clifford commutant, the operator associated with a stochastic rotation takes the form
\begin{equation}
\mathrm{r}(T_O):=\sum_{\boldx}\ket{O\boldx}\bra{\boldx}.
\end{equation}

The operator $\mathrm{r}(T_\rmN)$ is associated with a defect subspace defined below.
\begin{definition}[Defect subspace]\label{def:defect_subspace}
A linear subspace $\rmN\subseteq\mathbb{Z}_2^k$ is a defect subspace iff
\begin{itemize}
\item $\boldx\cdot\boldx=0 \mod 4$, $\forall\boldx\in \rmN$.
\item $\boldone_k\in \rmN^\perp$, i.e. $\boldx\cdot\boldone_k = 0 \mod 4$, $\forall \boldx \in \rmN$.
\end{itemize}
\end{definition}
Here, $\rmN^\perp$ is the orthogonal complement of the defect subspace $\rmN$ in $\bbZ_2^k$, which includes $\rmN$ as a subspace according to Def.~\ref{def:defect_subspace}.
The stochastic Lagrangian subspace $T_\rmN$ is given by $T_\rmN = \{(\boldx,\boldx+\boldy), \boldx\in \rmN^\perp, \boldy \in \rmN\}$.

The operator $r(T_\rmN)$ is proportional to the projection operator $P_{\CSS(\rmN)}$ onto the code space of the CSS code with stabilizers $\langle Z^{\boldu}X^{\boldv},\boldu,\boldv\in \rmN\rangle$:
\begin{equation}
\mathrm{r}(T_\rmN)=2^{\dim\rmN}P_{\CSS(\rmN)}.
\end{equation}
The code space is of dimension $2^{k-2\dim\rmN}$ specified by the stabilizer group with $\abs{\rmN}^2$ elements.

We are now ready to write down the norm of operators in the Clifford commutant:
\begin{equation}
\lVert \mathrm{r}(T)\rVert_1=2^{k-\dim\rmN},\qquad \lVert\mathrm{r}(T)\rVert_\infty=2^{\dim\rmN}.
\end{equation}
Here, we use the decomposition of $\mathrm{r}(T)$ and the fact that the unitary operators, $\mathrm{r}(T_{O,1})$ and $\mathrm{r}(T_{O,2})$, leave the Schatten norms invariant.

The overlap between two stochastic Lagrangian subspaces $T_{1,2}$ is given by
\begin{equation}
\bbrakket{T_1}{T_2} := \Tr[\mathrm{r}(T_1)^\top \mathrm{r}(T_2)]:=2^{k-|T_1,T_2|}=2^{\dim(T_1\cap T_2)}.
\end{equation}
Here, we introduce the vectorized notation $\kket{T}:=\mathrm{vec}(\mathrm{r}(T)) = \sum_{(\boldx,\boldy)\in T} \ket{\boldx}\otimes\ket{\boldy}$.
For two different subspaces $T_1\neq T_2$, $\dim(T_1\cap T_2)\leq k-1$, i.e. $|T_1,T_2|\geq 1$.

The $k$-fold channel associated with the Clifford group considered in this work can be expressed in terms of the operators in the Clifford commutant~\cite{Gross_2021,Li:2021dbh}:
\begin{align}
\Phi_{\C}^{(k)}(\cdot)=\Ens_{V\sim\C}V^{\otimes k}(\cdot)V^{\dagger\otimes k} 
=&\sum_{T_1,T_2}\Wg_{\C}(T_1,T_2)r(T_1)\Tr[r(T_2)^{\top}(\cdot)] \nn \\
:=&\sum_{T_1,T_2}\Wg_{\C}(T_1,T_2)|T_1)(T_2|(\cdot),\label{eq:clifford_commutant_weingarten_def}
\end{align}
where we assume $n\geq k-1$ such that $r(T)$ are linear independent. The coefficients $\Wg_{\C}(T_1,T_2)$ are called \emph{Clifford Weingarten function} satisfying
\begin{align}
    \sum_{T_2}\Wg_{\C}(T_1,T_2)\Tr[r(T_2)^{\top}r(T_3)]=\delta_{T_1,T_3}.
\end{align}
We introduce the notation $|T_1)(T_2|(\cdot) := r(T_1)\Tr[r(T_2)^{\top}(\cdot)]$, representing a superoperator.
The decomposition~\eqref{eq:clifford_commutant_weingarten_def} follows from the properties $[\Phi_{\C}^{(k)}(O),V^{\otimes k}]=0$, and $\Phi_{\C}^{(k)}([O,V^{\otimes k}])=0$ for any operator $O$ and Clifford unitary $V$.

The $k$-th moment of random stabilizer state over $n$ qubits is given by~\cite{Gross_2021}
\begin{equation}
\rho_{\C}^{(k)}=\frac{1}{Z_{n,k}}\sum_{T}r(T),
\end{equation}
where the prefactor $Z_{n,k} = 2^n\prod_{i=0}^{k-2}(2^n+2^i)$ and has an upper bound $Z_{n,k}\leq 2^{nk}(1 + k 2^{k-n})$ for $n\geq k+\log k-2$~\cite{Haferkamp:2020qel}.

\section{Uniformity condition}\label{sec:uniformity}
In this section, we introduce the approximate uniformity conditions for the $k$-fold channel of a Clifford unitary ensemble and the $k$-th moment of a stabilizer state ensemble.
We show that a two-layer random Clifford circuit with each gate operating on log-size qubit blocks can generate a unitary ensemble that satisfies approximate uniformity.
This property guarantees that such circuits, when acting on a product state, can generate a stabilizer state $k$-design with bounded additive error (Corollary~\ref{thm:state_additive_clifford_design}).

\begin{definition}[Approximate uniformity condition for Clifford unitaries]
Let $\calE$ be an ensemble of quantum circuits composed of Clifford gates acting on $N$ qubits. The $k$-fold channel for this circuit ensemble satisfies the $\epsilon$-approximate uniformity condition if 
\begin{itemize}
\item The $k$-fold channel takes the form
\begin{equation}
\Phi^{(k)}_\calE=2^{-Nk}\sum_T|T)(T|+
2^{-Nk}
\sum_{\vec{T}_1,\vec{T_2}}f(\vec{T}_1,\vec{T}_2)|\vec{T}_1)(\vec{T}_2|,\label{eq:channel_uniformity}
\end{equation}
where $\vec{T}$ is a vector that labels the spatial configuration of $T$ for each subregion, $|\vec{T}_1)(\vec{T}_2| = \bigotimes_x|T_{1,x})(T_{2,x}|$ is a tensor product of superoperators over spatial regions labeled by $x$, and $|T)(T|$ is a superoperator associated with a spatially uniform configuration of $T$.

\item The deviation from uniformity is bounded by $\epsilon$, i.e. $\Delta(\Phi^{(k)}_\calE) :=\sum_{\vec{T}_1,\vec{T_2}}|f(\vec{T}_1,\vec{T}_2)|\le\epsilon$.
\end{itemize}

\end{definition}

\begin{definition}[Approximate uniformity condition for stabilizer states]
Let $\calE$ be an ensemble of Clifford states on $N$ qubits. The $k$-th moment of this state ensemble satisfies the $\epsilon$-approximate uniformity condition if 
\begin{itemize}
\item The $k$-th moment can be written as
\begin{equation}
\rho^{(k)}_{\calE} =2^{-Nk}\sum_T r(T)+2^{-Nk}\sum_{\vec{T}}f(\vec{T})r(\vec{T}),
\end{equation} 
where $r(\vec{T})=\bigotimes_xr(T_x)$, and $r(T)$ is an operator associated with a spatially uniform configuration of $T$.

\item The deviation from uniformity is bounded by $\Delta(\rho_{\calE}^{(k)}) := \sum_{\vec{T}}|f(\vec{T})|\le\varepsilon$.
\end{itemize}
\end{definition}

The approximate uniformity condition for the unitary ensemble $\calE_u = \{V\}$ implies that the unitary ensemble can generate a state ensemble $\calE_s = \{V\ket{0}^{\otimes N}\}$ with approximate uniformity.
Specifically, the $k$-th moment of the state ensemble is given by
\begin{equation}
\begin{aligned}
\rho_{\calE_s}^{(k)}&=\Phi_{\calE_u}^{(k)}(\ketbra{0}^{\otimes N})=2^{-Nk}\left[\sum_T r(T)+\sum_{\vec{T}_1,\vec{T_2}}f(\vec{T}_1,\vec{T}_2)r(\vec{T}_1)\right].
\end{aligned}
\end{equation}
Its deviation from uniformity has an upper bound,
\begin{equation}\label{eq:uniformity_unitary_implies_state}
\Delta(\rho_{\calE_s}^{(k)}) =\sum_{\vec{T}_1}\Big|\sum_{\vec{T_2}}f(\vec{T}_1,\vec{T}_2)\Big|\le\sum_{\vec{T}_1,\vec{T_2}}|f(\vec{T}_1,\vec{T}_2)|=\epsilon.
\end{equation}

We remark that a state ensemble with approximate uniformity is a state $k$-design with bounded additive error. 
\begin{theorem}\label{thm:uniformity_implies_stabilizer_state_design}
Consider a stabilizer state ensemble satisfying the $\epsilon$-approximate uniformity condition. The state ensemble forms a stabilizer state design with bounded additive error $\epsilon+k2^{-N+k}$ assuming $N\geq k + \log k - 2$.
\end{theorem}
\begin{proof}
The $k$-th moment of the state ensemble with approximate uniformity is given by
\begin{align}
    \rho^{(k)}_{\calE} =2^{-Nk}\left[\sum_T r(T)+\sum_{\vec{T}}f(\vec{T})r(\vec{T})\right].
\end{align}
The additive error has an upper bound (assuming $N\geq k + \log k - 2$)
\begin{equation}
\begin{aligned}
&\lVert\rho_\calE^{(k)}-\rho_{\C}^{(k)}\rVert_1 \leq |2^{-Nk}Z_{N,k}-1|\cdot\lVert\rho_{\C}^{(k)}\rVert_1+2^{-Nk}\sum_{\vec{T}}|f(\vec{T})|\cdot\lVert r(\vec{T})\rVert_1 \leq k2^{-N+k}+\epsilon,
\end{aligned}
\end{equation}
where we use the triangle inequality, and the bound on the trace norm $\onenormsm{\mathrm{r}(T)} = 2^{k - \dim\rmN} \leq 2^k$.
\end{proof}

Although the approximate uniformity in general does not imply the ensemble to form designs under other approximation measures, it is a crucial property that allows one to convert the ensemble to an approximate design when augmented with magic gates (shown in Sec.~\ref{sec:rel_error} and~\ref{sec:add_error}).
In what follows, we show that a Clifford ensemble with approximate uniformity can be generated in low-depth circuits.

The result stems from a gluing lemma; one can obtain a Clifford ensemble with the approximate uniformity condition by gluing small blocks of random Clifford unitaries together.
\begin{lemma}[Gluing lemma for Clifford ensembles with approximate uniformity]\label{thm:uniformity_gluing}
Consider a system of three partitions $A$, $B$, and $C$. 
Clifford unitaries $V_{BC}$ are drawn from an (exact) Clifford $k$-design, and $V_{AB}$ are drawn from an ensemble of Clifford unitaries on $AB$ satisfying the $\epsilon'$-approximate uniformity condition with $\epsilon'<1$. 
\begin{equation*}
\includegraphics[width=11.5cm]{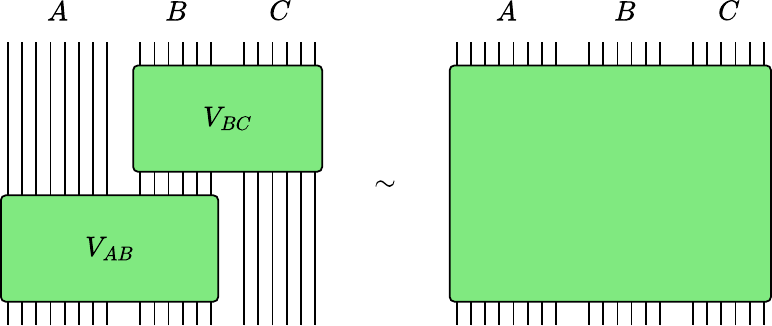}
\end{equation*}
Assume that $T$ is aligned in the overlapping region $B$, that is, the $k$-fold channel $\Phi^{(k)}_{AB}$ of $\{V_{AB}\}$ can be written as 
\begin{equation}
\Phi^{(k)}_{AB} = 2^{-N_{AB}k}\Big(\sum_{T_3}|T_3)(T_3|_{AB} +\sum_{T_3,T_4,\vec{T}_5,\vec{T}_6}f\big((\vec{T}_5,T_3),(\vec{T}_6,T_4)\big)|\vec{T}_5)(\vec{T}_6|_A\otimes |T_3)(T_4|_B\Big).
\end{equation}
Then, the unitary ensemble $\calE = \{ V_{BC} V_{AB} \}$ satisfies the uniformity condition up to the error
\begin{equation}
    \epsilon \le\epsilon'+2^{-N_B+O(k^2)}.
\end{equation}
\end{lemma}
\begin{proof}
The $k$-fold channel of ensemble $\calE$ is given by
\begin{equation}
\begin{aligned}
\Phi_\calE^{(k)}&=\Phi_{BC}^{(k)} \circ \Phi_{AB}^{(k)} 
= 2^{-N_{AB}k}\sum_{T_{1,2,3}}\Wg_{\C}(T_1,T_2)2^{N_B(k-|T_2,T_3|)} |T_1)(T_2|_C\otimes\\
&\left[|T_3)(T_3|_A\otimes |T_1)(T_3|_B+\sum_{T_4,\vec{T}_5,\vec{T}_6}f\big((\vec{T}_5,T_3),(\vec{T}_6,T_4)\big)|\vec{T}_5)(\vec{T}_6|_A\otimes |T_1)(\vec{T_4}|_B\right] ,
\end{aligned}
\end{equation}
as shown pictorially below:
\begin{equation*}
\includegraphics[height=5.5cm]{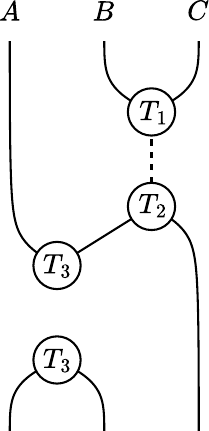}\qquad\qquad\qquad
\includegraphics[height=5.5cm]{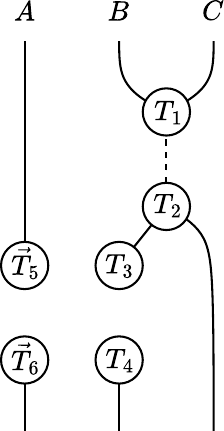}
\end{equation*}
where the dashed lines denote the Clifford Weingarten function $\Wg_{\C}(T_1,T_2)$ on $BC$, and the solid lines connecting two variables $T_2$ and $T_3$ denote $\Tr[r(T_2)_B^{\top}r(T_3)_B]=2^{N_B(k-|T_2,T_3|)}$.

To show that $\Phi_\calE^{(k)}$ satisfies the uniformity condition, we divide the sum into three parts. 
The first part contains the terms without $T_4$, $\vec{T}_5$, and $\vec{T}_6$ satisfying $T_1=T_2=T_3$. The contribution of such configurations to $\Phi_\calE^{(k)}$ is
\begin{equation}
2^{-N_{A}k}\sum_{T}\Wg_{\C}(T,T)|T)(T|.
\end{equation}
The contribution to the error $\Delta(\Phi_\calE^{(k)})$ is bounded by
\begin{equation}
\sum_T \left|2^{N_{BC}k}\Wg_{\C}(T,T)-1\right|\leq 2^{-N_{BC}+O(k^2)},
\end{equation}
where we use the asymptotic behavior $|\Wg_{\C}(T,T)-2^{-N_{BC}k}|\le2^{-N_{BC}k-N_{BC}+O(k^2)}$ of the Clifford Weingarten function in Appendix~\ref{sec:clifford_commutant_wg_asymptotics}.

The second part contains the terms without $T_4$, $\vec{T}_5$, and $\vec{T}_6$, but either $T_1\neq T_2$ or $T_2\neq T_3$. The contribution of such configurations to $\Phi_\calE^{(k)}$ is
\begin{gather}
2^{-N_{A}k-N_B\abs{T_2,T_3}}\sum_{T_1\neq T_3}\sum_{T_2}\Wg_{\C}(T_1,T_2)
|T_3)(T_3|_A\otimes |T_1)(T_3|_B\otimes |T_1)(T_2|_C.
\end{gather}
The contribution to the error $\Delta(\Phi_\calE^{(k)})$ is bounded by
\begin{equation}
\sum_{T_1\neq T_3}\sum_{T_2}2^{N_{BC}k}|\Wg_{\C}(T_1,T_2)|2^{-N_B|T_2,T_3|}\leq 2^{-N_B+O(k^2)},
\end{equation}
where we use (1) if $T_1\neq T_2$, the Weingarten function are suppressed as $|\Wg_{\C}(T_1,T_2)|\leq 2^{-N_{BC}k-N_{BC}+O(k^2)}$ as shown in Appendix~\ref{sec:clifford_commutant_wg_asymptotics}; (2) if $T_2\neq T_3$, the factor $2^{-N_B|T_2,T_3|}$ is suppressed as $|T_2,T_3|\geq 1$.

The third part contains the terms with $T_4$, $\vec{T}_5$, and $\vec{T}_6$. The contribution to the error $\Delta(\Phi_\calE^{(k)})$ is bounded by 
\begin{equation}
\begin{aligned}
&\sum_{T_3,T_4,\vec{T}_5,\vec{T}_6}|f\big((\vec{T}_5,T_3),(\vec{T}_6,T_4)\big)|\sum_{T_1,T_2}2^{N_{BC}k}|\Wg_{\C}(T_1,T_2)|2^{-N_B|T_2,T_3|} \\
\leq& \sum_{T_3,T_4,\vec{T}_5,\vec{T}_6}|f\big((\vec{T}_5,T_3),(\vec{T}_6,T_4)\big)|(1+2^{-N_B+O(k^2)}) \\
\le&\epsilon'(1+2^{-N_B+O(k^2)}).
\end{aligned}
\end{equation}
In the second line, the first term (in the small bracket) is from the terms with $T_1=T_2=T_3$. 
The second term comes from the terms that $T_1$, $T_2$, and $T_3$ are not aligned.
In this case, these terms has a suppression either from $|\Wg_{\C}(T_1,T_2)|$ if $T_1\neq T_2$ or from $2^{-N_B|T_2,T_3|}$ if $T_2\neq T_3$.

To sum up, the deviation from uniformity is bounded by
\begin{equation}
\Delta(\Phi_\calE^{(k)})\le\epsilon'+2^{-N_B+O(k^2)},
\end{equation}
where we use $\epsilon'<1$.
\end{proof}

Using the gluing lemma above, we can construct a global Clifford ensemble with approximate uniformity using two layers of Clifford $k$-design gates acting on qubit blocks.
\begin{theorem}[Approximate uniformity from two-layer Clifford circuits]
\label{thm:uniformity_log_depth}
Consider an ensemble $\calE_u = \{V\}$ of two-layer Clifford circuits with each gate drawn independently from an exact Clifford $k$-design. 
The Clifford ensemble $\calE_u$ and the stabilizer state ensemble $\calE_s = \{V\ket{0}^{\otimes N}\}$ satisfy the $\epsilon$-approximate uniformity condition with $\epsilon=N2^{-\xi-\log\xi+O(k^2)}$, where $\xi$ is the size of the smallest overlapping regions in the two-layer Clifford circuit.
\end{theorem}
\begin{proof}
We bound the deviation from uniformity by applying the gluing lemma sequentially for the two-layer Clifford circuit.
The error is accumulated by $2^{-\xi + O(k^2)}$ for each application.
We need to apply the gluing lemma at most $N/\xi$ times, where $\xi$ is the size of the smallest overlapping regions in the two-layer circuit.
Thus, the deviation of the $k$-fold channel $\Phi_{\calE_u}^{(k)}$ from uniformity has an upper bound $\Delta(\Phi_{\calE_u}^{(k)}) \leq N2^{-\xi-\log\xi+O(k^2)}$.
This implies the state ensemble $\calE_s = \{V\ket{0}^{\otimes N}\}$ satisfies the approximate uniformity up to error $\epsilon = N2^{-\xi-\log\xi+O(k^2)}$.
\end{proof}
According to Theorem~\ref{thm:uniformity_log_depth}, to achieve the $\epsilon$-approximate uniformity, it suffices to have overlapping regions of size $\xi = \log (N/\epsilon) + O(k^2)$.
Namely, one can realize it in shallow 1D Clifford circuits of depth $d = O(\log (N/\epsilon)+k^2)$~\cite{Bravyi:2020xdv}.

As a corollary, since the uniformity condition implies a stabilizer $k$-design with bounded additive errors (Theorem~\ref{thm:uniformity_implies_stabilizer_state_design}), one can generate an additive-error stabilizer state design using log-depth Clifford circuits\footnote{Here, the assumption in Theorem~\ref{thm:uniformity_implies_stabilizer_state_design} that $N > k + \log k -2$ is satisfied automatically as $N \ge \xi \ge O(k^2)$.}.
\begin{corollary}[Additive-error stabilizer state designs from low-depth Clifford circuits]
\label{thm:state_additive_clifford_design}
Let $V$ be a two-layer brickwork Clifford circuit, where each gate is drawn from an exact Clifford $k$-design. 
Let $\xi$ be the smallest size of the overlapping regions of the Clifford unitaries. 
Then the state ensemble $\{V\ket{0}^{\otimes N}\}$ forms an approximate stabilizer state $k$-design with additive error $\varepsilon\leq N2^{-\xi-\log\xi+O(k^2)}$.

In other words, one can achieve $\varepsilon$ additive error for $\xi=\log (N/\varepsilon)+O(k^2)$, or equivalently, in depth $O(\log (N/\varepsilon)+k^2)$ using 1D circuits or in depth $O(\log\log (N/\varepsilon)+\log k)$ using all-to-all circuits and $O(N(\log (N/\varepsilon)+k^2))$ ancillas~\cite{Jiang:2019aom}.
\begin{equation*}
\includegraphics[width=9.8cm]{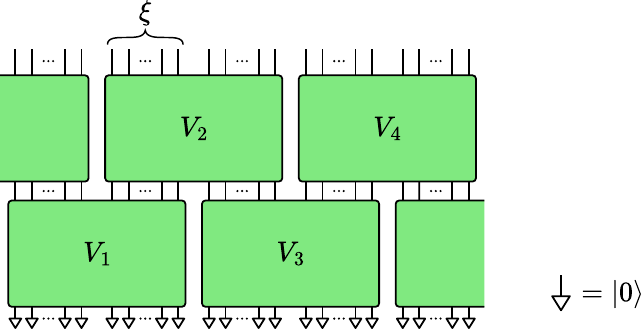}   
\end{equation*}
\end{corollary}

We note that a recent paper~\cite{Grevink:2025nez} proved that two-layer Clifford circuits operating on log-size qubit blocks can generate additive-error stabilizer state $k$-designs for $k = 4,5$. Our proof is different and shows that such circuits can in fact form designs for any $k$.

Furthermore, in Appendix~\ref{app:collision_prob}, we show that the uniformity also implies that the log-depth Clifford circuits can approximate the $k$-th collision probability in global random Clifford gates up to a bounded relative error.

\section{Designs with bounded relative error}\label{sec:rel_error}
In this section, we show that the Clifford ensemble with the approximate uniformity condition, when augmented with constant-depth magic gates, can generate state and unitary $k$-designs with bounded relative error.

\subsection{State designs with bounded relative error}\label{sec:rel_error_state}
We prove that two-layer Clifford unitaries with each gate acting on $O(\log (N/\epsilon)+k^2)$ qubits followed by $k$-design random unitaries acting on disjoint qubit clusters of size $O(k\log k)$ can generate a state $k$-design with bounded relative error.
The constructed state ensemble consists of states with strictly local magic~\cite{Amy_2013,Zhang:2024jlz,Korbany:2025noe,Wei:2025irp,Andreadakis:2025mfw,Parham:2025sxj}, namely each state can be converted to a stabilizer state by a constant-depth local unitary circuit.

\begin{theorem}[Relative-error state designs from low-depth Clifford unitaries followed by constant depth magic gates]
\label{thm:state_relative_log_depth}
Consider a state ensemble $\calE = \{UV\ket{0}^{\otimes N}\}$ generated by the circuit architecture below
\begin{equation*}
\includegraphics[width=10.8cm]{Figure/state_relative_log_depth.pdf}
\end{equation*}
Here, $U = \bigotimes_{x,i} U_{x,i}$ is a product of random unitaries drawn from exact unitary $k$-designs, each gate $U_{x,i}$ acts on a disjoint subsystem $M_{x,i}$ of size $N_{x,i}$ within $M_x$. 
The two-layer Clifford unitary $V$ consists of multi-qubit Clifford unitaries drawn independently from exact Clifford $k$-designs.

The ensemble $\calE$ forms a state $k$-design with relative error 
\begin{align}
\epsilon = N2^{-\xi-\log\xi+O(k^2)}+2^{O(k^2)-N}\left[k!(1+k^22^{-\ell})\right]^{N/\ell},
\end{align}
where $\xi$ is the size of the smallest overlapping regions of the Clifford unitaries, and $\ell = \min\{N_{x,i}\} \geq k \log k$ is the size of the smallest subsystem that $U_{x,i}$ acts on.

This implies that, with $\xi=O(\log(N/\epsilon)+k^2)$ and $\ell=O(k\log k)$, the ensemble $\calE$ can form a state $k$-design with $\epsilon$ relative error, assuming $N=\omega(k^2+\log(1/\epsilon))$.

\end{theorem}
\begin{proof}

To begin, the $k$-th moment of the state ensemble generated by a two-layer Clifford unitary can be expressed as
\begin{align}
    \rho_{\calE_0}^{(k)}=2^{-Nk}\left[\sum_T r(T) + \sum_{\vec{T}}f(\vec{T})r(\vec{T})\right],
\end{align}
where $\vec{T}$ is a vector of $T_x$, each of which emerged from averaging the $k$-fold channel associated with a Clifford gate in the second layer at location $x$.
The operator $r(\vec{T})=\bigotimes_x r(T_x)_{M_x}$, where $M_x$ labels the qubit block that the Clifford gate at location $x$ acts on.
It satisfies the $\epsilon'$-approximate uniformity condition with $\epsilon'=N2^{-\xi-\log\xi+O(k^2)}$ due to Theorem~\ref{thm:uniformity_log_depth}.

Next, we apply the $k$-design random unitary gates over subsystems $M_{x,i}$ of size $N_{x,i}$ within each qubit block $M_x$. 
Recall that the $k$-fold channel of each Haar random unitary is given by
\begin{align}
    \Ens_{U \sim \H}U^{\otimes k} (\cdot) U^{\dagger \otimes k} = \sum_{\sigma, \tau\in S_k}\Wg(\sigma,\tau) r(\sigma) \Tr[r(\tau)^\top (\cdot)],
\end{align}
where each $\sigma$ and $\tau$ belongs to the permutation group $S_k$, corresponding to a subset of the Clifford commutant.
The resulting $k$-th moment takes the form 
\begin{equation}\label{eq:bounded_relative_1}
\rho_{\calE}^{(k)} = 2^{-Nk}\sum_{\vec{\sigma},\vec{\tau}}\Big[\sum_Tr(\vec{\sigma}) \Wg(\vec{\sigma},\vec{\tau})  \bbrakket{\vec{\tau}}{T}
+ \sum_{\vec{T}}f(\vec{T})r(\vec{\sigma})\Wg(\vec{\sigma},\vec{\tau})\bbrakket{\vec{\tau}}{\vec{T}}\Big],
\end{equation}
where $\vec{\sigma}$ ($\vec{\tau}$) is a vector labeling configurations of permutations $\sigma_{x,i}$,  $\tau_{x,i}\in S_{k}$ at location $x,i$, and $\Wg(\vec{\sigma},\vec{\tau}) := \prod_{x,i}\Wg(\sigma_{x,i},\tau_{x,i})$, $\bbrakket{\vec{\tau}}{T}:=\prod_{x,i}\bbrakket{\tau_{x,i}}{T}$, $\bbrakket{\vec{\tau}}{\vec{T}}:=\prod_{x,i}\bbrakket{\tau_{x,i}}{T_x}$.

To bound the relative error between $\rho_{\calE}^{(k)}$ and $\rho_{\H}^{(k)}$, we divide the contribution to $\rho_{\calE}^{(k)}$ into three parts:
\begin{itemize}
    \item The first part is given by the first term in Eq. (\ref{eq:bounded_relative_1}), choosing $T = \pi\in S_{k}$. The summation over $\vec{\tau}$ results in a delta function, i.e. $\sum_{\tau_{x,i}}\Wg(\sigma_{x,i},\tau_{x,i})\bbrakket{\tau_{x,i}}{\pi}=\delta_{\sigma_{x,i},\pi}$, which identifies each $\sigma_{x,i}$ with $\pi$. This part of the contribution is
    \begin{align}
        \rho_\rmI^{(k)} = 2^{-Nk}\sum_{\pi\in S_k} r(\pi).
    \end{align}
    \item The second part is given by the first term in Eq. (\ref{eq:bounded_relative_1}), taking $T\notin S_{k}$
    \begin{align}
    \rho_{\rmII}^{(k)} = 2^{-Nk}\sum_{\vec{\sigma},\vec{\tau}}\sum_{T\notin S_k} r(\vec{\sigma}) \Wg(\vec{\sigma},\vec{\tau}) \bbrakket{\vec{\tau}}{T},    
    \end{align}
    which has a bounded $\infty$-norm
    \begin{align}
        \lVert \rho_{\rmII}^{(k)}\rVert_\infty \leq & 2^{-Nk}\sum_{\vec{\sigma},\vec{\tau}}\sum_{T\notin S_k} \lVert r(\vec{\sigma})\rVert_\infty\abs{\Wg(\vec{\sigma},\vec{\tau})} \bbrakket{\vec{\tau}}{T} \nonumber \\
        \leq & 2^{-Nk} 2^{O(k^2)-N} \prod_{x,i} k! (1 + k^2 2^{-N_{x,i}}) \nonumber \\
        \leq & 2^{-Nk} 2^{O(k^2)-N} \left[ k! (1 + k^2 2^{-\ell})\right]^{N/
        \ell}.\label{eq:rel_error_state_rhoII}
    \end{align}
    In the second inequality, we use $\lVert r(\vec{\sigma})\rVert_\infty = 1$, $\bbrakket{\vec{\tau}}{T} \leq 2^{(k-1)N}$, and $\sum_{\sigma_{x,i},\tau_{x,i}} |\Wg(\sigma_{x,i},\tau_{x,i})|\leq 2^{-kN}k!(1+k^2 2^{-N_{x,i}})$ for $k^2<2^{N_{x,i}}$~\cite{Aharonov:2021das}.
    In the third inequality, we take $\ell = \min\{N_{x,i}\}$ to be the size of the smallest qubit cluster.
    \item The third part contains the remaining terms in the summation
    \begin{align}
        \rho_{\rmIII}^{(k)} = 2^{-Nk} \sum_{\vec{\sigma},\vec{\tau}}\sum_{\vec{T}}f(\vec{T})r(\vec{\sigma})\Wg(\vec{\sigma},\vec{\tau})\bbrakket{\vec{\tau}}{\vec{T}}. 
    \end{align}
    Its $\infty$-norm has an upper bound
    \begin{align}
        \infnormsm{\rho_{\rmIII}^{(k)}} \leq &2^{-Nk} \sum_{\vec{T}} \abs{f(\vec{T})}
    \prod_{x,i} \bigg\lVert\sum_{\sigma_{x,i},\tau_{x,i}}r(\sigma_{x,i})\Wg(\sigma_{x,i},\tau_{x,i})\bbrakket{\tau_{x,i}}{T_x}\bigg\rVert_{\infty}.\label{eq:rel_error_state_rhoIII}
    \end{align}
    Here, we consider the summation over each $\sigma_{x,i}$ and $\tau_{x,i}$. We consider two cases for $T_x \in S_k$ and $T_x \notin S_k$. 
    For $T_x = \pi \in S_k$,
    \begin{align}
    \bigg\lVert\sum_{\sigma_{x,i},\tau_{x,i}}r(\sigma_{x,i})\Wg(\sigma_{x,i},\tau_{x,i})\bbrakket{\tau_{x,i}}{\pi}\bigg\rVert_\infty = \infnorm{r(\pi)} = 1
    \end{align}
    For $T_x \notin S_k$,
    \begin{align}
    \bigg\lVert\sum_{\sigma_{x,i},\tau_{x,i}}r(\sigma_{x,i})\Wg(\sigma_{x,i},\tau_{x,i})\bbrakket{\tau_{x,i}}{T_x}\bigg\rVert_{\infty}\leq& \sum_{\sigma_{x,i},\tau_{x,i}} \abs{\Wg(\sigma_{x,i},\tau_{x,i})} \bbrakket{\tau_{x,i}}{T_x} \nn \\
    \leq& 2^{-N_{x,i}}k!(1+k^22^{-N_{x,i}})\leq 1,
    \end{align}
    where we use $|\tau_{x,i},T_{x,i}|\geq 1$ and assume that $\ell \geq k\log k$.
    Hence, we have
    \begin{align}
        \infnormsm{\rho_{\rmIII}^{(k)}} \leq & 2^{-Nk} \sum_{\vec{T}} \abs{f(\vec{T})} \leq 2^{-Nk} \epsilon'.\label{eq:rel_state_rhoIII_infnorm}
    \end{align}

\end{itemize}

Combining these results, we obtain an upper bound on the relative error
\begin{align}
    \epsilon \leq & \frac{\infnormsm{\rho_{\calE}^{(k)} - \rho_{\H}^{(k)}}}{\infnormsm{\rho_{\H}^{(k)}}} \leq \frac{\infnormsm{\rho_{\rmI}^{(k)} - \rho_{\H}^{(k)}} + \infnormsm{\rho_{\rmII}^{(k)}} + \infnormsm{\rho_{\rmIII}^{(k)}}}{\infnormsm{\rho_{\H}^{(k)}}} \nn \\
    \leq & k^22^{-Nk} + 2^{O(k^2)-N}[k!(1+k^22^{-\ell})]^{N/\ell}+ (1 + k^2 2^{-Nk})\epsilon' \nn\\
    \le& 2^{O(k^2)-N}[k!(1+k^22^{-\ell})]^{N/
    \ell}+N2^{-\xi-\log\xi+O(k^2)}.
\end{align}
\end{proof}

We remark that the magic gates do not need to act on every qubit cluster in the final layer.
In fact, one can remove a finite fraction of Haar random gates within each qubit block $M_x$, while maintaining a bounded relative error.
Consider a qubit block of size $\xi$ that a Clifford gate acts on.
Suppose no magic gates act on a subregion $C$ of size $N_C$ inside the qubit block.
We can follow a similar analysis of the relative error except that in Eq.~\eqref{eq:rel_error_state_rhoII} and~\eqref{eq:rel_error_state_rhoIII}, the summation involves terms with the non-permutation $r(T)_C$ in region $C$, which can have a greater spectral norm $\infnormsm{r(T)_C} = 2^{N_C\dim \mathrm{N}}$ than those of permutations. 
Using the upper bound of its spectral norm $\lVert r(T)_C\rVert_\infty\le2^{N_Ck/2}$, we obtain a modified multiplicative factor, e.g. in Eq.~\eqref{eq:rel_error_state_rhoII},
\begin{equation}
    2^{-\xi}\prod_i k!(1+k^22^{-\ell})\,\mapsto\, 2^{-(\xi-N_C)+N_Ck/2}\prod_i k!(1+k^22^{-\ell}).
\end{equation}
This factor is still bounded for $N_C=\frac{1}{k}\xi$ and $\ell = O(k\log k)$, indicating the possibility to remove at least $1/k$ fraction of Haar random gates. 

In the case with no magic gates in the subregion $C$, it is tempting to distinguish this circuit from Haar random states using magic measures.
To begin, one can consider magic measures on the subregion $C$.
This cannot work because Page's theorem implies that in the thermodynamic limit $N \to \infty$, the state on $C$ is exponentially close to the maximally mixed state with no magic, assuming that the size of $C$ is less than half of the system.
Alternatively, one might try to measure the magic in the entire system. 
Since our circuits form a relative error design, the expectation value of any positive operator is relatively close to the Haar random value.
This implies that, for example, the expectation value of the operator $r(T_S)$ (defined in Eq.~\eqref{eq:rTs} in Sec.~\ref{sec:rel_error_nogo}), whose logarithm gives the second stabilizer Rényi entropy, will acquire a value in our state ensemble that is relatively close to the Haar value. 
Naively, one may expect the magic in the entire system to deviate from the maximal value by an extensive amount. 
However, magic gates on a subregion are capable of generating a near-maximum stabilizer Rényi entropy when acting on a highly entangled state, as shown in Ref.~\cite{Hou:2025bau}.

\subsection{Unitary designs with bounded relative error}\label{sec:rel_error_unitary}

In this section, we show that an ensemble of Clifford unitaries that satisfies the approximate uniformity condition can be converted to a unitary design with bounded relative error by sandwiching the Clifford unitaries with constant-depth Haar random unitaries.
This allows us to construct a unitary $k$-design with $\epsilon$ relative error using one-dimensional circuits with a total depth $O(\log (N/\epsilon)) + 2^{O(k\log k)}$ or using all-to-all circuits with depth $O(\log\log (N/\epsilon)) + 2^{O(k\log k)}$.

\begin{theorem}[Relative-error unitary designs from low-depth Clifford unitaries sandwiched by constant depth magic gates]
Consider an ensemble $\calE = \{U V U'\}$ of unitary circuits given by the architecture below
\begin{equation*}
\includegraphics[width=8.2cm]{Figure/unitary_relative_log_depth.pdf}
\end{equation*}
Here, $U = \bigotimes_{x,i} U_{x,i}$ ($U' = \bigotimes_{x',i'} U'_{x',i'}$) is a product of random unitaries drawn from exact $k$-design acting on disjoint clusters of qubits. Each gate $U_{x,i}$ ($U'_{x',i'}$) acts on a qubit cluster $M_{x,i}$ ($M'_{x',i'}$) of size $N_{x,i}$ ($N'_{x',i'}$). 
The two-layer random Clifford unitary $V$ consists of multi-qubit Clifford unitaries drawn independently from exact Clifford $k$-designs.

The unitary ensemble $\calE$ forms a unitary $k$-design with $\epsilon$ relative error 
\begin{align}
    \epsilon = 2^{O(k^2)-2N} \left[k!(1 + k^2 2^{-\ell})\right]^{2N/\ell} + N2^{-\xi-\log \xi+O(k^2)}
\end{align}
where $\xi$ is the size of the smallest overlapping regions of individual Clifford gates, and $\ell = \min\{N_{x,i},N'_{x',i'}\}$ is the size of the smallest qubit cluster.
Here, we have assumed that $\ell \geq k \log k$.

This implies that, with $\xi=O(\log(N/\epsilon)+k^2)$ and $\ell=O(k\log k)$, the ensemble $\calE$ can form a state $k$-design with $\epsilon$ relative error assuming $N=\omega(k^2+\log(1/\epsilon))$.
\end{theorem}
\begin{proof}
The relative error of $\Phi_\calE^{(k)}$ to the unitary $k$-design is bounded by~\cite{Schuster:2024ajb}
\begin{align}
    &\epsilon \leq \frac{2^{2Nk}}{k!}\left(1 + \frac{k^2}{2^N}\right)\infnormsm{\varrho_\calE^{(k)} - \varrho_0^{(k)}} + \frac{k^2}{2^N} \nn \\
    &\leq \frac{2^{2Nk}(1 + \frac{k^2}{2^N})}{k!}  \left( \infnormsm{\varrho_\rmI^{(k)} - \varrho_0^{(k)}} + \infnormsm{\varrho_\rmII^{(k)}} + \infnormsm{\varrho_\rmIII^{(k)}} \right)+ \frac{k^2}{2^N},
\end{align}
where $\varrho_\calE^{(k)} = [\Phi^{(k)}_\calE \otimes \id ](P_{\EPR}) := \varrho_\rmI^{(k)} + \varrho_{\rmII}^{(k)} + \varrho_{\rmIII}^{(k)}$, and $\varrho_0^{(k)} = 2^{-2Nk}\sum_\sigma r(\sigma) \otimes r(\sigma)$.
We divide the summation over $T$ in the $k$-fold channel associated with the two-layer Clifford unitaries (Eq.~\eqref{eq:channel_uniformity}) into three parts, giving rise to three contributions to $\varrho_\calE^{(k)}$:
\begin{align}
    \varrho_\rmI^{(k)} &= \frac{1}{2^{2Nk}} \sum_{\vec{\sigma}_{1,2},\vec{\tau}_{1,2}}\sum_{\pi \in S_k} \Wg(\vec{\sigma}_1,\vec{\tau}_1)\Wg(\vec{\sigma}_2,\vec{\tau}_2) \bbrakket{\vec{\tau}_1}{\pi}\bbrakket{\pi}{\vec{\sigma}_2} r(\vec{\sigma}_1)\otimes r(\vec{\tau}_2) \nn \\
    &= \frac{1}{2^{2Nk}} \sum_{\pi\in S_k}r(\pi) \otimes r(\pi) = \varrho_0^{(k)}, \\
    \varrho_\rmII^{(k)} &= \frac{1}{2^{2Nk}} \sum_{\vec{\sigma}_{1,2},\vec{\tau}_{1,2}} \sum_{T \notin S_k} \Wg(\vec{\sigma}_1,\vec{\tau}_1)\Wg(\vec{\sigma}_2,\vec{\tau}_2) \bbrakket{\vec{\tau}_1}{T}\bbrakket{T}{\vec{\sigma}_2} r(\vec{\sigma}_1)\otimes r(\vec{\tau}_2), \\
    \varrho_\rmIII^{(k)} &= \frac{1}{2^{2Nk}} \sum_{\vec{\sigma}_{1,2},\vec{\tau}_{1,2}}\sum_{\vec{T}_1,\vec{T}_2} f(\vec{T}_1,\vec{T}_2)\Wg(\vec{\sigma}_1,\vec{\tau}_1)\Wg(\vec{\sigma}_2,\vec{\tau}_2) \bbrakket{\vec{\tau}_1}{\vec{T}_1}\bbrakket{\vec{T}_2}{\vec{\sigma}_2} r(\vec{\sigma}_1)\otimes r(\vec{\tau}_2).
\end{align}
The first and the second part are from the spatially aligned $T \in S_k$ and $T\notin S_k$, respectively.
According to Theorem~\ref{thm:uniformity_log_depth}, the two-layer Clifford unitaries satisfies the uniformity condition up to error $\sum_{\vec{T}_1,\vec{T}_2}\abs{f(\vec{T}_1,\vec{T}_2)} < \epsilon' = N2^{-\xi-\log\xi + O(k^2)}$.

The $\infty$-norm of the second part is bounded
\begin{align}
    \infnormsm{\varrho_\rmII^{(k)}} \leq& 2^{-2Nk} 2^{O(k^2)} \left[\prod_{x,i} 2^{-N_{x,i}} k!(1 + k^2 2^{-N_{x,i}})\right] \left[\prod_{x',i'}2^{-N_{x',i'}} k!(1 + k^2 2^{-N_{x',i'}}) \right] \nn \\
    \leq& 2^{-2Nk-2N+O(k^2)}\left[ k!(1 + k^2 2^{-\ell})\right]^{2N/\ell},
\end{align}
where $\ell = \min\{N_{x,i},N'_{x',i'}\}$ is the size of the smallest qubit cluster.

The $\infty$-norm of the third part has an upper bound, assuming $\ell \geq k \log k$,
\begin{align}
    \infnormsm{\varrho_\rmIII^{(k)}} \leq 2^{-2Nk} \sum_{\vec{T}_1,\vec{T}_2} \abs{f(\vec{T}_1,\vec{T}_2)} \leq 2^{-2Nk} \epsilon'.
\end{align}
The derivation is similar to that for Eq.~\eqref{eq:rel_state_rhoIII_infnorm}.

Altogether, the relative error has an upper bound
\begin{align}
    \epsilon
    \leq& 2^{O(k^2)-2N}\left[ k!(1 + k^2 2^{-\ell})\right]^{2N/\ell} + N2^{-\xi-\log\xi + O(k^2)}+ k^22^{-N},
\end{align}
Here, one can absorb $k^22^{-N}$ into $N2^{-\xi-\log\xi + O(k^2)}$.
\end{proof}
We remark that, similar to the case of relative-error state design, one can remove a finite fraction of Haar random gates within each qubit block in our construction while maintaining a bounded relative error.

\begin{corollary}[Circuit depth]\label{thm:unitary_relative_circuit_depth}
In an $N$-qubit system with $N=\omega(k^2+\log(1/\epsilon))$, one can realize approximate unitary and state $k$-design with $\epsilon$ relative error in depth\footnote{We acknowledge Thomas Schuster for explaining the efficient constructions of global random Clifford gates, which helps us better appreciate our results.}
\begin{itemize}
\item $O(\log(N/\epsilon)) + 2^{O(k\log k)}$ with 1D circuits.
\item $O(\log\log(N/\epsilon)) +2^{O(k\log k)}$ with all-to-all circuits using $O(N(\log(N/\epsilon)+k^2))$ ancillas.
\end{itemize}

\end{corollary}
\begin{proof}
The first statement follows from (1) any Clifford unitary over $\xi$ qubits can be realized in depth $d = O(\xi)$~\cite{Bravyi:2020xdv}; (2) any unitary operation over $\ell$ qubits can be realized using $O(\ell 4^\ell)$ two-local gates~\cite{Knill:1995kz}. Hence, depth $O(\ell4^\ell)$ is sufficient for all-to-all circuits, depth $O(\ell^24^\ell)$ is sufficient for 1D circuits.

The second statement follows from~\cite{Jiang:2019aom} which shows that any Clifford unitary on $\xi$ qubits can be implemented with all-to-all circuits in depth $O(\log\xi)$ using $\xi^2$ ancillas.
\end{proof}

\subsection{No-go theorems}\label{sec:rel_error_nogo}

Here, we prove no-go theorems for designs with bounded relative error.
We first consider an ensemble of Clifford augmented states generated by applying Clifford unitaries to initial states with low entanglement.
A subclass of these states are the Clifford augmented matrix product states (CAMPS)~\cite{qian2023augmenting,qian2024augmenting}.
CAMPSs can be efficiently simulated and have the potential to host large-scale entanglement and an extensive amount of magic.
However, we show that an ensemble of CAMPSs with bond dimension $D = \tilde{o}((N/\epsilon)^{1/4})$ cannot form a state $k$-design with $\epsilon$ relative error for $k\geq 4$.
We further show that in any dimension, finite-depth local unitaries followed by a Clifford unitary cannot generate a state design with bounded relative error.
The no-go theorem also suggests that a unitary ensemble $\{V_1UV_2\}$ with $V_{1,2}$ being Clifford unitaries and $U$ being a generic shallow unitary circuit cannot form a relative-error unitary $k$-design.
The no-go theorems challenge the naive expectation that magic placed at the initial layer creates more randomness as the Clifford unitary can spread and scramble magic.

The main idea behind the no-go theorem is that the ensemble of Clifford augmented states is distinguished from the Haar random ensemble by measuring the Pauli $4$-th moment, whose logarithm defines the stabilizer Rényi entropy~\cite{Leone:2021rzd}. 
The $4$-th moment is given by the expectation value of a projection operator onto a CSS code defined in the replicated space (see Appendix~\ref{sec:stabilizer_renyi_entropy}), which are positive operators that belong to the Clifford commutant. 
The relative error of the $4$-th moment sets a lower bound on the relative error of the state $4$-design.
We note that although such states cannot form relative-error state designs, they can form additive-error state designs, as we show in Appendix~\ref{sec:state_additive_initial_magic}.

We consider an ensemble $\calE = \{V\ket{\psi_0}\}$ of Clifford augmented states generated by Clifford unitaries $V$ acting on short-range entangled states or one-dimensional matrix product states $\ket{\psi_0}$.
We show that the ensemble is distinguished from the ensemble of Haar random states by a positive observable:
\begin{equation}
r(T_S)=2^{-N}\sum_{P\in\calP}P^{\otimes 4},\qquad \calP=\{Z^{\boldu}X^{\boldv},\boldu,\boldv\in\mathbb{F}_2^{N}\}.\label{eq:rTs}
\end{equation}
The expectation value of this operator is related to the second stabilizer Rényi entropy $\calM_2(\psi)$ of state $\ket{\psi} = V\ket{\psi_0}$,
\begin{align}
    \Tr [r(T_S)\rho_{\calE}^{(4)} ] = \mathbb{E}_{\psi\sim\calE} \frac{1}{2^N} \sum_P \bra{\psi} P\ket{\psi}^4 = \mathbb{E}_{\psi\sim\calE} 2^{-\calM_2(\psi)}.
\end{align}
Note that the operator $r(T_S)$ belongs to the Clifford commutant; its expectation value is invariant under Clifford unitaries:
\begin{equation}
\bra{\psi_0}^{\otimes 4}V^{\otimes 4}r(T_S)V^{\dagger\otimes 4}\ket{\psi_0}^{\otimes 4}=\bra{\psi_0}^{\otimes 4}r(T_S)\ket{\psi_0}^{\otimes 4}.
\end{equation}

\begin{theorem}[Lower bound on relative error from non-maximal initial entanglement]
Consider disjoint subsystems $M_i$ of size $N_i$ in an $N$-qubit system, where the union $\cup_i M_i$ need not contain all the qubits.

The state ensemble $\calE=\{V\ket{\psi_0}\}$ generated by Clifford unitary $V$ acting on state $\ket{\psi_0}$ cannot form an approximate state $k$-design (for $k\ge4$) with relative error 
\begin{equation}
\epsilon\leq \frac{1}{24}\left[\Ens_{\calE}\sum_{i}2^{-2\calS^{(2)}_i}(1-2^{\calS^{(2)}_i-N_i})-3\right],
\end{equation}
where $\calS^{(2)}_i=-\log \Tr(\rho_i^2)$ is the second Rényi entropy evaluated on subsystem $M_i$ of the initial state $\ket{\psi_0}$, and $\rho_i$ is the reduced density matrix of $\ket{\psi_0}$ on $M_i$.

\end{theorem}

\begin{proof}
For each member of the state ensemble, we can lower bound the expectation value of $r(T_S)$ by the following procedure. On a single region $M_i$, the summation over all the Pauli operators that support on $M_i$ gives
\begin{equation}
\delta := \sum_{P_i}|\bra{\psi_0}P_i\ket{\psi_0}|^2=2^{N_i}\Tr(\rho_{M_i}^2) = 2^{N_i -\calS_i^{(2)}}.
\end{equation}
We can view $\{\chi_{P_i}=\delta^{-1}|\bra{\psi_0}P_i\ket{\psi_0}|^2\}$ as a probability distribution. Since $\chi_{I}=\delta^{-1}$, the contributions from non-identity Pauli operators are lower bounded by
\begin{equation}
2^{-N}\sum_{P_i\neq I}|\bra{\psi_0}P_i\ket{\psi_0}|^4=2^{-N}\delta^2\sum_{P_i\neq I}|\chi_{P_i}|^2 \geq 2^{-N}\delta^2\frac{1-\delta^{-1}}{2^{2N_i}-1} 2^{-N}2^{-2\calS^{(2)}_i}(1-2^{\calS^{(2)}_i-N_i}).
\end{equation}
Summing over all regions $i$, we get
\begin{equation}
\bra{\psi_0}^{\otimes 4}r(T_S)\ket{\psi_0}^{\otimes 4}\geq 2^{-N}\left[1+\sum_{i}2^{-2\calS^{(2)}_i}(1-2^{\calS^{(2)}_i-N_i})\right].
\end{equation}
Averaging over the state in the ensemble $\calE$, we have
\begin{equation}\label{eq:nogo_clifford_augmented_states_SRE_bound}
\Tr[r(T_S)\rho_{\calE}^{(4)}]\geq 2^{-N}\left[1+\Ens_{\calE}\sum_{i}2^{-2\calS^{(2)}_i}(1-2^{\calS^{(2)}_i-N_i})\right].
\end{equation}
The Haar value is~\cite{Leone:2021rzd}
\begin{equation}
\begin{aligned}
\Tr[r(T_S)\rho_{\H}^{(4)}]&=2^{-N}D_{N,4}^{-1}\sum_{\pi\in S_4,P\in \calP}\Tr[ P^{\otimes 4}r(\pi)]=\frac{4}{2^N+3}.
\end{aligned}
\end{equation}
The difference in the expectation value can give a lower bound to the relative error
\begin{equation}\label{eq:nogo_rel_error_bound_from_expectation_value}
\lVert\rho_{\calE}^{(4)}-\rho_{\H}^{(4)}\rVert_\infty\geq\frac{\left|\Tr[r(T_S)(\rho_{\calE}^{(4)}-\rho_{\H}^{(4)})]\right|}{\lVert r(T_S)\rVert_1} 
\geq 2^{-4N}\left[\Ens_{\calE}\sum_{i}2^{-2\calS^{(2)}_i}(1-2^{\calS^{(2)}_i-N_i})-3\right],
\end{equation}
where we use the H\"older's inequality, and $\lVert r(T_S)\rVert_1=2^{3N}$ (defect subspace dimension $\dim\rmN = 1$, and $k = 4$). 
The relative error is bounded by
\begin{equation}\label{eq:nogo_clifford_augmented_states_infnorm_bound}
\frac{\lVert\rho_{\calE}^{(4)}-\rho_{\H}^{(4)}\rVert_\infty}{\lVert\rho_{\H}^{(4)}\rVert_\infty}\geq\frac{1}{24}\left[\Ens_{\calE}\sum_{i}2^{-2\calS^{(2)}_i}(1-2^{\calS^{(2)}_i-N_i})-3\right].
\end{equation}

If $\calE$ is an approximate state $k$-design with $\epsilon$ relative error for $k\geq 4$, then $(1-\epsilon)\rho_{\H}^{(k)}\preceq\rho_{\calE}^{(k)}\preceq(1+\epsilon)\rho_{\H}^{(k)}$. Taking the partial trace over $k-4$ copies of the system, we find $(1-\epsilon)\rho_{\H}^{(4)}\preceq\rho_{\calE}^{(4)}\preceq(1+\epsilon)\rho_{\H}^{(4)}$. In other words, $\calE$ is an approximate state $4$-design with $\epsilon$ relative error. Since $\rho_{\H}^{(4)}$ has a flat spectrum, the above implies
\begin{equation}
\lVert\rho_{\calE}^{(4)}-\rho_{\H}^{(4)}\rVert_\infty\le\epsilon\lVert\rho_{\H}^{(4)}\rVert_\infty.
\end{equation}
Therefore an ensemble cannot form an approximate state $k$-design for $k \geq 4$ with relative error $\epsilon$ smaller than the right-hand side of~\eqref{eq:nogo_clifford_augmented_states_infnorm_bound}.

\end{proof}

The above proof can be understood as a classical statistical mechanics problem, where the lack of correlations proliferates domain wall excitations. We interpret $\left|\bra{\psi}\bigotimes_i P_i\ket{\psi}\right|^4\equiv W(\{P_i\})$ as the Boltzmann weight of a spin model, where we put a spin on every region $M_i$ that are labeled by the Pauli operators on $M_i$. We ask whether the partition function $Z=\sum_{\{P_i\}}W(\{P_i\})$ is relatively close to the Haar value that is $O(1)$. 

The identity operator $\prod_iI_i$ gives the maximal weight $W(\{I_i\})=1$. On Haar random states, the weights for any spin configuration with non-identity inputs are heavily suppressed, that is to say, if $\{P_i\}\neq \{I_i\}$, then
\begin{equation}\label{eq:nogo_clifford_augmented_states_stat_mech_haar_non_identity}
\Ens_{\psi\sim\H}W(\{P_i\})=D_{N,2}^{-1}\sum_{\pi\in S_4}\Tr[r(\pi) \bigotimes_i P_i^{\otimes 4}] =3\times 2^{-N}[1+O(2^{-N})].
\end{equation}
However, in states that are less entangled, there might be domain wall excitations that corresponds to $\{P_i\}$ with $I$ on all but one site. For every site $M_i$, such domain wall gives a contribution of at least $2^{-2\calS^{(2)}_i}(1-2^{\calS^{(2)}_i-N_i})$, which could be an $O(1)$ number if $N_i=O(1)$ and $N_i-\calS^{(2)}_i=O(1)$. The Boltzmann weights are also $O(1)$, in contrast with the Haar case~\eqref{eq:nogo_clifford_augmented_states_stat_mech_haar_non_identity}. Since there is an $O(N)$ number of such regions, these configurations gives an overall $O(N)$ contribution to the partition function. In stat mech language, the energy cost of these domain walls cannot suppress the large entropy, so they will proliferate.

As the first application of the relative error lower bound, we prove a no-go theorem for ensembles of Clifford augmented matrix product states. 

\begin{corollary}[No-go theorem for Clifford augmented MPS]
Consider an ensemble of states $\calE = \{V\ket{\psi}\}$, where $V$ is a Clifford unitary, and $\ket{\psi}$ are matrix product states whose bond dimensions are bounded by $D$.
\begin{equation*}
\includegraphics[width=6cm]{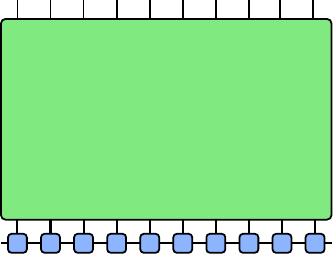}
\end{equation*}
The states cannot form an approximate state $k$-design for $k\geq 4$ with $\epsilon$ relative error for\footnote{To be more precise, the condition is $D=o\Big(\Big(\frac{N/\epsilon}{\log(N/\epsilon)}\Big)^{1/4}\Big)$. We used the $\tilde{o}$ notation to suppress factors of $\log(N/\epsilon)$.} $D=\tilde{o}((N/\epsilon)^{1/4})$.
\end{corollary}
\begin{proof}
Since the bond dimension is upper bounded by $D$, the subsystem Rényi entropy is bounded by $2^{\calS^{(2)}_i}\leq D^2$. Suppose we choose a region $M_i$ such that $N_i=\lceil 2\log D+1\rceil$, then $N_i-\calS^{(2)}_i\geq 1$. There exists a choice of the regions $M_i$ such that the total number of such region with $N_i=\lceil 2\log D+1\rceil$ is lower bounded by $\lfloor \frac{N}{2\log D+2}\rfloor\geq \frac{N}{2\log D+2}-1$. That is to say, the relative error is lower bounded by
\begin{equation}\label{eq:nogo_clifford_augmented_MPS_SRE_bound}
\epsilon\geq\frac{1}{24}\left[\left(\frac{N}{2\log D+2}-1\right)\frac{1}{2D^4}-3\right].
\end{equation}
\end{proof}

Our no-go theorem suggests that CAMPSs with bond dimension $\tilde{o}((N/\epsilon)^{1/4})$ cannot form a relative error $k$-design with $k\geq 4$. 
This implies that CAMPSs cannot represent the output state of highly random Brownian dynamics or the dynamics of generally time-dependent Hamiltonians~\cite{Onorati:2016met,Jian:2022pvj} that form relative-error unitary designs.
On the other hand, an ensemble of MPS with bond dimension $D = O(N^k)$ can form an approximate state $k$-design. 
This is because a state $k$-design can be generated in random circuits of depth $d = O(k\log N)$~\cite{Schuster:2024ajb}, whose output states have an exact representation as MPSs with bond dimension $D = O(N^k)$.
The representation power of CAMPSs with a bond dimension in between remains an open problem.

The no-go theorem can be extended to Clifford augmented tensor network states in higher dimensions. 
For example, for projected entangled pair states (PEPS) of bond dimension $D$ in two dimensions~\cite{cirac2021matrix}, we can choose a $l \times l$ square region with $l = \lceil 4\log D + 1 \rceil$.
The relative error would have a similar lower bound
\begin{align}
    \epsilon \geq \frac{1}{24}\left[ \left(\frac{L}{4\log D + 2} - 1\right)^2\frac{1}{2D^{8(4\log D + 2)}} - 3\right].
\end{align}
in a $L\times L$ system. Hence, one cannot achieve $\epsilon$ relative error with a constant bond dimension.

This result gives an upper bound on stabilizer Rényi entropy of Clifford augmented states: taking the log of~\eqref{eq:nogo_clifford_augmented_MPS_SRE_bound}, we find $\calM_2\leq N-O(\log N)$, suggesting that these states have at least logarithmic deviation from the maximum stabilizer Rényi entropy. 

\begin{corollary}[No-go theorem for states prepared by Clifford augmented FDLU]
\label{thm:nogo_clifford_augmented_FDLU}
Consider an ensemble of states $\{VU\ket{0}^{\otimes N}\}$, where $V$ are Clifford unitaries, and $U$ are $d$-dimensional local circuits with circuit depth bounded by $t$. The states cannot form an approximate state $k$-design for $k\geq 4$ with $\epsilon$ relative error for $t=\tilde{o}((\log(N/\epsilon))^{1/d})$, with $d=O(1)$.
\end{corollary}
\begin{proof}
Consider a hyper-cubic subsystem $M_i$ (with each side of length $l$) that contains $N_i=l^d$ qubits. The boundary of $M_i$ has area $2d l^{d-1}$. 
For states prepared by local unitary circuits with depth $t$, the entanglement entropy of subsystem $M_i$ has an upper bound 
\begin{equation}
\calS^{(2)}_i(t)\leq \calS_i(t)\leq 4 d l^{d-1} t,
\end{equation}
where $S_i(t)$ is the von-Neumann entropy of $M_i$. 

We consider the subsystems with $l = \lceil 4dt + 1\rceil$, so $\calS^{(2)}_i(t)\leq l^d-1$.
In a $d$-dimensional system of size $\underbrace{L\times\cdots\times L}_d$, the relative error has a lower bound
\begin{align}
    \epsilon \geq \frac{1}{24}\left[\left(\frac{L}{4dt + 2}-1\right)^d\frac{1}{2\cdot 2^{8dt(4dt +2)^{d-1}}}-3 \right].
\end{align}
Hence, the ensemble cannot form a state $k$-design with $\epsilon$ relative error for $t^d \log t = o(\log(N/\epsilon))$.
In particular, one cannot achieve a small $\epsilon$ relative error with $t=O(1)$.
\end{proof}

The no-go theorem for state $k$-design with bounded relative errors also implies a no-go theorem for unitary $k$-design with relative errors (Corollary~\ref{corollary:nogo_unitary_relative}). We start with the Lemma below.
\begin{lemma}[Choi state of relative-error unitary design is relative-error state design\footnote{We thank Thomas Schuster for insightful discussion that helps us improve this statement.}]\label{lemma:unitary_design_implies_state_design}
Let $\{U\}$ be an ensemble on $N$-qubit unitary $U$ that forms an approximate unitary $k$-design with relative error $\epsilon'$. 
Let $\ket{\EPR}$ be an EPR state between two copies of the system.
Then, the state ensemble $\calE = \{U \otimes I\ket{\EPR}\}$ is an approximate state $k$-design with relative error $\epsilon\le(1+\epsilon')(1+k^22^{-N})(1+k^22^{-2N})-1$, for $k^2\le2^N$.
\end{lemma}
\begin{proof}
The $k$-th moment of the state ensemble $\calE$ is $\rho_{\calE}^{(k)}=\varrho^{(k)} := [\Phi^{(k)}\otimes \id] (P_{\EPR})$, where $\Phi^{(k)}$ is the $k$-fold channel of the unitary ensemble $\{U\}$.
The unitary ensemble being a $k$-design with relative error $\epsilon'$ implies that $(1-\epsilon')\varrho_{\H}^{(k)}\preceq\varrho^{(k)}\preceq(1+\epsilon')\varrho_{\H}^{(k)}$.
As shown in~\cite{Schuster:2024ajb}, $\varrho_{\H}^{(k)}$ can be approximated by $\varrho_0^{(k)}=2^{-2Nk}\sum_{\sigma,\tau}r(\sigma)_L\otimes r(\tau)_R$, that is to say, $(1-k^22^{-N})\varrho_0^{(k)}\preceq\varrho_{\H}^{(k)}\preceq(1+k^22^{-N})\varrho_0^{(k)}$.
Therefore,
\begin{equation}\label{eq:nogo_choi_state_rho0}
(1-k^22^{-N})(1-\epsilon')\varrho_0^{(k)}\preceq\rho_\calE^{(k)}\preceq\varrho_0^{(k)}(1+\epsilon')(1+k^22^{-N}).
\end{equation}

Our goal is to bound the relative error between $\rho_\calE^{(k)}$ and the $k$-th moment of the Haar random state on the doubled system (the left system and the right system) $\rho_{\H,LR}^{(k)}=D_{2N,k}^{-1}\sum_{\pi}r(\pi)_L\otimes r(\pi)_R$.
The only difference between $\rho_{\H,LR}^{(k)}$ and $\varrho_0^{(k)}$ is the overall factors, which are relatively close to each other: $2^{-Nk}D_{N,k}-1\leq k^22^{-2N}$, so $(1-k^22^{-2N})\rho_{\H,LR}^{(k)}\preceq \varrho_0^{(k)}\preceq(1+k^22^{-2N})\rho_{\H,LR}^{(k)}$. 
Combined with~\eqref{eq:nogo_choi_state_rho0}, we conclude
\begin{equation}
(1-\epsilon)\rho_{\H,LR}^{(k)}\preceq\rho_{\calE}^{(k)}\preceq (1+\epsilon)\rho_{\H,LR}^{(k)},
\end{equation}
with $\epsilon=(1+\epsilon')(1+k^22^{-N})(1+k^22^{-2N})-1$.
\end{proof}

\begin{corollary}[No-go theorem for Clifford sandwiched FDLU]\label{corollary:nogo_unitary_relative}
Let $V_1$ and $V_2$ be Clifford unitaries, and $U$ be d-dimensional local circuits with depth bounded by $t$.
\begin{equation*}
\includegraphics[width=6cm]{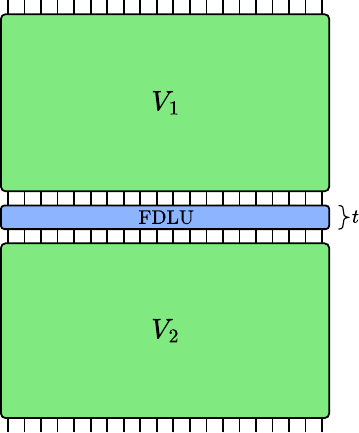}
\end{equation*}
Then, the unitary ensemble $\{V_1 U V_2\}$ does not form a unitary $k$-design for $k\geq 4$ with $\epsilon$ relative error for $t=\tilde{o}((\log(N/\epsilon)^{1/d})$, with $d=O(1)$.
\end{corollary}
\begin{proof}
Consider the state ensemble $\calE = \{\big((V_1 U V_2)_L \otimes I_R\big)\ket{\EPR}\}$.
Each state $\ket{\psi} \in \calE$ is a Clifford augmented short-range entangled state, i.e.
\begin{equation}
\ket{\psi}=\big((V_1 U V_2)_L\otimes I_R\big)\ket{\EPR} =\big((V_{1})_L\otimes (V_{2}^{\top})_R\big) (U_L\otimes I_R)\ket{\EPR}.
\end{equation}
Since $\ket{\EPR}$ can be prepared from $\ket{0}^{\otimes 2N}$ in a depth-$1$ circuit, $(U_L\otimes I_R)\ket{\EPR}$ can be prepared using a unitary circuit of depth $t' = t+1$ in a system of $2N$ qubits.
According to Corollary~\ref{thm:nogo_clifford_augmented_FDLU}, the relative error of this state ensemble from a $k$-design has a lower bound
\begin{align}
    \epsilon \geq \frac{1}{24}\left[\left(\frac{L}{4dt' + 2}-1\right)^d\frac{1}{2\cdot 2^{8dt'(4dt' +2)^{d-1}}}-3 \right],
\end{align}
where we assume the system is of size $N = L^d$ with each dimension of size $L$.
Hence, according to Lemma~\ref{lemma:unitary_design_implies_state_design}, the relative error $\epsilon'$ between the unitary ensemble and a $k$-design for $k \geq 4$ has a lower bound
\begin{align}
    \epsilon' \geq \frac{\frac{1}{24}\left[\left(\frac{L}{4dt' + 2}-1\right)^d\frac{1}{2\cdot 2^{8dt'(4dt' +2)^{d-1}}}-3 \right] + 1}{(1+16\cdot 2^{-N})(1+16\cdot2^{-2N})} - 1.
\end{align}
For a depth $t$ being a constant, the relative error $\epsilon'$ is unbounded.
\end{proof}

\section{Designs with bounded additive error}\label{sec:add_error}
In this section, we construct state and unitary designs with bounded additive errors using Clifford circuits augmented with magic gates.

\subsection{State designs with bounded additive errors}
\label{sec:state_additive}
We prove that two layers of Clifford gates followed by $\tilde{O}(k^2)$ single-qubit magic gates can generate a state $k$-design with bounded additive error.
A physical understanding of our result is provided in Sec.~\ref{sec:stat_mech_state_additive} by mapping the frame potential to statistical mechanics models, where magic plays the role of a symmetry-breaking field.
We also prove in Appendix~\ref{sec:state_additive_initial_magic} that one can alternatively place random unitaries drawn from $k$-design before the Clifford gates. 
To achieve $\epsilon$ additive error in this case, one needs to apply a product of $k$-design unitaries over $O(k^2+\log(1/\epsilon))$ qubits, while each unitary acts on disjoint qubit clusters of size $O(k\log k)$.

To prove that our construction forms a state design, we first show that one can obtain additive-error state $k$-designs over $N$ qubits by applying single-qubit magic gates to global random stabilizer states. 
Our proof uses the upper bound of the additive error for state designs in terms of the frame potential
\begin{equation}
    \lVert \rho_\calE^{(k)}-\rho_{\H}^{(k)}\rVert_1^2 \leq 2^{N k}\big[F_{\calE}^{(k)}-F_{\H}^{(k)}\big].
\end{equation}
We prove a bounded additive error by providing an upper bound on the frame potential difference.

\begin{lemma}[Additive-error state design from applying constant magic to random stabilizer states]\label{thm:state_additive_linear_depth_final_magic}
Consider an $N$-qubit state $\ket{\psi}$ drawn from an approximate stabilizer state $k$-design with $\epsilon'$ additive error. 
Apply $N_M$ single-qubit unitary gates, $U_i$ with $i = 1,\cdots, N_M$, to $N_M$ distinct qubits. 
Assume $N\geq k+\log k-2$. 
\begin{equation*}\label{fig:state_additive_linear_depth}
\includegraphics[width=8.9cm]{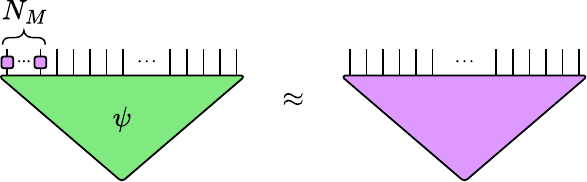}    
\end{equation*}
Then, the state ensemble $\calE=\{U_1\cdots U_{N_M}\ket{\psi}\}$ forms an approximate state $k$-design with $\epsilon$ additive error.
\begin{itemize}
\item For $U_i$ drawn from (exact) unitary $k$-designs,
\begin{align}
\epsilon&\leq 2^{O(k^2)}\left(\frac{7}{8}\right)^{\frac{N_M}{4}} + \varepsilon'.
\end{align}
\item For $U_i$ drawn from $\{\mathds{1}, u, u^\dagger\}$, 
\begin{align}
\epsilon&\leq 2^{O(k^2)}\left(1 - c\log^{-2}k\right)^{\frac{N_M}{2}} + \varepsilon',
\end{align}
where $c$ is a constant depending on the magic gate $u$.
\end{itemize}

\end{lemma}
\begin{proof}
First, the additive error of the $k$-th moment of the circuit ensemble can be upper bounded by the triangular inequality
\begin{equation}
\begin{aligned}
&\lVert \Phi^{(k)}_{u}\big(\rho_{\C,\epsilon'}^{(k)}\big)-\rho_{\H}^{(k)}\rVert_1 \leq \lVert \Phi^{(k)}_{u}\big(\rho_{\C}^{(k)}\big)-\rho_{\H}^{(k)}\rVert_1 + \lVert \Phi^{(k)}_{u}\big(\rho_{\C,\epsilon'}^{(k)}\big)-\Phi^{(k)}_{u}\big(\rho_{\C}^{(k)}\big)\rVert_1 \\
\leq& \lVert \Phi^{(k)}_{u}\big(\rho_{\C}^{(k)}\big)-\rho_{\H}^{(k)}\rVert_1 + \varepsilon',\label{eq:addtive_error_state_triangle}
\end{aligned}
\end{equation}
where $\Phi_u^{(k)}(\cdot)$ is the $k$-fold channel associated with the single-qubit gate set that $U_i$ is drawn from, $\rho_{\C,\epsilon'}^{(k)}$ and $\rho_{\C}^{(k)}$ are the $k$-th moments of the stabilizer state $k$-design with additive error $\epsilon'$ and an exact stabilize state $k$-design, respectively.
For the second inequality, we use the monotonicity of 1-norm under quantum channels, 
\begin{align}
\lVert \Phi^{(k)}_{u}\big(\rho_{\C}^{(k)}\big)-\Phi^{(k)}_{u}\big(\rho_{\C,0}^{(k)}\big)\rVert_1 \leq \lVert\rho_{\C}^{(k)}-\rho_{\C,0}^{(k)}\rVert_1 \leq \varepsilon'.
\end{align}

Next, we bound the distance between $\Phi^{(k)}_{u}\big(\rho_{\C}^{(k)}\big)$ and $\rho_{\H}^{(k)}$ using the $k$-th frame potential.
Given that $F_\calE^{(k)} \geq F_{\H}^{(k)}$, we here provide an upper bound for $F_{\calE}^{(k)}$.
Specifically,
\begin{equation}
\begin{aligned}
F_{\calE}^{(k)}&=\Tr\big[\Phi^{(k)}_{u}\big(\rho_{\C}^{(k)}\big)\big]^2=\Tr\big[\rho_{\C}^{(k)}\Phi^{(k)}_{u}\big(\rho_{\C}^{(k)}\big)\big]=\frac{1}{Z_{n,k}^2}\sum_{T_1,T_2}\bbrakket{T_1}{T_2}^{N-M}\bbra{T_1}\hat{\Phi}_{u}^{(k)}\kket{T_2}^{M}.
\end{aligned}
\label{eq:frame_potential_E}
\end{equation}
where $\hat{\Phi}_{u}^{(k)}=\Ens_{U\sim \calE_u}U^{\otimes k}\otimes U^{*\otimes k}$ with $\calE_u$ being either $\{\mathds{1},u,u^\dagger\}$ or an exact $k$-design, and $\kket{T_1}$ and $\kket{T_2}$ are over single qubits.

A key step of our proof relies on the bound on the overlap $\bbra{T_1}\hat{\Phi}_{u}^{(k)}\kket{T_2}$~\cite{Haferkamp:2020qel}. 
In the case that single-qubit gates are drawn from $\{\mathds{1},u,u^\dagger\}$ (Lemma 2 of Ref.~\cite{Haferkamp:2020qel}), there exists a constant $c>0$ such that for either $T_1$ or $T_2 \notin S_k$
\begin{align}
\abs{\bbra{T_1}\hat{\Phi}_u^{(k)}\kket{T_2}} \leq \left(1 - c\log^{-2}k\right)2^k.\label{eq:overlap_bound_magic}
\end{align}
In the case that single-qubit gates are drawn from an exact $k$-design (Lemma 13 of Ref.~\cite{Haferkamp:2020qel}),
we have for $T \notin S_k$,
\begin{equation}
\bbra{T}\hat{\Phi}_{\H}^{(k)}\kket{T}\leq \frac{7}{8} 2^k.
\end{equation}
When $T_1$ is not a permutation, the Cauchy-Schwartz inequality indicates that
\begin{equation}
|\bbra{T_1}\hat{\Phi}_{\H}^{(k)}\kket{T_2}|\le\sqrt{\bbra{T_1}\hat{\Phi}_{\H}^{(k)}\kket{T_1}\bbrakket{T_2}{T_2}}\leq \sqrt{\frac{7}{8}}2^k.\label{eq:overlap_bound_k-design}
\end{equation}
We note that for $T_1 = T_2 = \pi\in S_k$ being a permutation, $\bbra{\pi}\hat{\Phi}_{u}^{(k)}\kket{\pi}=2^k$.
Thus, inserting a single-qubit magic gate suppresses the overlap $\bbra{T_1}\hat{\Phi}_{u}^{(k)}\kket{T_2}$ when either $T_1$ or $T_2$ is not a permutation.

We divide the frame potential in Eq.~\eqref{eq:frame_potential_E} into two parts.
The contributions from $T_1,T_2\in S_k$ are 
\begin{equation}
\begin{aligned}
F_\rmI^{(k)} :=& \frac{1}{Z_{N,k}^2}\sum_{\sigma,\tau\in S_k}\bbrakket{\sigma}{\tau}^{N-N_M}\bbra{\sigma}\hat{\Phi}_{u}^{(k)}\kket{\tau}^{N_M}=\frac{1}{Z_{N,k}^2}\sum_{\sigma,\tau\in S_k}\bbrakket{\sigma}{\tau}^{N}=\frac{k! D_{N,k}}{Z_{N,k}^2},
\end{aligned}
\end{equation}
where $D_{N,k} = \frac{(2^N+k-1)!}{(2^N-1)!}$.
The deviation of $F_\rmI^{(k)}$ from the frame potential $k!2^{-N k}$  of the Haar random state is bounded by 
\begin{equation}
\left|\frac{k!D_{N,k}}{Z_{N,k}^2}-\frac{k!}{2^{N k}}\right|\le2^{O(k^2)-N-kN},
\end{equation}
where we use an upper bound $Z_{N,k}\leq 2^{Nk}(1 + k 2^{k-N})$ for $N\geq k+\log k-2$~\cite{Haferkamp:2020qel}.

The contribution from either $T_1\notin S_k$ or $T_2\notin S_k$ has an upper bound
\begin{equation}
\begin{aligned}
F_\rmII^{(k)} :=& \frac{1}{Z_{N,k}^2}\sum_{T_1 \text{ or } T_2\notin S_k}\bbrakket{T_1}{T_2}^{N-N_M}\bbra{T_1}\hat{\Phi}_{u}^{(k)}\kket{T_2}^{N_M} \leq \frac{2^{N k}}{Z_{N,k}^2}2^{O(k^2)}a^{N_M},
\end{aligned}
\end{equation}
where $0<a<1$ depends on the gate set of single-qubit random gates. Specifically, $a = \sqrt{7/8}$ when the magic gates $U_i$ are drawn from exact $k$-design, $a = 1 - c\log^{-2}k$ for $U_i$ drawn from $\{\mathds{1}, u, u^\dagger\}$ with $u$ being a magic gate.

Putting these together $F_{\calE}^{(k)} = F_\rmI^{(k)}+F_{\rmII}^{(k)}$, we have
\begin{equation}
\begin{aligned}
\lVert \Phi^{(k)}_{u}\big(\rho_{\C}^{(k)}\big)-\rho_{\H}^{(k)}\rVert_1^2 &\leq 2^{N k}\big[F_{\calE}^{(k)}-F_{\H}^{(k)}\big]\le2^{O(k^2)-N}+ 2^{O(k^2)}a^{N_M}\le2^{O(k^2)}a^{N_M}.    
\end{aligned}
\end{equation}
Inserting this result in Eq.~\eqref{eq:addtive_error_state_triangle}, we obtain the upper bound on the additive error $\epsilon$.
\end{proof}

\begin{theorem}[Additive-error state designs from low-depth Clifford unitaries followed by constant magic]\label{thm:state_additive_final_magic}
Apply a two-layer brickwork circuit $V$ to the initial state $\ket{0}^{\otimes N}$, where each gate is drawn form an exact Clifford $k$-design. 
Let $\xi$ be the smallest size of the overlapping regions of the Clifford unitaries.
Then apply $N_M$ single-qubit unitary gates $U_1,\cdots, U_{N_M}$ to $N_M$ distinct qubits, where each gate is independently drawn from either an exact unitary $k$-design or a gate set $\{\mathds{1},u,u^\dagger\}$ with $u$ being a magic gate. 

The state ensemble $\calE=\left(\bigotimes_{i=1}^{N_M}U_i\right)V\ket{0}^{\otimes N}$ forms an approximate state $k$-design for $k< O(\sqrt{N})$ with additive error 
\begin{align}
    \epsilon\leq 2^{O(k^2)}a^{\frac{N_M}{2}} + N2^{-\xi-\log\xi+O(k^2)},
\end{align}
where $a = \sqrt{7/8}$ if the single-qubit gates are drawn from an exact $k$-design, and $a = 1 -c\log^{-2}k < 1$ with the constant $c>0$ if the gate is drawn from $\{\mathds{1},u,u^\dagger\}$. 
In other words, one can achieve $\epsilon$ additive error with $\xi=O(\log(N/\epsilon)+k^2)$ and 
\begin{itemize}
\item $N_M=O(k^2+\log(1/\epsilon))$ for $k$-design random unitary;
\item $N_M = O( (k^2 + \log(1/\epsilon))\log^2 k)$ for single-qubit gates drawn from $\{1, u, u^\dagger\}$.
\end{itemize}
\begin{equation*}
\includegraphics[width=10.8cm]{Figure/state_additive_final_magic.pdf}
\end{equation*}
\end{theorem}
\begin{proof}
A combination of Lemma~\ref{thm:state_additive_linear_depth_final_magic} and Corollary~\ref{thm:state_additive_clifford_design}. We note that the assumption $N \geq k + \log k - 2$ in Lemma~\ref{thm:state_additive_linear_depth_final_magic} is satisfied automatically as we require $N \ge \xi \ge O(k^2)$.
\end{proof}
We note that if the single qubit random unitaries are drawn from approximate $k$-designs with additive error $\epsilon'$, the additive error has an additional term $(1+\epsilon')^{N_M}-1$, which follows from the triangular inequality and the definition of diamond norm. This term is bounded for $N_M = O(k^2+\log(1/\epsilon))$ as long as $\epsilon'$ is small enough.
We also note that in Appendix~\ref{app:gluing_lemma_haar_clifford}, we proved a Lemma showing that one can glue a state $k$-design and a stabilizer state $k$-design using a Clifford $k$-design to form a state $k$-design over a larger region. 
This Lemma combined with Lemma~\ref{thm:state_additive_linear_depth_final_magic} can give an alternative proof of Theorem~\ref{thm:state_additive_final_magic}.

\begin{corollary}[Circuit depth of additive-error state designs]
In an $N$-qubit system, one can realize approximate state $k$-designs with $\epsilon$ additive error using $O(k^2+\log(1/\epsilon))$ magic gates in depth 
\begin{itemize}
\item $O(\log (N/\epsilon)+k^2)$ using 1D circuits.
\item $O(\log\log (N/\epsilon)+\log k)$ using all-to-all circuits with $O(N(\log(N/\epsilon)+k^2))$ ancillas. 
\end{itemize}
\begin{proof}
Similar to the proof of Corollary~\ref{thm:unitary_relative_circuit_depth}.
\end{proof}

\end{corollary}

Our construction of state $k$-designs with additive error only use $O(k^2+\log(1/\epsilon))$ (a system size independent number) of magic gates.
One may wonder why the $k$-th moment of the ensemble cannot be distinguished from that of Haar random states by measuring $r(T_S)$ defined in Sec.~\ref{sec:rel_error_nogo}. 
This observable is related to stabilizer R\'enyi entropy, which measures the magic of pure states.
The reason is that a state design with bounded additive error does not imply a bounded difference between $\Tr r(T_S)\rho_{\H}^{(k)}$ and $\Tr r(T_S)\rho_{\calE}^{(k)}$.
We may attempt to use the H\"older's inequality to obtain an upper bound
\begin{equation}
\left|\Tr[r(T_S)(\rho_{\calE}^{(k)}-\rho_{\H}^{(k)})]\right|\leq \lVert r(T_S)\rVert_\infty\cdot\lVert\rho_{\calE}^{(k)}-\rho_{\H}^{(k)}\rVert_1.
\end{equation}
However, $\lVert r(T_S)\rVert_\infty=2^N$, leading to a very loose bound $2^N\epsilon$.
Furthermore, we note that a small $O(1)$ difference between $\Tr r(T_S)\rho_{\H}^{(k)}$ and $\Tr r(T_S)\rho_{\calE}^{(k)}$ does not mean the two states are close in the stabilizer R\'enyi entropy.
This is because $r(T_S)$ has an expectation value $\Tr r(T_S)\rho_{\H}^{(k)}$ that is exponentially small in $N$ for a $k$-design; an $O(1)$ difference does not imply a bounded error for its logarithm, which is the stabilizer R\'enyi entropy.

\subsection{Unitary designs with bounded additive error}\label{sec:unitary_additive}
We consider the circuit that consists of Haar random gates over $O(k^3)+\log(1/\varepsilon)$ qubits followed by global random Clifford unitaries. 
We prove that the circuit ensemble forms an approximate unitary $k$-design up to $\varepsilon$ additive error.

\begin{theorem}[Additive-error unitary designs from circuits with magic gates over constant-number of qubits]
Consider an $ N$-qubit system with $M$ being a subsystem of size $N_M$. $U$, acting on $M$, is drawn from approximate unitary $k$-design with additive error $\epsilon'$. $V$ is drawn from (exact) Clifford $k$-design acting on all qubits. Then, the ensemble $\calE = \{VU\}$ forms an approximate unitary $k$-design up to additive error $\varepsilon\leq \epsilon'+ 2^{O(k^3)-N_M}(1+k^22^{-N_M})$ for $k \leq N+1$.
\begin{equation*}
\includegraphics[width=5.5cm]{Figure/unitary_additive.pdf}
\end{equation*}
\end{theorem}
\begin{proof}
Let $\bar{M}$ denote the complement of subsystem $M$, with size $N_{\bar{M}} = N - N_M$. 
Let $\Phi_{\C}^{(k)}$ denote the $k$-th moment of random Clifford unitaries on $N$ qubits. 
Let $\Phi_{\H,M,\epsilon'}^{(k)}$ denote the $k$-fold channel of an approximate unitary design acting on $M$. 
Let $\Phi_{\H,M}^{(k)}$ denote the $k$-fold channel of an exact unitary design acting on $M$. 
By assumption $\lVert\Phi_{\H,M,\epsilon'}^{(k)}-\Phi_{\H,M}^{(k)}\rVert_\Diamond\le\epsilon'$.

The $k$-fold channel of our unitary ensemble is $\Phi_{\calE,\epsilon'}^{(k)}=\Phi_{\C}^{(k)}\circ(\Phi_{\H,M,\epsilon'}^{(k)}\otimes \id_{\bar{M}})$. 
We first bound its diamond distance from $\Phi_0^{(k)}=2^{-Nk}\sum_{\pi}|\pi)(\pi|$ using the triangle inequality:
\begin{equation}
\lVert\Phi_{\calE,\epsilon'}^{(k)}-\Phi_0^{(k)}\rVert_\Diamond\leq \lVert\Phi_{\calE,\epsilon'}^{(k)}-\Phi_{\calE,0}^{(k)}\rVert_\Diamond+ \lVert\Phi_{\calE,0}^{(k)}-\Phi_0^{(k)}\rVert_\Diamond,
\end{equation}
where $\Phi_{\calE,0}^{(k)}$ is obtained by replacing $\Phi_{\H,M,\epsilon'}^{(k)}$ with the exact $k$-design $\Phi_{\H,M}^{(k)}$. The first term is bounded by
\begin{equation}
\begin{aligned}
\lVert\Phi_{\calE,\epsilon'}^{(k)}-\Phi_{\calE,0}^{(k)}\rVert_\Diamond&=\max_{\lVert O\rVert_1\leq 1}\lVert\Phi_{\C}\circ(\Phi_{\H,M,\epsilon'}-\Phi_{\H,M})\otimes \id (O) \rVert_1 
 \\
&\le\max_{\lVert O\rVert_1\leq 1}\onenorm{(\Phi_{\H,M,\epsilon'}-\Phi_{\H,M})\otimes \id (O)}  \\
&=\lVert\Phi_{\H,M,\epsilon'}-\Phi_{\H,M}\rVert_\Diamond\le\epsilon'.
\end{aligned}
\end{equation}

Next we bound the distance $\lVert\Phi_{\calE,0}^{(k)}-\Phi_0^{(k)}\rVert_\Diamond$, where
\begin{equation}\label{eq:unitary_additive_moments}
\begin{aligned}
\Phi_{\calE,0}^{(k)}=&\Phi_{\C}^{(k)}\circ(\Phi_{\H,M}^{(k)}\otimes \id_{\bar{M}})\\
=&\sum_{T_1,T_2,\sigma,\tau}\Wg_{\C}(T_1,T_2)2^{N_M(k-|T_2,\sigma|)}\Wg(\sigma,\tau)|T_1)(\tau|^{\otimes N_M}\otimes |T_1)(T_2|^{\otimes N_{\bar{M}}}.
\end{aligned}
\end{equation}
The diamond norm of each term is given by (Appendix~\ref{app:diamond_norm})
\begin{equation}
\begin{aligned}
&\lVert|T_1)(\tau|^{\otimes N_M}\otimes |T_1)(T_2|^{\otimes N_{\bar{M}}}\rVert_\Diamond=2^{Nk-N_M\dim\rmN_i+N_{\bar{M}}(\dim\rmN_j-\dim\rmN_i)},
\end{aligned}
\end{equation}
which follows from~\eqref{eq:clifford_commutant_diamond_norm_bound} and the fact that the diamond norm is multiplicative under tensor product~\cite{Aharonov:1998zf}. 
We decompose the sum into contributions from four classes of $(T_1,T_2)$ pairs:
\begin{equation}
\begin{aligned}
\rmI&:=\{(T_1,T_2)|\,T_1=T_2\in S_k\} \\
\rmII&:=\{(T_1,T_2)|\,T_2\in S_k, T_1\neq T_2\} \\
\rmIII&:=\{(T_1,T_2)|\, T_1=T_2\notin S_k\}\\
\rmIV&:=\{(T_1,T_2)|\, T_2\notin S_k, T_1\neq T_2\}
\end{aligned}
\end{equation}
Let $\Phi_\rmI^{(k)}$, $\Phi_\rmII^{(k)}$, $\Phi_\rmIII^{(k)}$ and $\Phi_\rmIV^{(k)}$ denote the contributions of the classes of $(T_1,T_2)$ pairs to the sum~\eqref{eq:unitary_additive_moments}. Using the triangle inequality:
\begin{equation}
\begin{aligned}
&\lVert\Phi_{\calE,0}^{(k)}-\Phi_0^{(k)}\rVert_\Diamond \leq \lVert\Phi_\rmI^{(k)}-\Phi_0^{(k)}\rVert_\Diamond+\lVert\Phi_\rmII^{(k)}\rVert_\Diamond+\lVert\Phi_\rmIII^{(k)}\rVert_\Diamond+\lVert\Phi_\rmIV^{(k)}\rVert_\Diamond.
\end{aligned}
\end{equation}
Notice that when $T_2=\pi\in S_k$, we can sum over $\sigma$ to lock $\tau$ to be $\pi$:
\begin{equation}
\sum_{\sigma}2^{N_M(k-|\pi,\sigma|)}\Wg(\sigma,\tau)=\delta_{\pi,\sigma}. 
\end{equation}
Thus,
\begin{equation}
\begin{aligned}
\lVert\Phi_\rmI^{(k)}-\Phi_0^{(k)}\rVert_\Diamond+\lVert\Phi_\rmII^{(k)}\rVert_\Diamond
\leq& \sum_{\pi}|\Wg_{\C}(\pi,\pi)-2^{-Nk}|\!\cdot\!\lVert|\pi)(\pi|^{\otimes N}\rVert_\Diamond+\sum_{T\neq\pi}|\Wg_{\C}(T,\pi)|\!\cdot\!\lVert|T)(\pi|^{\otimes N}\rVert_\Diamond \\
\leq& k!2^{-N+O(k^2)} + k!2^{-N+O(k^2)}\le2^{-N+O(k^2)} .
\end{aligned}
\end{equation}
For the second inequality, we use~\eqref{eq:clifford_weingartens_asymptotics} and~\eqref{eq:clifford_commutant_diamond_norm_bound} to bound the two terms. For the last inequality, we absorb the factor $k!$ into $2^{O(k^2)}$.

In classes III and IV, $T_2$ is not a permutation. In this case, there is an overall suppression because $T_2$ and $\sigma$ are not aligned:
\begin{equation}
\begin{aligned}
\sum_{\sigma,\tau}2^{N_M(k-|T_2,\sigma|)}|\Wg(\sigma,\tau)|&\le2^{N_M(k-1)}\sum_{\sigma,\tau}|\Wg(\sigma,\tau)|\le2^{-N_M}k!(1+k^22^{-N_M}),
\end{aligned}
\end{equation}
where we used $|T_2,\sigma|\geq 1$ and $\sum_{\sigma,\tau}|\Wg(\sigma,\tau)|\leq 2^{-kN_M}k!(1+k^22^{-N_M})$~\cite{Aharonov:2021das}. In class III $T_1=T_2=T \notin S_k$, we use the upper bound on the diamond norm $\lVert|T)(\tau|^{\otimes N_M}\otimes |T)(T|^{\otimes N_{\bar{M}}}\rVert_\Diamond\leq 2^{Nk}$ and obtain
\begin{equation}
\begin{aligned}
\lVert\Phi_\rmIII^{(k)}\rVert_\Diamond\leq&\sum_{T \notin S_k}|\Wg_{\C}(T,T)|\sum_{\sigma,\tau}2^{N_M(k-|T,\sigma|)}|\Wg(\sigma,\tau)|\cdot\lVert|T)(\tau|^{\otimes N_M}\otimes |T)(T|^{\otimes N_{\bar{M}}}\rVert_\Diamond \\
\leq& 2^{-N_M+O(k^2)}k!(1+k^22^{-N_M}).
\end{aligned}
\end{equation}
For $T_1\neq T_2$, the diamond norm can be greater than $2^{Nk}$, but these contributions are suppressed by $\Wg_{\C}(T_1,T_2)$ as $T_1$ and $T_2$ are not aligned. 
Using the third line in Eq.~\eqref{eq:clifford_weingartens_asymptotics}, we obtain
\begin{equation}
\begin{aligned}
&\sum_{T_1\neq T_2}|\Wg_{\C}(T_1,T_2)|\,\dnorm{|T_1)(\tau|^{\otimes N_M}\otimes |T_1)(T_2|^{\otimes N_{\bar{M}}}}\leq 2^{O(k^3)}
\end{aligned}
\end{equation}
This implies
\begin{equation}
\lVert\Phi_\rmIV^{(k)}\rVert_\Diamond\leq 2^{-N_M+O(k^3)}k!(1+k^22^{-N_M}).
\end{equation}

Putting these together, we obtain an upper bound
\begin{equation}
\lVert\Phi_{\calE}^{(k)}-\Phi_0^{(k)}\rVert_\Diamond\leq \epsilon'+2^{-N_M+O(k^3)}(1+k^22^{-N_M}),
\end{equation}
where we have absorbed factors of $k!$ into $2^{O(k^3)}$.

As a final step, we use the triangle inequality to bound the additive error of the $k$-fold channel:
\begin{equation}
\begin{aligned}
\lVert\Phi_{\calE}^{(k)}-\Phi_{\H}^{(k)}\rVert_\Diamond &\le\lVert\Phi_{\calE}^{(k)}-\Phi_0^{(k)}\rVert_\Diamond+\lVert\Phi_0^{(k)}-\Phi_{\H}^{(k)}\rVert_\Diamond \le\epsilon'+2^{O(k^3)-N_M}(1+k^22^{-N_M}).
\end{aligned}
\end{equation}

\end{proof}

\begin{corollary}[Circuit depth]
With $O(k^4\poly\log(k/\epsilon))$ two-qubit magic gates, one can realize approximate unitary $k$-designs with $\epsilon$ additive error in depth $O(N+k\poly\log (k/\epsilon))$ using 1D circuits and in depth $O(\log N+k\poly\log (k/\epsilon))$ using all-to-all circuits with $N^2$ ancillas. 
\end{corollary}
\begin{proof}
A sufficient condition for achieving approximate unitary designs with $\epsilon$ additive error is to take $N_M=O(k^3+\log(1/\epsilon))$. To realize an additive-error unitary design on these qubits, circuits with depth $O(\log (N_M/\epsilon)\cdot k\poly\log k)$ suffice~\cite{Schuster:2024ajb}.
\end{proof}

We note that, in this construction, we place the magic gate at the beginning of the circuit for a technical reason to suppress the diamond norm because $\lVert|T_1)(T_2|\rVert_{\Diamond}$ is smaller for $T_2 \in S_k$.
We do not exclude the possibility of constructing an additive-error unitary design if all magic gates are placed at the final layer.

We further note that the global Clifford unitary in this construction might be necessary.
In a recent paper~\cite{Grevink:2025nez}, it is shown that to form a Clifford $k$-design with bounded additive error, Clifford circuits are required to have $\Omega(N)$ depth in one dimension or $\Omega(\log N)$ depth with all-to-all connectivity. This theorem can be modified to show that shallow-Clifford circuits with the insertion of magic gates acting on a constant number of qubits, even at arbitrary locations within the Clifford circuit, cannot realize a unitary $k$-design with bounded additive error, since the causal light-cone from the magic gates cannot cover a finite fraction of the qubits in the system.

\section{Statistical mechanics picture}\label{sec:stat_mech}
In this section, we provide a physical understanding of several results in this work based on the mapping from the $k$-th moment of the random circuit to classical statistical mechanics models.
Section~\ref{sec:stat_mech_uniformity} develops a stat-mech understanding of why Clifford circuits operating on qubit blocks of $O(\log N)$ qubits can satisfy the approximate uniformity condition.
This also explains why shallow Haar random circuits can form a unitary design with bounded relative error~\cite{Schuster:2024ajb}.
Section~\ref{sec:stat_mech_state_additive} provides an understanding for the required number of magic gates to convert the shallow Clifford circuits into state designs with bounded additive errors.

\subsection{Uniformity condition as strong ordering}\label{sec:stat_mech_uniformity}
Here, we consider the $k$-fold channel of random Clifford circuits and bound its deviation from the uniformity condition using the partition functions of statistical mechanics models.
We show that the two-layer circuit operating on log-size qubit blocks can achieve approximation uniformity as the stat-mech model develops \emph{strong ordering}, i.e. excitations at all scales are suppressed due to large excess free energy.

The $k$-fold channel of two-layer random Clifford circuits can be expressed as
\begin{equation}
\begin{aligned}
\Phi^{(k)}_\calE&=2^{-Nk}\sum_{\vec{T}_1,\vec{T_2}}g(\vec{T}_1,\vec{T}_2)|\vec{T}_1)(\vec{T}_2|=2^{-Nk}\left[\sum_T|T)(T|+\sum_{\vec{T}_1,\vec{T_2}}f(\vec{T}_1,\vec{T}_2)|\vec{T}_1)(\vec{T}_2|\right].
\end{aligned}
\end{equation}
The deviation from the uniformity condition is given by
\begin{align}
\Delta(\Phi^{(k)}_\calE)=\sum_{\vec{T}_1,\vec{T_2}}|f(\vec{T}_1,\vec{T}_2)
=& \sum_{\vec{T}_1,\vec{T_2}}|g(\vec{T}_1,\vec{T}_2)| - \abs{\Sigma_{k,k}} 
+ \sum_{T}\left( 1 - |g(T\vec{1},T\vec{1})| + |g(T\vec{1},T\vec{1}) - 1|\right) \nn\\
\leq& \sum_{\vec{T}_1,\vec{T_2}}|g(\vec{T}_1,\vec{T}_2)| - \abs{\Sigma_{k,k}} + \sum_{T}2|g(T\vec{1},T\vec{1}) - 1|.\label{eq:stat_mech_deviation_uniformity} 
\end{align}

Each term $g(\vec{T}_1,\vec{T}_2)$ in Eq.~\eqref{eq:stat_mech_deviation_uniformity} can be written as a summation over classical ``spins" that label the stochastic Lagrangian subspace in the Clifford commutant $\Sigma_{k,k}$, i.e.
\begin{equation}\label{eq:g_T1_T2}
\begin{aligned}
g(\vec{T}_1,\vec{T}_2)=2^{Nk}&\sum_{\vec{T}_{3,4}} \prod_x \bbrakket{T_{3,x}}{T_{4,x-1}}\bbrakket{T_{3,x}}{T_{4,x}}\Wg_{\C}(T_{1,x},T_{3,x})\Wg_{\C}(T_{4,x},T_{2,x}).
\end{aligned}
\end{equation}
Each spin $T$ originates from the average over random Clifford gates and is arranged as shown below.
\begin{equation*}
\includegraphics[width=15cm]{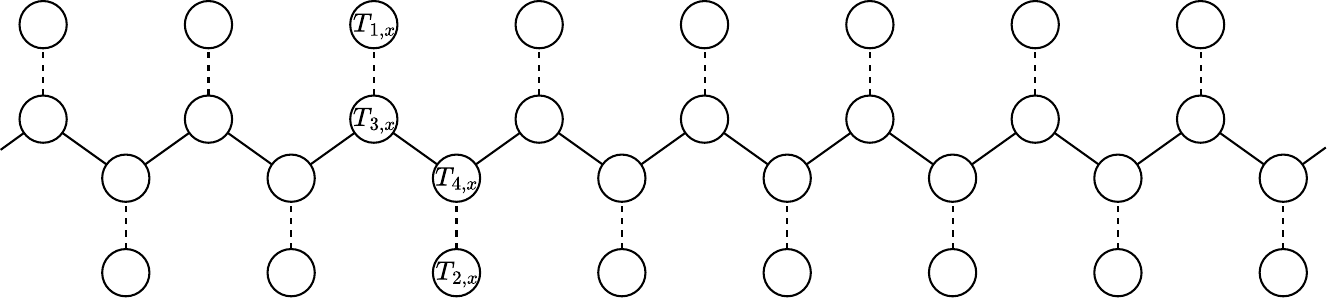}
\end{equation*}
Note that the summation is not yet the partition function of a stat-mech model because the Weingarten function can be negative. 
We introduce a function $g'(\vec{T}_1,\vec{T}_2)$, replacing the Weingarten function in Eq.~\eqref{eq:g_T1_T2} with its absolute value,
\begin{equation}
\begin{aligned}
    g'(\vec{T}_1,\vec{T}_2)=2^{Nk}&\sum_{\vec{T}_{3,4}} \prod_x \bbrakket{T_{3,x}}{T_{4,x-1}}\bbrakket{T_{3,x}}{T_{4,x}}|\Wg_{\C}(T_{1,x},T_{3,x})|\,|\Wg_{\C}(T_{4,x},T_{2,x})|.
\end{aligned}
\end{equation}
The function $g'(\vec{T}_1,\vec{T}_2)$ is an upper bound of $|g(\vec{T}_1, \vec{T}_2)|$ and is given by a summation over positive terms.
Hence, it can be interpreted as a partition function with fixed boundary conditions set by $\vec{T}_1$ and $\vec{T}_2$.

First, we consider the quantity $\calZ:=\sum_{\vec{T}_1,\vec{T_2}}g'(\vec{T}_1,\vec{T}_2)$, which can be viewed as the partition function of a statistical mechanics model of $\vec{T}_{1,2,3,4}$ with free boundary conditions. 
We consider the low-temperature expansion of the partition function.
The configuration with the lowest energy has spins aligned in space and a degeneracy of $|\Sigma_{k,k}|$.
We invoke the Peierls argument to analyze the excitations.
The interaction between two neighboring spins $T_i$ and $T_j$ contributes a factor of $\bbrakket{T_i}{T_j}=2^{\xi(k-|T_i,T_j|)}$ or $|\Wg_{\C}(T_i,T_j)|$ to the Boltzmann weight. 
The Boltzmann weight is suppressed by $2^{-\xi}$ when $T_i$ and $T_j$ are not aligned. 
For $\xi = a\log N$ with a sufficiently large prefactor $a$, the energetic cost of creating a domain wall outweighs the entropic factor.
In this case, the domain walls in the stat-mech model are suppressed, and the partition function is governed by the aligned spin configurations:
\begin{equation}
\calZ=|\Sigma_{k,k}|(1+N2^{-\xi-\log\xi+O(k^2)}).
\end{equation}
To make this more rigorous, one can analytically compute the partition function using the transfer matrix method.

Next, consider the quantity $g'(T\vec{1},T\vec{1})$, which maps to the partition function of the same statistical mechanics model but with the boundary spins fixed to be $T$. 
A similar analysis gives
\begin{equation}
g'(T\vec{1},T\vec{1})=1+N2^{-\xi-\log\xi+O(k^2)}.
\end{equation}
We note that $g(T\vec{1},T\vec{1})$ and $g'(T\vec{1},T\vec{1})$ differ only in the subleading term, and $|g(T\vec{1},T\vec{1})-1| \leq |g'(T\vec{1},T\vec{1})-1| \leq N2^{-\xi-\log\xi+O(k^2)}$.

The above analysis leads to an upper bound on the deviation from the uniformity in Eq.~\eqref{eq:stat_mech_deviation_uniformity},
\begin{align}
\Delta(\Phi^{(k)}_\calE)\leq N2^{-\xi-\log\xi+O(k^2)}.
\end{align}
Hence, one can achieve approximate uniformity with $\xi = O(\log N + k^2)$, i.e. in depth $O(\log N +k^2)$.

We remark that the approximate uniformity requires the stat-mech model to develop \emph{strong ordering}.
In the usual notion of an ordered phase, the domain wall excitations that flip an extensive number of spins are suppressed; however, small domain walls with $O(1)$ energy costs are usually proliferated.
In contrast, the strong ordering here requires the smallest possible domain walls to be suppressed as the small domain walls can cause an $O(1)$ deviation from the uniformity.
In the 1D stat-mech model we obtain, two notions coincide because the domain walls cause the same amount of energy regardless of the number of spins it flips.
However, in higher dimensions, these two notions are generally different. For example, in two dimensions, the spin model can order when each gate in the circuit acts on a constant number of qubits~\cite{Napp:2019wcz,Bao:2021con}; however, strong order still requires each gate to act on $O(\log N)$ qubits.

The stat-mech mapping also explains the results in Ref.~\cite{Schuster:2024ajb} that relative-error unitary design requires two-layer random unitary circuits with each gate drawn from a $k$-design acting on $O(\log N)$ qubits. 
In that case, we can introduce a similar uniformity condition for the $k$-fold channel of Haar random circuits, 
\begin{equation}
    \Phi_\calE^{(k)} = \underbrace{\frac{1}{2^{Nk}} \sum_{\sigma \in S_k} |\sigma)(\sigma|}_{\Phi_0^{(k)}} + \frac{1}{2^{Nk}} \sum_{\vec{\sigma},\vec{\tau}} f(\vec{\sigma},\vec{\tau}) |\vec{\sigma})(\vec{\tau}|.
\end{equation}
Its deviation from the diagonal approximation $\Phi_0^{(k)}$ of the $k$-fold channel of the Haar random unitary, $\Delta_{\H}(\Phi_\calE^{(k)}) = \sum_{\vec{\sigma},\vec{\tau}} \abs{f(\vec{\sigma},\vec{\tau})}$ has a similar statistical mechanics mapping.
Here, the only difference is that spins take values only in the permutation group $S_k$ due to Haar averaging.
We can show that each gate acting on $O(\log N)$ qubits is sufficient for a small deviation.
Crucially, a small deviation here implies that the unitary ensemble forms a design with bounded relative error because for the infinite norm $\infnormsm{\mathrm{r}(\sigma)} = 1$ for a permutation $\sigma \in S_k$.
However, for the Clifford circuits, $\infnorm{\mathrm{r}(T)} = 2^{\dim \rmN}$ can take a large value; the uniformity condition does not imply Clifford $k$-design with bounded relative error.

We further remark that the gluing lemma essentially implies strong ordering.
It requires the spins to align while introducing one additional spin to the stat-mech model; the energetic cost of misaligned spins overcomes the entropic factor.

\subsection{Additive-error state designs: magic as symmetry breaking field}\label{sec:stat_mech_state_additive}

The additive error of state $k$-th moment is controlled by the frame potential. 
Here, we map the frame potential of random Clifford circuits with magic insertions to the partition function of a classical statistical mechanics model.
The mapping provides physical understandings of the depth of the Clifford circuit and the number $\tilde{O}(k^2)$ of magic gates required for an approximate state $k$-design.  
Similar stat-mech mapping has been developed for random tensor networks~\cite{Hayden:2016cfa,Vasseur:2018gfy}, Haar random circuits~\cite{Nahum:2017yvy,vonKeyserlingk:2017dyr}, and also more specifically random Clifford circuits~\cite{Li:2021dbh,Zhang:2024fyp}.

First, we consider the state ensemble $\calE$ generated by two-layer random Clifford circuits with open boundary conditions. 
The frame potential provides an upper bound for the trace distance between the $k$-th moment of $\calE$ and that of the ensemble $\C$ of random stabilizer states
\begin{equation}\label{eq:stat_mech_1norm_bound_from_frame_potential}
\lVert\rho^{(k)}_{\calE}-\rho_{\C}^{(k)}\rVert_1^2\leq 2^{Nk}\lVert\rho^{(k)}_{\calE}-\rho_{\C}^{(k)}\rVert_2^2\leq 2^{Nk}[F^{(k)}_{\calE}-F^{(k)}_{\C}],
\end{equation}
where
\begin{equation}\begin{aligned}
F^{(k)}_{\C}&=\Ens_{V, V'\sim\C} \abs{\bra{0}^{\otimes N} V'^\dagger V\ket{0}^{\otimes N}}^{2k}
=\bra{0}^{\otimes Nk}\rho_{\C}^{(k)}\ket{0}^{\otimes Nk}\\
&=Z_{N,k}^{-1}\sum_{T\in\Sigma_{k,k}}\bra{0}^{\otimes Nk}r(T)\ket{0}^{\otimes Nk}=Z_{N,k}^{-1}|\Sigma_{k,k}|=2^{-Nk}|\Sigma_{k,k}|(1+O(k2^{k-N})).
\end{aligned}\end{equation}
Here, we use $\bra{0}^{\otimes Nk}r(T)\ket{0}^{\otimes Nk} = 1$ as $(0, 0, \cdots, 0)\in T$ for any $T$.
The frame potential of $\calE$ is given by
\begin{align}
F_\calE^{(k)} &= \Ens_{V_1, V_2\sim\calE} \abs{\bra{0}^{\otimes N} V_1^\dagger V_2\ket{0}^{\otimes N}}^{2k} =\bra{0}^{\otimes Nk}\rho_{\tilde{\calE}}^{(k)}\ket{0}^{\otimes Nk},
\end{align}
where $\tilde{\calE} = \{V_1^\dagger V_2\ket{0}\}$ for $V_1, V_2 \in \calE$ contains states generated by a three-layer random Clifford circuit $V_1^\dagger V_2$.
Expressing the average $k$-th moment of each Clifford gate in terms of operators in the Clifford commutant [Eq.~\eqref{eq:clifford_commutant_weingarten_def}], we have
\begin{equation}\begin{aligned}
F_{\calE}^{(k)}=&Z_{2\xi,k}^{-2L}\sum_{T_{1,x}\sim T_{4,x}}\bbrakket{T_{1,1}}{T_{4,1}}\bbrakket{T_{1,L}}{T_{4,L}}\left[\prod_{x=1}^{L-1}\bbrakket{T_{1,x}}{T_{2,x}}\Wg_{\C}(T_{2,x},T_{3,x})\bbrakket{T_{3,x}}{T_{4,x}}\right] \\
&\qquad\qquad\bbrakket{T_{1,L}}{T_{2,L-1}}\bbrakket{T_{3,L-1}}{T_4},
\end{aligned}\end{equation}
where $L=N/(2\xi)$. 
We note that the left-hand side of Eq.~\eqref{eq:stat_mech_1norm_bound_from_frame_potential} is positive, and therefore $F^{(k)}_{\calE}\geq F^{(k)}_{\C}$.

The frame potential $F_{\calE}^{(k)}$ is expressed as a summation over classical variables $T_x$.
At this moment, the summation involves terms with negative weights, as the Clifford Weingarten function can be negative.
Instead, we can consider an upper bound $F_\calE^{(k)} \leq F'^{(k)}_{\calE}$,
\begin{equation}\begin{aligned}
F'^{(k)}_{\calE}=&Z_{2\xi,k}^{-2L}\sum_{T_{1,x}\sim T_{4,x}}\bbrakket{T_{1,1}}{T_{4,1}}\bbrakket{T_{1,L}}{T_{4,L}}\left[\prod_{x=1}^{L-1}\bbrakket{T_{1,x}}{T_{2,x}}\abs{\Wg_{\C}(T_{2,x},T_{3,x})}\bbrakket{T_{3,x}}{T_{4,x}}\right] \\
&\qquad\qquad\bbrakket{T_{1,L}}{T_{2,L-1}}\bbrakket{T_{3,L-1}}{T_4},
\end{aligned}\end{equation}
which is a sum over positive terms and can be viewed as the partition function of a one-dimensional classical statistical mechanics model (illustrated below).
\begin{equation*}
\includegraphics[width=14cm]{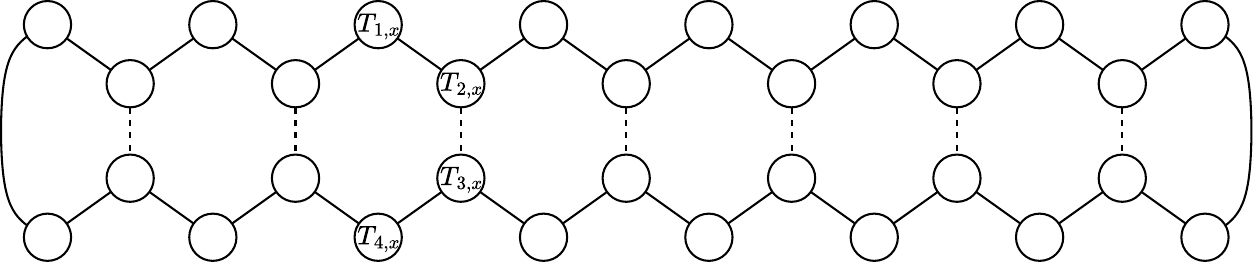}
\end{equation*}
The interaction between two neighboring spins $T_i$ and $T_j$ contributes a factor of $\bbrakket{T_i}{T_j}=2^{\xi(k-|T_i,T_j|)}$ (represented by a solid line) or $|\Wg_{\C}(T_i,T_j)|$ to the Boltzmann weight (represented by a dashed line).
The Boltzmann weight is suppressed by $2^{-\xi}$ when $T_i$ and $T_j$ are not aligned.
For $\xi = a\log N$ with a sufficiently large prefactor $a$, the energetic cost of creating a domain wall outweighs the entropic factor.
In this case, the domain walls in the stat-mech model are suppressed, and the frame potential is given by the aligned spin configurations, leading to a value close to $F_{\C}^{(k)}$.

Next, we consider inserting single-qubit magic gates in the circuit. 
We insert the gates in the leftmost qubit block for the simplicity of presentation. 
We comment on inserting magic gates at other locations later on in this section.
The frame potential $F_\calE^{(k)}$ is then given by a slightly modified partition function below.
\begin{equation*}
\includegraphics[width=14cm]{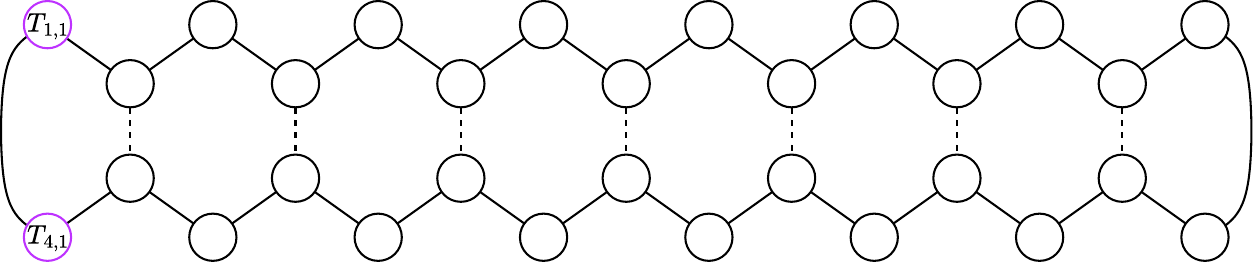}
\end{equation*}
Averaging over single-qubit random gates drawn from either an exact $k$-design or $\calE_u = \{\mathds{1},u,u^\dagger\}$ with $u$ being a magic gate creates a suppressed Boltzmann weight (i.e. an energy cost) when the spins on the leftmost boundary $T_{1,1}$ or $T_{4,1}$ above are not permutations.
The weight associated with $T_{1,1}$ and $T_{4,1}$ then acquires an upper bound,
\begin{equation}
\begin{aligned}
&\bbra{T_{1,1}}\bigotimes_{i=1}^{N_M} (\hat{\Phi}_{u})_{i}\kket{T_{1,4}}=2^{\xi(k-|T_{1,1},T_{4,1}|)},\,\text{if }T_1,T_2\in S_k, \\
&|\bbra{T_{1,1}}\bigotimes_{i=1}^{N_M} (\hat{\Phi}_{u})_{i}\kket{T_{1,4}}|\leq 2^{k\xi}a^{\frac{N_M}{2}},\, \text{if }T_{1,1} \text{ or } T_{4,1}\notin S_k,
\end{aligned}
\end{equation}
where $(\hat{\Phi}_{u})_{i}=\Ens_{U_i\sim\calE_u}U_i^{\otimes k}\otimes U_i^{*\otimes k}$ is the $k$-fold operator associated with the random gate set on the $i$-th qubit, and $0< a < 1$ depends on the gate set (see Eqs.~\eqref{eq:overlap_bound_magic} and~\eqref{eq:overlap_bound_k-design}).
The suppressed weight can be regarded as a magnetic field acting on $T_{1,1}$ and $T_{1,4}$.
The field strength is $O(1)$ for $k$-design single-qubit gates and is $O(1/\log^2k)$ for single-qubit gates drawn from $\{\mathds{1},u,u^\dagger\}$.

The ``magnetic field" favors permutation $T \in S_k$ over other stochastic Lagrangian subspace $T \in \Sigma_{k,k}$.
The energy cost for $T \notin S_k$ is proportional to $N_M$, while the number of $T \in \Sigma_{k,k}$ but $T \notin S_k$ is $2^{O(k^2)}$~\footnote{Here, we consider $\xi > k - 1$ such that operators $r(T)$ in the commutant are linearly independent. Hence, the number of internal degrees of freedom for each spin is $2^{O(k^2)}$.}.
Hence, $N_M = \tilde{O}(k^2)$ is sufficient to overcome the entropic factor.
In this case, the frame potential $F'^{(k)}_\calE$ is dominated by the lowest energy configurations, in which the spins are aligned spatially and are given by the same permutation.
This leads to a frame potential $F'^{(k)}_\calE \approx k!/2^{Nk}$ close to the Haar random value.

We note that inserting Haar random gates at other locations in the circuit modifies the stat-mech model in a similar way by introducing magnetic fields at the corresponding locations.
Without a rigorous proof, we believe that adding $\tilde{O}(k^2)$ single-qubit magic gates (drawn from either $\{\mathds{1},u,u^\dagger\}$ or an exact $k$-design) at random locations in the log-depth Clifford circuit is sufficient to realize a state design with bounded additive error.

\section{Discussion}\label{sec:discussion}
In this work, we have presented constructions of approximate state and unitary designs, which improve the previous results in terms of either the circuit depth or the number of non-Clifford gates.
We prove no-go theorems for various architectures to form designs with bounded relative errors.
We also provide physical understanding for several key results based on the statistical mechanics mapping.

Our study leaves several open questions:
\begin{itemize}
\item The scaling of $k$ in the circuit depth of our constructions for relative error designs may not be optimal. 
Our constructions require each Haar random gate to act on a qubit cluster of size $O(k\log k)$ because we use a rather loose bound on the distance between operators in the Clifford commutant, $|\sigma, T| \geq 1$ for $\sigma \in S_k$ and $T \notin S_k$.
Suppose one can obtain a bound on the number of permutations of a certain distance from stochastic Lagrangian subspaces $T \notin S_k$. 
In that case, we might be able to reduce the size of the Haar random gates, which would improve the scaling of $k$ in the total circuit depth.

\item Our design constructions use gates acting on qubit blocks. It remains open whether the circuit ensemble still satisfies the approximate uniformity if one replaces two-layer random circuits with a brickwork random circuit of depth $O(\log N)$ operating on qubits.

\item In our construction for additive-error state design, the magic gates are not required to be drawn from unitary $k$-designs and can instead be drawn from a less random gate set $\{\mathds{1},u,u^\dagger\}$ with $u$ being a magic gate (e.g. T-gate). 
A natural question is, in our other constructions of approximate designs, whether one can replace the $k$-design random unitaries with a specific set of magic gates, which might be easier to realize in practice.

\item Different from the unitary group, orthogonal $k$-designs, which approximate the $k$-th moment of Haar random orthogonal matrices, cannot be generated in sub-linear depth~\cite{Schuster:2024ajb}. Reducing the magic depth or the usage of magic gates in the construction is of particular interest in this case. It remains open whether one can use linear-depth Clifford circuits augmented by constant-depth magic gates to construct such designs.
\end{itemize}

\section*{Acknowledgments}
We are especially grateful to Thomas Schuster and Nicholas Hunter-Jones for many insightful discussions. 
We thank Thomas Schuster and Yunchao Liu for helpful feedback on our manuscript.
We also thank Matthew P. A. Fisher, Haimeng Zhao, Lorenzo Piroli, Xiaoliang Qi, and Shao-Kai Jian for discussions and comments. 

\emph{Note added.---} After the completion of this work, we became aware of an independent work that achieves $\epsilon$-approximate $k$-designs over $N$ qubits using a different circuit architecture of depth $O(\log k \cdot \log\log Nk/\epsilon)$~\cite{CuiSchuster2025}.
The scaling of circuit depth with $N$ is similar to ours, while the scaling with $k$ is different. 
Our circuit depth for additive-error state designs saturates the depth lower bound proved in their paper, and is therefore optimal.

\appendix

\section{Operational meaning of additive-error state designs}\label{app:operation_state_additive_error}

In this appendix, we consider an adaptive experiment that queries a random state $\ket{\psi}$ drawn from an ensemble for $k$ times.
If the ensemble forms a state $k$-design with bounded additive error, the output state would have a bounded trace distance from the output of the experiment that queries genuine Haar random states.

Any adaptive experiment that queries the state $k$ times can be achieved by querying the states in parallel.
\begin{equation*}
\includegraphics[width=12cm]{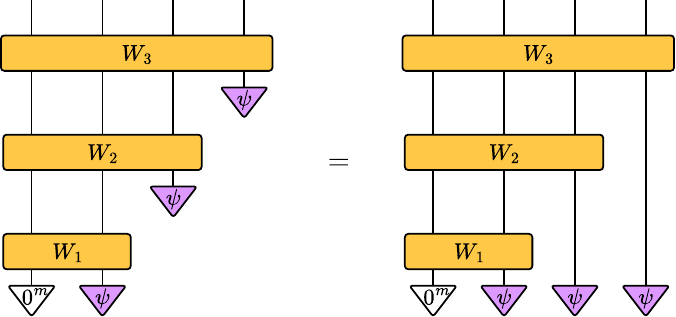}
\end{equation*}
The most general adaptive experiment that queries $k$ copies of the state $\ket{\psi}$ is shown on the left. 
Let $m$ denote the number of ancillas used in the experiments, and these ancillas are initialized in $\ket{0}^{\otimes m}$. 
In the first step, we apply a unitary $W_1$ acting on the ancillas and a copy of $\ket{\psi}$. 
In the $t$-th step, we introduce one copy of $\ket{\psi}$ and apply a unitary $W_t$ that couple $\ket{\psi}$ to all the degrees of freedom that we have introduced previously.
In the end, the output state of the experiment is $\ket{\psi_{\text{out}}}=W_k\cdots W_1\ket{0}^{\otimes m}\ket{\psi}^{\otimes k}$, which can be viewed as applying a unitary to the initial state $\ket{\psi}^{\otimes k}$, as shown on the right. 
If the state $\ket{\psi}$ is drawn from an ensemble $\calE$ that forms an approximate state $k$-design with $\epsilon$ additive error, the trace distance between the output state and the output of an experiment that queries Haar random states is
\begin{equation}
\begin{aligned}
&\left\lVert\Ens_{\psi\sim\calE}\ketbra{\psi_{\text{out}}}-\Ens_{\psi\sim\H}\ketbra{\psi_{\text{out}}}\right\rVert_1 =\lVert\rho_{\calE}^{(k)}-\rho_{\H}^{(k)}\rVert_1\le\epsilon.
\end{aligned}
\end{equation}
The first equality is because the trace norm is invariant under unitaries.

\section{Additive-error state designs with magic gates in the first layer}\label{sec:state_additive_initial_magic}

In this appendix, we provide a construction of the state $k$-design with bounded additive errors by inserting magic gates prior to the two-layer Clifford unitary.

\begin{corollary}[Additive-error state designs from constant magic followed by low-depth Clifford unitaries]
\label{thm:state_additive_initial_magic}
Consider a system of $N$ qubits. Let ${M_i}$ denote disjoint subsystems, each containing $N_i$ qubits. Assume $N_i\geq 2\log k$. Let $N_M=\sum_i N_i$ be the total size of these subsystems. The unitaries $U_i$, acting on subsystem $M_i$, are independently drawn from exact unitary $k$-designs. Let $V$ be a two-layer brickwork circuit, where each gate is drawn form an exact Clifford $k$-design. Let $\xi$ be the smallest size of the overlapping regions of the Clifford unitaries. Then the states $V\left(\bigotimes_i U_i\right)\ket{0}^N$ form an approximate state $k$-design with additive error 
\begin{equation}
    \epsilon\leq 2^{-\xi-\log\xi+O(k^2)}+ 2^{O(k^2)}\prod_i2^{-N_i}k!(1+k^22^{-N_i}).
\end{equation} 
In other words, one can achieve $\varepsilon$ additive error with $N_M= O\big(k^2+\log(1/\epsilon)\big)$ and $\xi=O(\log(N/\epsilon)+k^2)$.
\begin{equation*}
\includegraphics[width=8.2cm]{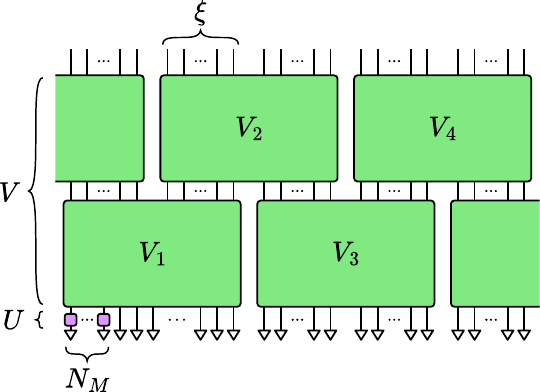}
\end{equation*}
\end{corollary}
\begin{proof}
Before applying the Clifford unitary, the state on $M=\cup_i M_i$is a tensor product of Haar random states on $M_i$, whose $k$-th moment is
\begin{equation}
\rho_{\calE_0}^{(k)}=\left[\sum_{\sigma_i,\tau_i}\bigotimes_i\Wg(\sigma_i,\tau_i)r(\sigma_i)_{M_i}\right]\otimes \ket{0}\bra{0}^{\otimes N-N_M}.
\end{equation}
According to Lemma~\ref{thm:uniformity_log_depth}, the Clifford unitaries $V$ satisfies the $\epsilon'$-approximate uniformity condition with $\epsilon'=2^{-\xi-\log\xi+O(k^2)}$. The $k$-th moment of these Clifford unitaries can be written as
\begin{equation}
\Phi_{\calE_u}^{(k)}=2^{-Nk}\left[\sum_T|T)(T|+\sum_{\vec{T}_1,\vec{T_2}}f(\vec{T}_1,\vec{T}_2)|\vec{T}_1)(\vec{T}_2|\right].
\end{equation}
Applying it to the initial state, the state $k$-th moment becomes
\begin{equation}
\begin{aligned}
&\rho_{\calE}^{(k)}=2^{-Nk}\left[\sum_T r(T)\prod_i\sum_{\sigma_i,\tau_i}2^{N_i(k-|T,\sigma_i|)}\Wg(\sigma_i,\tau_i)\right. \\
&\qquad\qquad\left.+\sum_{\vec{T}_1,\vec{T}_2}f(\vec{T}_1,\vec{T}_2) r(\vec{T}_1)\prod_i\sum_{\sigma_i,\tau_i}\Wg(\sigma_i,\tau_i)2^{N_i(k-|T_{2,i},\sigma_i|)}\right],
\end{aligned}
\end{equation}
where $i$ is the index for subregions $M_i$, and $T_{2,i}$ is the value of $\vec{T}_2$ on region $M_i$. 
We decompose $\rho_{\calE}^{(k)}$ into three parts $\rho_{\calE}^{(k)}=\rho_{\rmI}^{(k)}+\rho_{\rmII}^{(k)}+\rho_{\rmIII}^{(k)}$, where
\begin{equation}
\begin{aligned}
\rho_{\rmI}^{(k)}=&2^{-Nk}\sum_{\pi}r(\pi)=2^{-Nk}D_{N,k}\rho_{\H}^{(k)},\\
\rho_{\rmII}^{(k)}=&2^{-Nk}\sum_{T\notin S_k} r(T)\prod_i\sum_{\sigma_i,\tau_i}2^{N_i(k-|T,\sigma_i|)}\Wg(\sigma_i,\tau_i), \\
\rho_{\rmIII}^{(k)}=&2^{-Nk}\sum_{\vec{T}_1,\vec{T}_2}f(\vec{T}_1,\vec{T}_2) r(\vec{T}_1)\prod_i\sum_{\sigma_i,\tau_i}\Wg(\sigma_i,\tau_i)2^{N_i(k-|T_{2,i},\sigma_i|)}.
\end{aligned}
\end{equation}
Using the triangle inequality, the additive error is bounded by
\begin{equation}
\lVert\rho_\calE^{(k)}-\rho_{\H}^{(k)}\rVert_1
\le\lVert\rho_{\rmI}^{(k)}-\rho_{\H}^{(k)}\rVert_1+\lVert\rho_{\rmII}^{(k)}\rVert_1+\lVert\rho_{\rmIII}^{(k)}\rVert_1.
\end{equation}
Each term on the right has an upper bound:
\begin{equation}
\lVert\rho_{\rmI}^{(k)}-\rho_{\H}^{(k)}\rVert_1=|2^{-Nk}D_{N,k}-1|\cdot\lVert\rho_{\H}^{(k)}\rVert_1\leq k^2 2^{-N}.
\end{equation}
\begin{equation}
\begin{aligned}
\lVert\rho_{\rmII}^{(k)}\rVert_1
&\leq 2^{-Nk} \sum_{T\notin S_k} \lVert r(T)\rVert_1\prod_i\sum_{\sigma_i,\tau_i}2^{N_i(k-|T,\sigma_i|)}|\Wg(\sigma_i,\tau_i)| \\
&\leq 2^{O(k^2)}\prod_i\sum_{\sigma_i,\tau_i}2^{N_i(k-1)}|\Wg(\sigma_i,\tau_i)|  \\
&\leq 2^{O(k^2)} \prod_i2^{-N_i}k!(1+k^22^{-N_i}).
\end{aligned}
\end{equation}
Here, we use $\lVert r(T)\rVert_1\leq 2^{Nk}$ and $\sum_{\sigma_i,\tau_i}|\Wg(\sigma_i,\tau_i)|\leq k!(1+k^22^{-N_i})$ for $k^2\leq 2^{N_i}$.
\begin{equation}\label{eq:state_additive_initial_magic_third_term}
\begin{aligned}
\lVert\rho_{\rmIII}^{(k)}\rVert_1\leq 2^{-Nk} \sum_{\vec{T}_1,\vec{T}_2}&|f(\vec{T}_1,\vec{T}_2)|\cdot \lVert r(\vec{T}_1)\rVert_1
\prod_i\Big|\sum_{\sigma_i,\tau_i}2^{N_i(k-|T_{2,i},\sigma_i|)}\Wg(\sigma_i,\tau_i)\Big|.
\end{aligned}
\end{equation}
If $T_{2,i}\in S_k$, then $\sum_{\sigma_i,\tau_i}\Wg(\sigma_i,\tau_i)2^{N_i(k-|T_{2,i},\sigma_i|)}=1$. If $T_{2,i}\notin S_k$, then
\begin{equation}
\begin{aligned}
&\Big|\sum_{\sigma_i,\tau_i}2^{N_i(k-|T_{2,i},\sigma_i|)}\Wg(\sigma_i,\tau_i)\Big|
\le \sum_{\sigma_i,\tau_i}2^{N_i(k-1)}|\Wg(\sigma_i,\tau_i)|
\le2^{-N_i}k!(1+k^22^{-N_i})\le1.
\end{aligned}
\end{equation}
The last inequality is from the assumption $N_i\geq 2k\log k$. Plugging it back to~\eqref{eq:state_additive_initial_magic_third_term} and using $\lVert r(\vec{T}_1)\rVert_1\leq 2^{Nk}$, we have $\lVert\rho_{\rmIII}^{(k)}\rVert_1\le\sum_{\vec{T}_1,\vec{T}_2}|f(\vec{T}_1,\vec{T}_2)| \le\epsilon'=2^{-\xi-\log\xi+O(k^2)} $.

Putting these together, the additive error is bounded by
\begin{equation}
\begin{aligned}
\lVert\rho_\calE^{(k)}-\rho_{\H}^{(k)}\rVert_1\le&2^{O(k^2)} \prod_i2^{-N_i}k!(1+k^22^{-N_i})+2^{-\xi-\log\xi+O(k^2)}.
\end{aligned}
\end{equation}
\end{proof}

In this construction of an additive-error state $k$-design, each state in the ensemble can be described by a Clifford augmented matrix product state.
The possibility of such states to realize additive error designs has been studied in Ref.~\cite{lami2025quantum,fux2024disentangling}.
Compared to the existing results, our construction has a reduced depth of the Clifford gates.

Our construction also suggests that an additive-error state $k$-design with $k=o(N^{1/2})$ does not require maximum magic. 
This improves the observation in Ref.~\cite{Liu:2020yso} that states with sub-maximum magic can realize additive-error state $k$-designs for $k = \tilde{o}(N^{1/4})$, which is based on the result of Ref.~\cite{Haferkamp:2020qel}.

\section{Collision probabilities in log-depth Clifford circuits}\label{app:collision_prob}

In this appendix, we study the $k$-th collision probability in an ensemble of Clifford unitaries. 
We show that this quantity obtained in two-layer random Clifford circuits approximates the one for global random Clifford gates with a bounded relative error.
Our result is consistent with the fact that anti-concentration happens at a logarithmic depth shown in Ref.~\cite{Barak:2020svy,Dalzell:2020vxs,Magni:2025zqf,aharonov2023polynomial}.

For a unitary ensemble $\calE= \{V\}$, the $k$-th collision probability for a given computation basis $\ket{\boldx}$ takes the form
\begin{equation}
p_{\calE}^{(k)}(\boldx):=\Ens_{V\sim\calE}|\bra{\boldx}V\ket{0}^{\otimes N}|^{2k}.
\end{equation}
For global random Clifford gates, this quantity is given by
\begin{equation}
p_{\C}^{(k)}:=\Ens_{V\sim\C}|\bra{\boldx}V\ket{0}^{\otimes N}|^{2k}=Z_{N,k}^{-1}|\Sigma_{k,k}|.
\end{equation}

\begin{corollary}[Relative error for the collision probability in log-depth Clifford circuits]
Consider an ensemble $\calE= \{V\}$ of two-layer Clifford circuits with each gate drawn independently from an exact Clifford $k$-design. 
Let $\xi$ be the size of the smallest overlapping regions in the two-layer Clifford circuit. 
The two-layer Clifford circuit achieves a relative error $\epsilon=N2^{-\xi-\log\xi+O(k^2)}$ in the $k$-th collision probability with respect to global random Clifford unitaries, i.e.
\begin{equation}
(1-\epsilon)p_{\C}^{(k)}\leq p_{\calE}^{(k)}(\boldx)\leq (1+\epsilon)p_{\C}^{(k)}\quad \text{for}\ \forall\boldx.
\end{equation}
In other words, one can achieve $\epsilon$ relative error in depth $O(\log(N/\epsilon)+k^2)$.
\end{corollary}
\begin{proof}
The $k$-fold channel of two-layer random Clifford circuits satisfies the approximate uniformity according to Theorem~\ref{thm:uniformity_log_depth}. Hence, it takes the form
\begin{equation}
\Phi_\calE^{(k)}=2^{-Nk}\left[\sum_T|T)(T|+\sum_{\vec{T}_1,\vec{T_2}}f(\vec{T}_1,\vec{T}_2)|\vec{T}_1)(\vec{T}_2|\right]
\end{equation}

The $k$-th collision probability is given by
\begin{equation}
\begin{aligned}
\Ens_{U\sim\calE}|\bra{\boldx}U\ket{0}^{\otimes N}|^{2k}&=\Tr[\ketbra{\boldx}^{\otimes k}\Phi_\calE^{(k)}(\ketbra{0}^{\otimes kN})] \\
&=2^{-Nk}\Bigg[\sum_T \Tr[\ketbra{\boldx}^{\otimes k}r(T)]\Tr[r(T)^{\top}\ketbra{0}^{\otimes kN}] \\
&+\sum_{\vec{T}_1,\vec{T_2}}f(\vec{T}_1,\vec{T}_2)\Tr[\ketbra{\boldx}^{\otimes k}r(\vec{T}_1)]\Tr[r(\vec{T}_2)^{\top}\ketbra{0}^{\otimes kN}]\Bigg] \\
&= 2^{-Nk}\Bigg[|\Sigma_{k,k}|+\sum_{\vec{T}_1,\vec{T_2}}f(\vec{T}_1,\vec{T}_2)\Bigg].
\end{aligned}
\end{equation}
Here, we note that $\Tr[\ketbra{\boldx}^{\otimes k} r(T)] = \Tr[\ketbra{0}^{\otimes kN} r(T)] = 1$ because any computational basis state $\ket{\boldx}$ is related to $\ket{0}^{\otimes N}$ by a Clifford gate, whose $k$-th tensor power commutes with $r(T)$ in the commutant.

The relative error therefore has an upper bound
\begin{equation}
\begin{aligned}
\frac{\big|p_{\calE}^{(k)}(\boldx)-p_{\C}^{(k)}\big|}{p_{\C}^{(k)}}
\le& |2^{-Nk}Z_{N,k}-1|+2^{-Nk}Z_{N,k}\sum_{\vec{T}_1,\vec{T_2}}|f(\vec{T}_1,\vec{T}_2)| \\
\le& k2^{k-N}+N2^{-\xi-\log\xi+O(k^2)}.
\end{aligned}
\end{equation}
\end{proof}

\section{Growing additive-error state designs with random Clifford gates}\label{app:gluing_lemma_haar_clifford}
We prove that a state $k$ design and a stabilizer state $k$ design can be glued by Clifford unitaries drawn from a Clifford $k$ design to form a state $k$ design over a larger region with bounded additive error.

\begin{lemma}[Gluing a state design and a stabilizer state design using unitaries from Clifford design]\label{thm:state_additive_haar_gluing}
Let $A$, $B$, $C$, and $D$ be disjoint subsystems of qubits. $N_A$, $N_B$, $N_C$, $N_D$ denote the size of the subsystems. Assume $N_C\geq k+\log k-2$. $\ket{\psi}_{AB}$ is drawn from a state $k$-design up to additive error $\varepsilon_{AB}$, $\ket{\psi}_{CD}$ is drawn from an exact stabilizer $k$-design. $V_{BC}$ is drawn from a (exact) Clifford $k$-design. Then $V_{BC}\ket{\psi}_{AB}\otimes \ket{\psi}_{CD}$ is an approximate state $k$-design with additive error $\varepsilon\leq \varepsilon_{AB} +2^{O(k^2)}(2^{-N_B}+2^{-N_C})$.
\begin{align*}
\includegraphics[width=14cm]{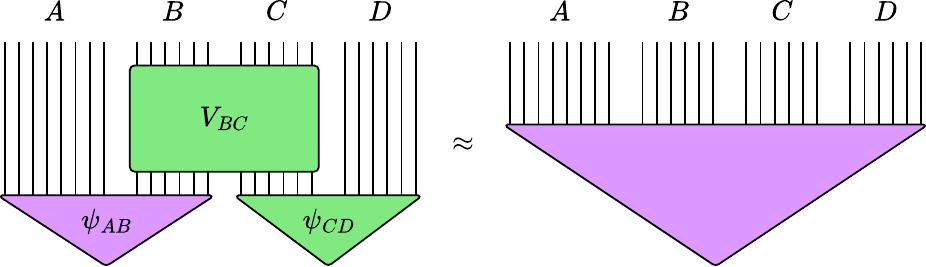}
\end{align*}
\end{lemma}

\begin{proof}
We write the resulting state as $\rho_{\H,AB,\varepsilon}^{(k)}\otimes\rho_{\C,CD}^{(k)}$. It can be approximated up to additive error $\varepsilon$ by
\begin{equation}
\rho_{\H,AB}^{(k)}\otimes\rho_{\C,CD}^{(k)}=D_{N_{AB},k}^{-1}Z_{N_{CD},k}^{-1}\sum_{\sigma,T}{r(\sigma)}_{AB}\otimes r(T)_{CD}
\end{equation}
where $\rho_{\H,AB}^{(k)}$ is the $k$-th moment of an exact $k$-design on $AB$.
After applying the random Clifford unitaries, the $k$-th moment of the state becomes
\begin{equation}
\begin{aligned}
&\Phi_{\C,BC}^{(k)}(\rho_{\H,AB}^{(k)}\otimes\rho_{\C,CD}^{(k)})\\
&\quad=\frac{2^{N_{BC}k}}{D_{N_{AB},k}Z_{N_{CD},k}}\sum_{\sigma,T_{1,2,3}}\Wg_{\C}(T_1,T_2)2^{-N_B|T_2,\sigma|-N_C|T_2,T_3|}r(\sigma)_A\otimes r(T_1)_{BC}\otimes r(T_3)_D
\end{aligned}
\end{equation}
\begin{equation*}
\includegraphics[width=4cm]{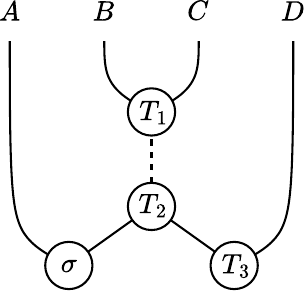}  
\end{equation*}

The trace distance from the $k$-th moment of Haar random states takes the form
\begin{equation}
\begin{aligned}
&\lVert \Phi_{\C,BC}^{(k)}(\rho_{\H,AB}^{(k)}\otimes\rho_{\C,CD}^{(k)})-\rho_{\H,ABCD}^{(k)}\lVert_1 \\
&\leq 2^{kN}\sum_{\sigma,T_{1,3}}\Bigg|\frac{\sum_{T_2}\Wg_{\C}(T_1,T_2)2^{N_B(k-|T_2,\sigma|)+N_C(k-|T_2,T_3|)}}{D_{N_{AB},k}Z_{N_{CD},k}}
-\frac{\delta_{T_1,\sigma}\delta_{T_3,\sigma}}{D_{N,k}}\Bigg|
\end{aligned}
\end{equation}

where we use $\lVert r(\sigma)_A\otimes r(T_1)_{BC}\otimes r(T_3)_D\rVert_1\leq 2^{kN}$, $N=N_A+N_B+N_C+N_D$. 
For $T_3=\sigma$, the summation over $T_2$ yields a delta function $\delta_{T_1,\sigma}$. The corresponding contribution is
\begin{equation}
\begin{aligned}
2^{kN}\sum_{\sigma \in S_k}\left|D_{N_{AB},k}^{-1}Z_{N_{CD},k}^{-1}-D_{N,k}^{-1}\right|
\leq & k!\big[(1+k^22^{-N_{AB}})(1+k2^{k-N_{CD}})-1+k^22^{-N}\big] \\
\leq &k!(k^2 2^{-N_{AB}} + k 2^{k-N_{CD}} + (k^3+k^2) 2^{-N}) %
\end{aligned}
\end{equation}
For $T_3\notin S_k$, the contribution 
\begin{equation}
\begin{aligned}
& \frac{2^{kN}}{D_{N_{AB},k}Z_{N_{CD},k}}\big|\Wg_{\C}(T_1,T_2)\big|2^{N_B(k-|T_2,\sigma|)+N_C(k-|T_2,T_3|)}
\end{aligned}
\end{equation}
is either suppressed by $2^{-N_B}$ or $2^{-N_C}$ due to misalignment of inside $2^{N_B(k-|T_2,\sigma|)+N_C(k-|T_2,T_3|)}$, or suppressed by misalignment in $T_1$ and $T_2$ to give $|\Wg_{\C}(T_1,T_2)|\leq 2^{-N_{BC} k-N_{BC}+O(k^2)}$. The number of such terms are upper bounded by $2^{O(k^2)}$. When we sum those up, we at most get $2^{O(k^2)}(2^{-N_B}+2^{-N_C})$. To sum up, we have
\begin{equation}
\begin{aligned}
&\lVert \Phi_{\C,BC}^{(k)}(\rho_{\H,AB}^{(k)}\otimes\rho_{\C,CD}^{(k)})-\rho_{\H,ABCD}^{(k)}\lVert_1 
\leq 2^{O(k^2)}(2^{-N_B}+2^{-N_C}).
\end{aligned}
\end{equation}
We can thus use the triangle inequality to upper-bound the additive error
\begin{equation}
\begin{aligned}
&\lVert \Phi_{\C,BC}^{(k)}(\rho_{\H,AB,\varepsilon}^{(k)}\otimes\rho_{\C,CD}^{(k)})-\rho_{\H,ABCD}^{(k)}\lVert_1 \\
&\leq \lVert \Phi_{\C,BC}^{(k)}(\rho_{\H,AB,\varepsilon}^{(k)}\otimes\rho_{\C,CD}^{(k)})-\Phi_{\C,BC}^{(k)}(\rho_{\H,AB}^{(k)}\otimes\rho_{\C,CD}^{(k)})\lVert_1  +\lVert \Phi_{\C,BC}^{(k)}(\rho_{\H,AB}^{(k)}\otimes\rho_{\C,CD}^{(k)})-\rho_{\H,ABCD}^{(k)}\lVert_1 \\
&\leq\lVert\rho_{\H,AB,\varepsilon}^{(k)}\otimes\rho_{\C,CD}^{(k)}-\rho_{\H,AB}^{(k)}\otimes\rho_{\C,CD}^{(k)}\lVert_1  +\lVert \Phi_{\C,BC}^{(k)}(\rho_{\H,AB}^{(k)}\otimes\rho_{\C,CD}^{(k)})-\rho_{\H,ABCD}^{(k)}\lVert_1 \\
&\leq \varepsilon+2^{O(k^2)}(2^{-N_B}+2^{-N_C}).
\end{aligned}
\end{equation}
\end{proof}

\section{More on the Clifford commutant}\label{app:clifford_commutant}

\subsection{Asymptotics of Clifford Weingarten functions}
\label{sec:clifford_commutant_wg_asymptotics}
We present the asymptotics of the Weingarten functions of the Clifford group over $n$ qubits.
The results are based on properties of an orthogonal basis analyzed in~\cite{Haferkamp:2020qel}. 

To begin, the $k$-fold channel $\Phi^{(k)}_{\C}$ can be formulated as an operator in the doubled Hilbert space via the Choi–Jamiołkowski isomorphism
\begin{equation}
\hat{\Phi}_C^{(k)}:=\Ens_{V\sim\C}V^{\otimes k}\otimes V^{*\otimes k}=\sum_{T_i,T_j}\Wg_{\C}(T_i,T_j)\kket{T_i}^{\otimes n}\bbra{T_j}^{\otimes n}.
\end{equation}
Here, the states $\{\kket{T_i}\}$, vectorization of the operators $r(T)$, are linearly independent but not orthogonal, with inner products given by $\bbrakket{T_i}{T_j}=2^{n(k-|T_i,T_j|)}$. 
One can construct an orthogonal basis using the Gram-Schmidt orthogonalization\footnote{The disadvantage of this basis is that it is not a tensor product of operators on single qubits. Therefore, we work in the $\{\kket{T_i}\}$ basis to proceed in the main text.}:
\begin{gather}
\kket{E_i}=\sum_{j=1}^iA_{ij}\kket{T_j}^{\otimes n}, \qquad
A_{ij}=\sum_{\Pi\in S_i,\Pi(i)=j}\sgn(\Pi)\prod_{l=1}^{j-1}\bbrakket{T_l}{T_{\Pi(l)}}^n.
\end{gather}
It has the following properties~\cite{Haferkamp:2020qel}
\begin{equation}
\begin{aligned}
&|A_{ii}-2^{-nk}|\leq 2^{-nk-n+O(k^2)}), \, \forall i\\
&|A_{ij}|\leq 2^{-nk-n+O(k^2)},\, \forall i\neq j\\
&|A_{ij}|\leq 2^{-n(k+|\dim\rmN_i-\dim\rmN_j|)+O(k^3)}, \, \forall i\neq j.
\end{aligned}
\end{equation}
The normalization is
\begin{equation}
|\bbrakket{E_i}{E_i}-2^{nk}|\leq 2^{nk-2n+O(k^2)}.
\end{equation}
Since $\hat{\Phi}_{\C}^{(k)}$ is a projection operator, it is an equal weight sum of an orthonormal basis:
\begin{equation}
\hat{\Phi}_{\C}^{(k)}=\sum_j\frac{1}{\bbrakket{E_i}{E_i}}\kket{E_i}\bbra{E_i}.
\end{equation}
So the Weingarten functions can be expressed as
\begin{equation}
\Wg_{\C}(T_i,T_j)=\sum_k\frac{1}{\bbrakket{E_k}{E_k}}A_{ki}A_{kj}^*.
\end{equation}
We have the following asymptotics
\begin{equation}\label{eq:clifford_weingartens_asymptotics}
\begin{aligned}
&|\Wg_{\C}(T_i,T_i)-2^{-nk}|\le2^{-nk-n+O(k^2)},\, \forall i\\
&|\Wg_{\C}(T_i,T_j)|\leq 2^{-nk-n+O(k^2)},\, \forall i\neq j \\
&|\Wg_{\C}(T_i,T_j)|\leq 2^{-n(k+|\dim\rmN_i-\dim\rmN_j|)+O(k^3)},\, \forall i\neq j
\end{aligned}
\end{equation}
For $i\neq j$, we used the triangle inequality $|\dim\rmN_i-\dim\rmN_k|+|\dim\rmN_k-\dim\rmN_j|\ge|\dim\rmN_i-\dim\rmN_j|$.

\subsection{Diamond norm}\label{app:diamond_norm}
To bound the additive error of approximate unitary designs, we use the diamond norm for the superoperator $|T_1)(T_2|(\cdot) = \mathrm{r}(T_1)\Tr[\mathrm{r}(T_2)^\top (\cdot)]$ acting on a single qubit:
\begin{equation}
\lVert |T_1)(T_2|\rVert_\Diamond=2^{k+\dim\rmN_2-\dim\rmN_1}.\label{eq:clifford_commutant_diamond_norm_bound}
\end{equation}
To prove this, we first use an upper bound of the diamond norm (Lemma 3 of~\cite{Haferkamp:2020qel}):
\begin{equation}
\begin{aligned}
\lVert|T_1)(T_2|\rVert_\Diamond&=\max_{\lVert O\rVert_1\leq 1}\lVert|T_1)(T_2|\otimes \id(O)\rVert_1 \\
&=\max_{\lVert O\rVert_1\leq 1}\lVert \mathrm{r}(T_1)\otimes\Tr_1[(\mathrm{r}(T_2)^\top\otimes I) O]\rVert_1 \\
&=\lVert \mathrm{r}(T_1)\rVert_1\max_{\lVert O\rVert_1\leq 1}\lVert \Tr_1[(\mathrm{r}(T_2)^\top\otimes I) O]\rVert_1 \\
&\le\lVert \mathrm{r}(T_1)\rVert_1\max_{\lVert O\rVert_1\leq 1}\lVert (\mathrm{r}(T_2)^\top\otimes I) O\rVert_1 \\
&=\lVert \mathrm{r}(T_1)\rVert_1\lVert \mathrm{r}(T_2)\otimes I\rVert_\infty=\lVert \mathrm{r}(T_1)\rVert_1 \lVert \mathrm{r}(T_2)\rVert_\infty\\
&=2^{k+\dim\rmN_2-\dim\rmN_1}.
\end{aligned}
\end{equation}
The notation $\Tr_1$ is the partial trace over the first copy of the system. 
The second line is because the trace norm is multiplicative under the tensor product. 
The third line is because the trace norm is non-increasing under the partial trace. 
The other direction is given by
\begin{equation}
\begin{aligned}
\lVert|T_1)(T_2|\rVert_\Diamond&\geq\frac{\lVert|T_1)(T_2|(\mathrm{r}(T_2))\rVert_1}{\lVert\mathrm{r}(T_{2})\rVert_1}=\frac{\lVert 2^k\mathrm{r}(T_1)\rVert_1}{\lVert\mathrm{r}(T_{2})\rVert_1}=2^{k+\dim\rmN_2-\dim\rmN_1}.
\end{aligned}
\end{equation}

\section{Stabilizer Rényi entropy}\label{sec:stabilizer_renyi_entropy}

The stabilizer Rényi entropies are defined as~\cite{Leone:2021rzd}
\begin{equation}
\begin{aligned}
\calM_{k}&=-\log\Bigg(2^{-N}\sum_{\boldu,\boldv\in \mathbb{F}_2^N}|\bra{\psi}P_{\boldu,\boldv}\ket{\psi}|^{k}\Bigg)=-\log\Tr\left(2^{N}P^{(k)}\ket{\psi}\bra{\psi}^{\otimes k}\right),
\end{aligned}
\end{equation}
where $P_{\boldu,\boldv}=i^{\boldu\cdot\boldv}Z^{\boldu}X^{\boldv}=i^{\boldu\cdot\boldv}\bigotimes_{i=1}^NZ^{u_i}X^{v_i}$, representing all Pauli strings up to global phases. We take $k$ to be an integer multiple of $4$.
The Pauli moment $P^{(k)}$ factorizes as a tensor product of single-qubit Pauli moments. The single-qubit Pauli moment is
\begin{equation}
\begin{aligned}
P^{(k)}&=4^{-1}\sum_{m,n}(i^{mn}Z^mX^n)^{\otimes k}=4^{-1}\sum_{m,n}Z^{m\boldone_{k}}X^{n\boldone_{k}},\qquad
\boldone_{k}=(\underbrace{1,\cdots,1}_{k})\in \mathbb{Z}_2^{k}.    
\end{aligned}
\end{equation}
This is a member of the Clifford commutant corresponding to $r(T_S)=2P_{\CSS(\{\boldone_{k}\})}$, where $P_{\CSS(\{\boldone_{k}\})} = P^{(k)}$ is the projection onto the CSS codespace associated with the defect subspace $\{\boldone_{k}\}$ that contains two elements, $\boldzero_{k}$ and $\boldone_{k}$.
In the case of $k=4$, this is the only CSS code projector in the Clifford commutant~\cite{Gross_2021}, playing an important role in the representation theory of the Clifford group~\cite{Zhu:2016gfm}.

\section{List of notation}\label{sec:list_of_notations}
\begin{itemize}
\item $\Phi_{\calE}^{(k)}$, $\Phi_{\H}^{(k)}$, $\Phi_{\C}^{(k)}$: $k$-fold channels of an unitary ensemble $\calE$, the ensemble of Haar random unitaries, and the ensemble of random Clifford unitaries.
\item $\varrho_{\calE}^{(k)}$, $\varrho_{\H}^{(k)}$: $k$-fold channels of a unitary ensemble $\calE$ and the ensemble of Haar random unitaries acting on EPR pairs.
\item $\varrho_0^{(k)}$: an approximation of $\varrho_{\H}^{(k)}$.
\item $\rho_{\calE}^{(k)}$, $\rho_{\H}^{(k)}$, $\rho_{\C}^{(k)}$: $k$-th moments of a state ensemble $\calE$, the ensemble of Haar random states, and the ensemble random stabilizer states.
\item $\mathrm{r}(T)$: a member of the Clifford commutant on a single qubit.
\item $r(T)_S$: a member of the Clifford commutant on a multi-qubit system $S$.
\item $r(\vec{T})$: $\bigotimes_x r(T_x)$.
\item $\kket{T}$: vectorization of $r(T)$.
\item $\bbrakket{T_1}{T_2}$: $\Tr[r(T_1)^{\top}r(T_2)]$.
\item $|T_1)(T_2|$: a superoperator $|T_1)(T_2|(\cdot)=r(T_1)\Tr[r(T_2)^\top(\cdot)]$.
\item $|\vec{T}_1)(\vec{T}_2|$: $\bigotimes_x |T_{1,x})(T_{2,x}|$.
\item $\sigma$: a permutation that belong to $S_k$.
\item $T\notin S_k$: a stochastic Lagrangian subspace that is not associated with permutations.
\item $r(\sigma)$: a member of the Clifford commutant that is a permutation operator.
\item $\Wg(\sigma,\tau)$: Weingarten functions for Haar random unitaries.
\item $\Wg_{\C}(T_1,T_2)$: Weingarten functions for random Clifford unitaries.
\item $|T_1,T_2|$: distance between $T_1$ and $T_2$ defined via $\Tr[\mathrm{r}(T_1)^\top \mathrm{r}(T_2)]=2^{k-|T_1,T_2|}$
\item $D_{N,k}$: $\prod_{i=0}^{k-1}(2^N+i)$, appears in $\rho_{\H}^{(k)}$.
\item $Z_{N,k}$: $2^N\prod_{i=0}^{k-2}(2^N+2^i)$, appears in $\rho_{\C}^{(k)}$.
\item $\log$: logarithmic function with base $2$.
\item $\lVert\cdot\rVert_p$: Schatten $p$-norm.

\end{itemize}

\bibliography{reference.bib}

\end{document}